\documentclass[journal]{IEEEtran}
\usepackage{amsmath,amsfonts}
\usepackage{algorithmic}
\usepackage{algorithm}
\usepackage{array}
\usepackage{textcomp}
\usepackage{url}
\usepackage{verbatim}
\usepackage{graphicx}
\usepackage{xcolor}
\usepackage{amssymb}
\usepackage{mathtools}
\usepackage{array}
\usepackage{mathrsfs}
\usepackage{enumerate}
\usepackage{setspace}
\usepackage{subfigure}
\usepackage{epsfig}
\usepackage{cite}
\usepackage{color}
\usepackage{makecell}
\usepackage{tablefootnote}
 \usepackage{multirow}
 \usepackage{tabularx}
\usepackage[justification=centering]{caption}
\def\BibTeX{{\rm B\kern-.05em{\sc i\kern-.025em b}\kern-.08em
		T\kern-.1667em\lower.7ex\hbox{E}\kern-.125emX}}
\begin{document}
\newtheorem{Asump}{\bf Assumption}
\newtheorem{lemma}{\bf Lemma}
\newtheorem{theorem}{\bf Theorem}
\newtheorem{proof}{Proof}
\setlength{\belowcaptionskip}{-8pt}
\setlength{\abovedisplayskip}{1pt}
\setlength{\belowdisplayskip}{1pt}
\title{AMP-based Joint Activity Detection and Channel Estimation for Massive Grant-Free Access in OFDM-based Wideband Systems\\
}
\author{Zhiyan Li,~\IEEEmembership{Graduate student member,~IEEE,} Ying Cui,~\IEEEmembership{Member,~IEEE,} Danny H.K. Tsang,~\IEEEmembership{Life Fellow,~IEEE}
	
	\thanks{The authors are with the IoT Thrust, Hong Kong University of Science and Technology (Guangzhou), Guangzhou 511400, China. This paper was presented in part at IEEE GLOBECOM Workshop 2024 \cite{zhiyanGlobecom}.
	}

}

\maketitle

\begin{abstract}
To realize orthogonal frequency division multiplexing (OFDM)-based grant-free access for wideband systems under frequency-selective fading, existing device activity detection and channel estimation methods need substantial accuracy improvement or computation time reduction. In this paper, we aim to resolve this issue. First, we present an exact time-domain signal model for OFDM-based grant-free access under frequency-selective fading. 
Then, we present a maximum a posteriori (MAP)-based device activity detection problem and two minimum mean square error (MMSE)-based channel estimation problems. The MAP-based device activity detection problem and \textcolor{black}{one of the} MMSE-based channel estimation \textcolor{black}{problems} are formulated for the first time.
Next, we build a new factor graph that captures the exact statistics of time-domain channels and device activities. Based on it, we propose two approximate message passing (AMP)-based algorithms, \emph{AMP-A-EC} and \emph{AMP-A-AC}, \textcolor{black}{to approximately solve} the MAP-based device activity detection problem and two MMSE-based channel estimation problems. Both proposed algorithms alleviate the AMP's inherent convergence problem when the pilot length is smaller or comparable to the number of active devices. Then, we analyze \emph{AMP-A-EC}'s error probability of activity detection and mean square error (MSE) of channel estimation via state evolution and show that \emph{AMP-A-AC} has the lower computational complexity (in dominant term). Finally, numerical results show the two proposed AMP-based algorithms' superior performance and respective preferable regions, revealing their significant values for OFDM-based grant-free access.
\end{abstract}
\begin{IEEEkeywords}
 mMTC, grant-free access, OFDM, activity detection, channel estimation, and AMP.
\end{IEEEkeywords}

\section{Introduction}

As the global IoT device count continues to \textcolor{black}{increase}, massive machine-type communication (mMTC) has emerged as one of the primary use cases for both 5G and beyond 5G (B5G) wireless networks \cite{Erik,Chawla19TSP,Dey20TWC}. However, supporting many energy-limited and intermittently active IoT devices transmitting small packets \textcolor{black}{remain challenging}. Grant-free access has been proposed as a \textcolor{black}{promising} solution, with device activity detection and channel estimation being critical components.

Existing work investigates activity detection and channel estimation mainly for narrowband systems under flat fading \cite{Caire18ISIT,Yu19ICC,zhongGLOBECOM,Yu21SPAWC,Jiang21TWC,Chen18TSP,Liu18TSP,Senel18TCOM,Qiu,Shao19IoTJ,QiuSPL,Zhang23TSP,JSAC_li}.
Most of them adopt statistical estimation methods \textcolor{black}{that} rely on channel, noise, and activity statistics, as illustrated in Table~\ref{tab:table_method}. Specifically, statistical estimation methods include classical statistical estimation methods (which assume parameters of interest to be deterministic but unknown), e.g., maximum likelihood estimation (MLE) and Bayesian estimation methods (which view the parameters of interest as realizations of random variables with known prior distributions), e.g., maximum a posterior estimation (MAPE) and minimum mean square error (MMSE).
First, MLE \cite{Caire18ISIT,Yu19ICC,zhongGLOBECOM,Yu21SPAWC} and MAPE \cite{Jiang21TWC} are applied to device activity detection, and
the block coordinate descent (BCD) method is proposed to obtain stationary points of the MLE \cite{Caire18ISIT,Yu19ICC,zhongGLOBECOM,Yu21SPAWC} and MAPE \cite{Jiang21TWC} problems for device activities. The resulting BCD-based algorithms have convergence guarantees and achieve high detection accuracies, but with long computation times due to their sequential updates in each iteration. Given very accurate device activity detection results, channel estimations of active devices are readily obtained via classic channel estimation methods. 
Next, MMSE \cite{Chen18TSP,Liu18TSP,Senel18TCOM,Shao19IoTJ,Qiu,QiuSPL,Zhang23TSP} is adopted to estimate the effective channel of each device, which is the channel if the device is active and zero otherwise. 
The approximate message passing (AMP) \cite{Chen18TSP,Liu18TSP,Senel18TCOM,Shao19IoTJ} and its variants \cite{Qiu,QiuSPL,Zhang23TSP} are adopted to approximately solve the MMSE problems. The resulting AMP-based algorithms have short computation times thanks to the low-complexity parallel update mechanism. However, they have more restrictive requirements of system parameters than MLE and MAPE-based ones, and the estimation errors may not monotonically decrease over iterations when the pilot length is smaller or comparable to the number of active devices. Based on the estimated effective channel, the device activity can be obtained using thresholding. Compared to MLE-based and MAPE-based methods, MMSE-based methods can achieve short computation times at the sacrifice of detection accuracy due to the inherent approximation mechanism for enforcing computational complexity reduction in solving optimization problems. \textcolor{black}{Some existing work formulates} effective channel estimation problems as regularized norm approximation problems, called GROUP LASSO \cite{JSAC_li,ADMM,OMPold}, assuming small noise and sparse device activities without using any statistics. They employ optimization algorithms, e.g., alternating direction method of multipliers (ADMM) \cite{ADMM,JSAC_li}, or greedy algorithms, e.g., orthogonal matching pursuit (OMP) \cite{OMPold}, to solve the GROUP LASSO problems. The resulting accuracy is typically worse than those of MLE, MAPE, and MMSE-based methods due to the lack of statistical information.

Notably, the abovementioned state-of-the-art methods for narrowband systems with flat fading are no longer suitable for wideband systems under frequency-selective fading. Orthogonal frequency division multiplexing (OFDM) is widely adopted for wideband systems to cope with channel frequency selectivity. Activity detection and channel estimation for OFDM-based grant-free access for wideband systems under frequency-selective fading has drawn increasing attention \cite{OMP,JiaTWC, JiangMP_TWC,GaoAMP_OFDM_TWC} and is the focus of this paper. 
For instance, \cite{JiaTWC} adopts an exact time-domain signal model under the assumption that the pilot length and the number of subcarriers are identical;
\cite{JiangMP_TWC, GaoAMP_OFDM_TWC,OMP} adopt simplified frequency-domain signal models which completely ignore the dependence caused by the shared device activities and the underlying time-domain frequency-selective fading channels \cite{GaoAMP_OFDM_TWC} or approximate capture the dependence caused by the underlying time-domain frequency-selective fading channels \cite{JiangMP_TWC,OMP}. Specifically, \cite{JiaTWC} proposes MLE and MAPE-based device activity detection methods via BCD; \cite{GaoAMP_OFDM_TWC} proposes MMSE-based device activity detection and channel estimation methods via AMP; \cite{JiangMP_TWC} proposes MMSE-based device activity detection and channel estimation methods via message passing (MP) method's variants, which have higher computational complexity than AMP;
\cite{OMP} proposes a GROUP-LASSO-based activity detection and channel estimation method via OMP for a single antenna base station (BS). The accuracies and computation times of the MLE, MAPE, MMSE, and GROUP-LASSO-based methods for OFDM-based wideband systems are consistent with their correspondences for narrowband systems.


\begin{table}[t] 
	\caption{Activity Detection and Channel Estimation. \label{tab:table_method}}
	\centering
	\scriptsize
	\begin{tabular}{|c|c|c|c|}
		\hline
		\textbf{\makecell{Application\\cases}} &\textbf{\makecell{Estimation\\ problems}} & \textbf{Methods} & \textbf{\makecell{Prior \\information}}\\
		\hline
		\multirow{2}{*}{\makecell{Activity \\ detection}}&MLE  & {\makecell{BCD\\ \cite{Caire18ISIT,Yu19ICC,zhongGLOBECOM,Yu21SPAWC}\cite{JiaTWC}}} & Gaussian noise \\
		\cline{2-4} 
		&	MAPE  & {\makecell{BCD\\ \cite{Jiang21TWC}\cite{JiaTWC}}} & \makecell{Gaussian noise \\ Bernoulli activities} \\
		\hline
		\multirow{3}{*}{ \makecell[c]{Activity\\ detection and \\channel\\ estimation}}&MMSE  & \makecell{AMP \\ \cite{Chen18TSP,Liu18TSP,Senel18TCOM,Shao19IoTJ,Qiu,QiuSPL,Zhang23TSP} \cite{GaoAMP_OFDM_TWC}} & \makecell{Gaussian noise \\ Bernoulli activities \\ Gaussian channels} \\ 
				\cline{2-4}
		&\multirow{2}{*}{\makecell{GROUP \\LASSO}}  & ADMM  \cite{JSAC_li}\cite{ADMM} &  \multirow{2}{*}{\makecell{Small noise \\ Sparse effective channels}} \\
		\cline{3-3}
		&	& OMP \cite{OMPold} \cite{OMP} &  \\
		\hline
	\end{tabular}
\end{table}

Four primary limitations exist in the existing work for OFDM-based grant-free access under frequency-selective fading. Firstly, the existing AMP-based algorithms \cite{Chen18TSP,Liu18TSP,Senel18TCOM,Qiu,Shao19IoTJ,QiuSPL,Zhang23TSP, GaoAMP_OFDM_TWC} return the estimates of the effective channels that may not be the best among all iterates when the pilot length is smaller or comparable to the number of active devices. Secondly, the AMP-based algorithm \cite{GaoAMP_OFDM_TWC}, MP-based algorithm \cite{JiangMP_TWC}, and OMP-based algorithm \cite{OMP} incur additional accuracy loss partially due to the approximations on effective channels. Thirdly, the BCD-based algorithm \cite{JiaTWC} and MP-based algorithm \cite{JiangMP_TWC} are computationally expensive for practical systems. 
Finally, the performance analysis of activity detection and channel estimation remains unresolved. 

In this paper, we would like to address the above limitations. Specifically, we investigate joint device activity detection and channel estimation for OFDM-based grant-free access in a wideband system under frequency-selective fading. The main contributions are summarized as follows. 

\noindent\textbullet\ We extend the exact time-domain signal model in \cite{JiaTWC} to incorporate more flexible \textcolor{black}{configurations} for the pilot length and number of subcarriers and precisely characterize the relationship among frequency-domain channels. Based on this model, we present a MAP-based device activity detection problem and two MMSE-based channel estimation problems for estimating effective channels and active devices' channels, respectively. Notably, neither the MAP-based device activity detection problem nor the MMSE-based channel estimation problem for active devices has ever been considered.

\noindent\textbullet\  We build a new factor graph that captures the exact statistics of time-domain channels and device activities. Based on this factor graph, we present the MP method and its approximation, which lay a foundation for solving the MAP-based device activity detection problem and the two MMSE-based channel estimation problems.

\noindent\textbullet\  We develop two AMP-based algorithms, \emph{AMP-A-EC} and \emph{AMP-A-AC}, relying on the approximations of the MP method with different parameter simplification operations. \emph{AMP-A-EC} approximately solves the MAP-based device activity detection problem and MMSE-based effective channel estimation problem. In contrast, \emph{AMP-A-AC} approximately solves the MAP-based device activity detection problem and MMSE-based actual channel estimation problem for active devices. Both proposed algorithms keep tracking the best estimates over iterations, alleviating the AMP method's inherent convergence problem when the pilot length is smaller or comparable to the number of active devices \cite{Chen18TSP,Liu18TSP,Senel18TCOM,Qiu,Shao19IoTJ,QiuSPL,Zhang23TSP, GaoAMP_OFDM_TWC}.

\noindent\textbullet\  We analyze \emph{AMP-A-EC}'s error probability of activity detection and mean square error (MSE) of channel estimation using the classic state evolution (SE) technique. \emph{AMP-A-AC} does not follow SE, making the analysis intractable, but exhibits lower computational complexity (in dominant term) than \emph{AMP-A-EC}.

\noindent\textbullet\  Numerical results show that \emph{AMP-A-EC} and \emph{AMP-A-AC} can reduce the error probability and MSE by up to $94\%$ and $33\%$, respectively, compared to the existing AMP-based algorithm in \cite{GaoAMP_OFDM_TWC} for an OFDM-based wideband system, the extension of the AMP-based algorithms in \cite{Chen18TSP,Liu18TSP,Senel18TCOM} to OFDM-based wideband systems, and the extension of OMP in \cite{OMP} to a multi-antenna BS, and reduce the computation time by up to $96\%$, compared to the MLE-based method in \cite{JiaTWC}. \textcolor{black}{The} numerical results also validate the performance analysis of \emph{AMP-A-EC}, and the respective preferable regions of \emph{AMP-A-EC} (preferable at short pilot lengths and small numbers of antennas) and \emph{AMP-A-AC} (preferable at long pilot lengths and large numbers of antennas). The gains of the proposed algorithms reveal their significant practical value in OFDM-based massive grant-free access.


{\emph{Notation:}} Use bold uppercase (e.g., $\mathbf{X}$) and bold lowercase (e.g., $\mathbf{a}$) to denote matrices and vectors, respectively. $\mathbf{X}_{i,m:n}$ denotes a row vector consisting of the elements from columns $m$ to $n$ in the $i$-th row of matrix $\mathbf{X}$, $\mathbf{X}_{i:j,m}$ denotes a column vector consisting of the elements from rows $i$ to $j$ in the $m$-th column, and $\mathbf{X}_{i:j,m:n}$ denotes a block containing the elements from rows $i$ to $j$ and columns $m$ to $n$. $\mathbf{X}_{\overline{n,p,m}}$ represents all elements in $\mathbf{X}$ except the element indexed by $(n,p,m)$. $\mathbf{I}$ denotes the unit matrix. Use  uppercase letters \textcolor{black}{in calligraphy} (e.g., $\mathcal{N}$) to denote sets. The conjugate and conjugate transpose operators are denoted by $(\cdot)^*$ and $(\cdot)^H$, respectively. Use $f_{\mathcal{CN}}(x;a,b)$ to denote the complex Gaussian probability density function (PDF) with mean $a$ and variance $b$. 
Use $f_{\mathcal{B}}(x;c)$ to denote the probability mass function (PMF) of the Bernoulli distribution with parameter $c$.
Use $\delta(\cdot)$ to denote the Delta function.
Use $f_{\mathcal{B-CN}}(x;a,b,c) \triangleq c f_{\mathcal{CN}}(x;a,b) + (1-c)\delta(x)$ to denote the Bernoulli-Gaussian distribution. Let $\|\mu_1-\mu_2\|_K \triangleq \sup_{a \in \mathbb{R}} |\int_{-\infty}^a \left(\mu_1(x)-\mu_2(x)\right)  dx |$ denote the Kolmogorov distance between distributions $\mu_1$ and $\mu_2$.
We denote the convergence in  \textcolor{black}{the} Kolmogorov distance as $\xrightarrow{\|\cdot\|_K}$.
\section{System model}
We study a single-cell cellular network that consists of one $M$-antenna BS and $N$ single-antenna IoT devices and operates on a wide band. Let $\mathcal M\triangleq \{1,2,\cdots,M\}$ and $\mathcal N\triangleq \{1,2,\cdots,N\}$.
The wireless channel is characterized by large-scale fading and small-scale fading. For all $n \in \mathcal{N}$, let $\beta_n>0$ denote the large-scale fading power of the channel between device $n$ and the BS. Suppose that $\beta_n,n\in \mathcal{N}$ are known to the BS \cite{Liu18TSP,JiaTWC}. We consider block fading and frequency-selective fading models for small-scale fading with $P$ channel taps. Denote $\mathcal P\triangleq \{1,2,\cdots,P\}$. For all $n\in\mathcal{N}$, $p\in\mathcal{P}$, and $m \in \mathcal{M}$, let $g_{n,p,m}$ denote the small-scale fading coefficient of the $p$-th tap of the channel between device $n$ and the BS's antenna $m$. Suppose that the small-scale fading coefficients $g_{n,p,m}, n\in\mathcal{N}, p \in \mathcal{P}, m \in\mathcal{M}$ are unknown to the BS and follow i.i.d. $\mathcal{CN}(0,1)$. 
For all $n\in\mathcal{N}$, $p\in\mathcal{P}$, and $m \in \mathcal{M}$, the overall wireless channel between device $n$ and the BS's antenna $m$ can be represented as $h_{n,p,m} = \sqrt{\beta_n}g_{n,p,m} \in \mathbb{C}$.

We investigate the massive access scenario stemming from mMTC. In particular, only a few devices activate and access the BS within each coherence block. For all $n\in\mathcal{N}$, $a_n\in\{0,1\}$ denotes the activity state of device $n$, where $a_n=1$ indicates that device $n$ is active and $a_n=0$ otherwise. In the  scenario \textcolor{black}{considered}, $\sum_{n\in\mathcal{N}} a_n =N_a \ll N$, i.e., $\boldsymbol{a}\triangleq (a_n)_{n\in\mathcal N} \in \{0,1\}^N$ is sparse. As in \cite{Liu18TSP}, suppose that $a_n,n\in\mathcal{N}$ are i.i.d. Bernoulli random variables, each with probability $\rho_n \in (0,1)$. For ease of exposition, suppose that $\rho_n,n\in\mathcal{N}$ are known to the BS. The probability mass function (PMF) of $\boldsymbol{a}$ is given by:
\begin{align} \label{eq:PMF_a}
	p(\boldsymbol{a}) = \prod_{n \in \mathcal{N}} p(a_n) = \prod_{n \in \mathcal{N}}\rho_n^{a_n}(1-\rho_n)^{1-a_n}.
\end{align} 
For all $n\in\mathcal{N}, p\in\mathcal{P}$, and $m \in \mathcal{M}$, let $x_{n,p,m}=a_n h_{n,p,m}$ represent the $p$-th effective channel between device $n$ and the BS's antenna $m$. By (\ref{eq:PMF_a}) and $h_{n,p,m} \sim \mathcal{CN}(0,\beta_n)$, we have: 
\begin{align} \label{eq:prior_xa}
	p(x_{n,p,m}|a_n) = f_{\mathcal{B-CN}}(x_{n,p,m}; 0,\beta_n,a_n).
\end{align}
For all $n\in N$, $x_{n,p,m}=a_n h_{n,p,m},p\in P,m \in M$ are unconditionally dependent but independently conditioned on $a_n$. \textcolor{black}{Noting} that $a_n$ and $h_{n,p,m}$ are unknown, we will estimate active devices' channel conditions from the estimate of $x_{n,p,m},p\in P,m \in M$, which will be illustrated shortly.

As illustrated in Fig. \ref{fig_pilot}, we consider an OFDM-based massive grant-free access scheme to prompt efficient uplink transmission over the wideband system in an arrive-and-go manner \cite{JiaTWC}. Denote $K$ as the number of subcarriers, and $\mathcal K\triangleq \{1,2,\cdots,K\}$. Each device $n\in\mathcal{N}$ is allocated a unique length-$L$ pilot sequence. Suppose $L = KQ \ll N$, for some $Q \in \mathbb{Z}_{++}$. Denote $\mathcal L\triangleq \{1,2,\cdots,L\}$ and $\mathcal Q\triangleq \{1,2,\cdots,Q\}$. 
For all $n \in \mathcal{N}$, device $n$ inserts its $L$ pilot symbols\footnote{When $L$ increases, more time-frequency communication resources are consumed in the pilot transmission phase, and the accuracies of activity detection and channel estimation increase.} into the $K$ subcarriers, forming $Q$ OFDM pilot symbols, denoted by $\tilde{\mathbf s}_{q,n} \in \mathbb C^{K \times 1},q\in\mathcal{Q}$. 
\begin{figure}[!t]
	\centering
	\includegraphics[width=4.5cm]{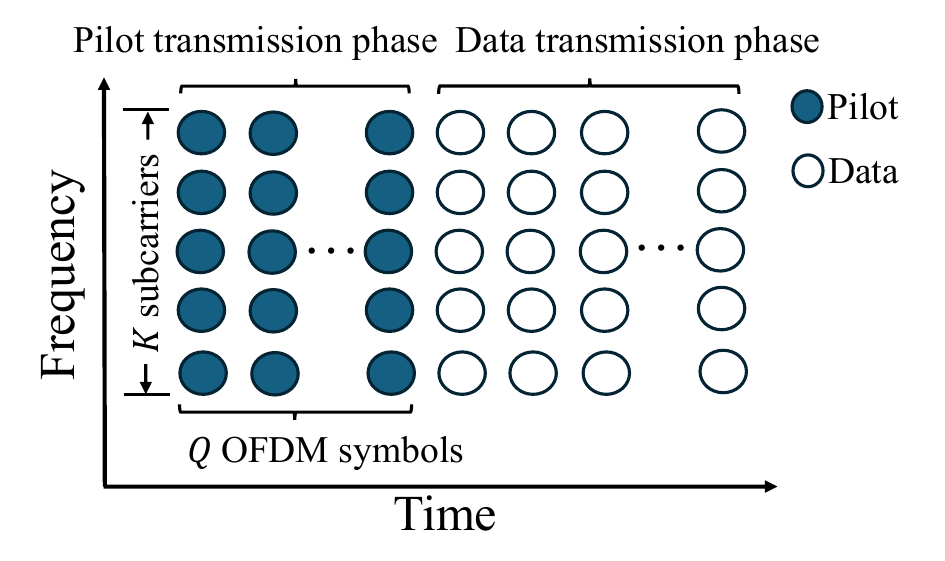}
	\vspace{-3mm}
	\caption{Illustration of time-frequency structure.}
	\label{fig_pilot}
\end{figure}

Assume $\sum_{q \in \mathcal{Q}} \|\tilde{\mathbf s}_{q,n}\|^2_2 = K$, $n \in \mathcal{N}$ and \textcolor{black}{that} each element of $\tilde{\mathbf s}_{q,n}$ is $\mathcal{O}(\frac{K}{\sqrt{L}})$, as $L\rightarrow\infty$.
For all $q\in\mathcal{Q}$ and $n\in\mathcal{N}$, the time-domain representation of $\tilde{\mathbf{s}}_{q,n}$ is given by its normalized inverse discrete Fourier transform (IDFT):
\begin{align}
	\mathbf{s}_{q,n} =\mathbf F^H\tilde{\mathbf{s}}_{q,n} \in \mathbb{C}^{K \times 1},\quad q\in\mathcal{Q},n \in \mathcal{N}.
\end{align}
Here, $\mathbf F \triangleq (F_{k,k'})_{k,k'\in\mathcal{K}} \in\mathbb C^{K\times K}$ represents the normalized discrete Fourier transform (DFT) matrix, where $F_{k,k'}\triangleq \frac{1}{\sqrt{K}}e^{-\frac{j2\pi (k-1) (k'-1)}{K}}$. Each device $n\in\mathcal{N}$ appends a cyclic prefix \textcolor{black}{with length} no smaller than $P-1$  to $\mathbf{s}_{q,n}$, for all $q\in\mathcal{Q}$, and transmits the $Q$ OFDM pilot symbols with cyclic prefixes. After removing the respective cyclic prefixes, the received signal for the $q$-th OFDM pilot symbol at antenna $m\in\mathcal{M}$ of the BS, $\mathbf y_{q,m} \triangleq (y_{q,k,m})_{k\in\mathcal K}\in\mathbb{C}^{K}$, can be written as:
\begin{align}
	\mathbf y_{q,m} = \sum_{n\in\mathcal N} a_n \mathbf H_{n,m}\mathbf{F}^{H} \tilde{\mathbf s}_{q,n} + \mathbf n_{q,m}, q\in\mathcal{Q},m\in\mathcal{M},
	\label{eq:receivesignal}
 \end{align}
where
\begin{align} \nonumber
	\mathbf H_{n,m}\triangleq & \begin{bmatrix}
		h_{n,1,m} & h_{n,K,m} &\cdots & h_{n,2,m}\\
		h_{n,2,m} & h_{n,1,m} &\cdots & h_{n,3,m}\\
		\vdots & \vdots &\ddots & \vdots\\
		h_{n,K,m} & h_{n,K-1,m}&\cdots & h_{n,1,m}\\
	\end{bmatrix} \in \mathbb{C}^{K\times K},
\end{align}
 and $\mathbf{n}_{q,m} \triangleq (n_{q,k,m})_{k\in\mathcal{K}} \in \mathbb{C}^{K}$ is the additive white Gaussian noise (AWGN)  with $n_{q,k,m},q\in\mathcal{Q},k \in \mathcal{K},m\in\mathcal{M}$ following i.i.d. $\mathcal {CN}(0,\sigma^2)$.
Here, for notation convenience, we let $h_{n,k,m}=0$, $k \in \mathcal{K}\backslash\mathcal{P}, n\in\mathcal N$, $m\in\mathcal M$. For all $n\in\mathcal{N},m\in\mathcal{M}$,  $\mathbf{H}_{n,m}$ is a circulant matrix that is specified by the first column or row vector.

To facilitate device activity detection and channel estimation, we present an equivalent model for received signals. 
By left-multiplying a DFT matrix to $\mathbf y_{q,m}$ in (\ref{eq:receivesignal}), we have: 
\begin{align}\nonumber
	&\tilde{\mathbf y}_{q,m}=\mathbf F\mathbf y_{q,m} = \sum_{n\in\mathcal N}a_n \mathbf F\mathbf H_{n,m}\mathbf F^H\tilde{\mathbf s}_{q,n} + \mathbf{F}\mathbf{n}_{q,m}\\ \nonumber
	&= \sum_{n\in\mathcal N}a_n{\rm diag}(\tilde{\mathbf{s}}_{q,n})\mathbf{F}(\mathbf H_{n,m})_{:,1}+\tilde{\mathbf n}_{q,m},q\in\mathcal{Q},m \in \mathcal{M}, 
\end{align}
where the last equality is because $\mathbf F\mathbf H_{n,m}\mathbf F^H = {\rm diag}(\mathbf{F}(\mathbf H_{n,m})_{:,1}) \in \mathbb{C}^{K\times K}$ is a diagonal matrix. Then,
by left-multiplying the normalized IDFT matrix to $\tilde{\mathbf y}_{q,m}$, we rewrite $\mathbf{y}_{q,m}$ in \eqref{eq:receivesignal} as:
\begin{align} \nonumber
	\mathbf y_{q,m} = \mathbf F^H\tilde{\mathbf y}_{q,m}  
	=  \sum_{n\in\mathcal N} a_n  \tilde{\mathbf{A}}_{q,n} \mathbf h_{n,m} + \mathbf n_{q,m}, q\in\mathcal{Q},m\in\mathcal{M}, 
\end{align}
where $\tilde{\mathbf{A}}_{q,n} \triangleq (\mathbf F^H{\rm diag}(\tilde{\mathbf{s}}_{q,n})\mathbf F)_{:,1:P} \in \mathbb C^{K\times P}$, $\mathbf h_{n,m} \triangleq (h_{n,p,m})_{p\in\mathcal{P}} \in \mathbb C^{P}$, and the last equality is due to $\mathbf{F}^H\mathbf{F}=\mathbf{I}_K$ and $h_{n,k,m}=0$, $k \in \mathcal{K}\backslash\mathcal{P}, n\in\mathcal N$, $m\in\mathcal M$. 
Define $\mathbf{Y}\in \mathbb{C}^{L \times M}$ with $\mathbf{Y}_{(q-1)K+1:qK,m} = \mathbf y_{q,m}$, $\mathbf{N}\in \mathbb{C}^{L \times M}$ with $\mathbf{N}_{(q-1)K+1:qK,m} = \mathbf n_{q,m}$, $\mathbf{A}\in \mathbb{C}^{L \times NP}$ with $\mathbf{A}_{(q-1)K+1:qK,(n-1)P+1:nP} = \tilde{\mathbf{A}}_{q,n}$, and $\mathbf{X}\in \mathbb{C}^{NP \times M}$ with $\mathbf{X}_{(n-1)P+1:nP,m} = a_n \mathbf h_{n,m}$. Let $x_{n,p,m}$ denote the $((n-1)P+p,m)$-th element of $\mathbf{X}$, and let $A_{l,n,p}$ denote the $(l,(n-1)P+p)$-th element of $\mathbf{A}$. By the assumptions on $\tilde{\mathbf{s}}_{q,n}$, we have $\sum_{l \in \mathcal{L}}|A_{l,n,p}|^2 = 1$, $A_{l,n,p} = \mathcal{O}(\frac{1}{\sqrt{L}})$, and $|A_{l,n,p}|^2 = \mathcal{O}(\frac{1}{L})$, as $L\rightarrow\infty$.
More compactly, the received signals for the $Q$ OFDM symbols \textcolor{black}{on} the $M$ antennas at the BS can be expressed as:
\begin{align} \label{eq:system_model}
	\mathbf{Y}  = \mathbf{A} \mathbf{X}+ \mathbf{N}.
\end{align}

\emph{Remark 1 (Comparisons for Signal Models):} The considered signal model for pilot length $L=KQ$ with $Q \in \mathbb{Z}_{++}$ under frequency-selective fading is more complex than those for pilot length $L$ under flat fading \cite{Liu18TSP} and pilot length $L=K$ under frequency-selective fading \cite{JiaTWC}.

\section{MAP Detection, MMSE Estimation, and Solutions}

\subsection{MAP Activity Detection}
For all $n\in \mathcal{N}$, one can apply the MAP detector for $a_n$, which maximizes a posterior probability $p(a_n|\mathbf{Y})$. The MAP detection problem of $a_n$ is formulated as follows\footnote{The joint detection of devices' activities $\mathbf{a}$ can be formulated as a MAP detection problem, i.e., $\max_{\mathbf{a}\in\{0,1\}^N} p(\mathbf{a}|\mathbf{Y})$, which is NP-hard. Hence, we approximate it with $N$ simpler scalar optimization problems as shown in (\ref{eq:op_prob_a_n}).}:
\begin{align}\label{eq:op_prob_a_n}
	\max_{a_n\in\{0,1\}} p(a_n|\mathbf{Y}).
\end{align} 
The optimal solution of the problem in (\ref{eq:op_prob_a_n}), denoted by $\hat{a}^{\star}_n(\mathbf{Y})$, can be written as:
\begin{align} \label{eq:detector_true}
\hat{a}^\star_n(\mathbf{Y})\triangleq\begin{cases}
	0,\ \theta_n < 0,\\
	1,\ \theta_n \geq 0,\\
\end{cases}
\end{align}
where 
\begin{align} \label{eq:theta_n}
\theta_n \triangleq \log \left(\frac{p(a_n=1|\mathbf{Y})}{p(a_n=0|\mathbf{Y})}\right).
\end{align}

\subsection{MMSE Channel Estimation}
\textcolor{black}{Noting} that $x_{n,p,m}=a_n h_{n,p,m}, n \in \mathcal{N}, p \in \mathcal{P}, m \in \mathcal{M}$, there are two methods for channel estimation. Firstly, for all $n\in\mathcal{N},p\in\mathcal{P}, m \in \mathcal{M}$, we can apply the Bayesian MMSE estimator for $x_{n,p,m}$, which minimizes the Bayesian MSE. 
The MMSE estimation problem of $x_{n,p,m}$ is formulated as follows:
\begin{align} \label{eq:op_prob_x_npm}
	\min_{\hat{x}_{n,p,m}(\mathbf{Y})}&\mathbb{E}\left[|{x}_{n,p,m}-\hat{x}_{n,p,m}(\mathbf{Y})|^2 |\mathbf{Y}\right].
\end{align}
By taking the first derivative of the objective function with respect to $\hat{x}_{n,p,m}(\mathbf{Y})$ and setting it to zero, we can obtain the optimal solution of the optimization problem in (\ref{eq:op_prob_x_npm}), given by:
\begin{align}\label{eq:MMSE_x_equation}
\hat{x}^{\star}_{n,p,m}(\mathbf{Y}) \triangleq   &\int x_{n,p,m} p(x_{n,p,m}|\mathbf{Y})d x_{n,p,m}.
\end{align}
Secondly, for all $n\in\mathcal{N},p\in\mathcal{P}, m \in \mathcal{M}$, we can apply the Bayesian MMSE estimator for $h_{n,p,m}$ by treating $a_n=1$.
The MMSE estimation problem of $h_{n,p,m}$ is formulated as follows:
\begin{align} \label{eq:op_prob_h_npm}
	\min_{\hat{h}_{n,p,m}(\mathbf{Y})}&\mathbb{E}\left[|{h}_{n,p,m}-\hat{h}_{n,p,m}(\mathbf{Y})|^2 |\mathbf{Y},a_n=1\right].
\end{align}
Similarly, by taking the first derivative of the objective function with respect to $\hat{h}_{n,p,m}(\mathbf{Y})$ and setting it to zero, we can obtain the optimal solution of the optimization problem in (\ref{eq:op_prob_h_npm}):
\begin{align} \label{eq:MMSE_h_equation}
	\hat{h}^{\star}_{n,p,m}(\mathbf{Y}) \triangleq  & \int h_{n,p,m} p(h_{n,p,m}|\mathbf{Y},a_n=1)d h_{n,p,m}.
\end{align}
Consequently, for all $n\in\mathcal{N}$, if device $n$ is detected \textcolor{black}{to be} active, the MMSE estimations of $h_{n,p,m},p\in\mathcal{P}, m \in \mathcal{M}$ are $\hat{h}^{\star}_{n,p,m}(\mathbf{Y}),p\in\mathcal{P}, m \in \mathcal{M}$; otherwise,  the MMSE estimation of $h_{n,p,m},p\in\mathcal{P}, m \in \mathcal{M}$ are void.

%
\subsection{\textcolor{black}{Solutions}} \label{sec:3c}

To calculate $\hat{a}^{\star}_n(\mathbf{Y})$ in (\ref{eq:detector_true}), $\hat{x}^{\star}_{n,p,m}(\mathbf{Y})$ in (\ref{eq:MMSE_x_equation}), and $\hat{h}^{\star}_{n,p,m}(\mathbf{Y})$ in (\ref{eq:MMSE_h_equation}), we need to first obtain $p(a_n|\mathbf{Y})$ (which determines the threshold $\theta_n$), $p(x_{n,p,m}|\mathbf{Y})$, and $p(h_{n,p,m}|\mathbf{Y},a_n=1)$, respectively.
\begin{lemma}\label{lemma:MMSE}
	We have
	\begin{align}\label{eq:mmse_a}
		&p(a_n|\mathbf{Y}) = \frac{\sum_{\mathbf{a}_{\bar{n}}} \int p(\mathbf{X},\mathbf{a},\mathbf{Y})d\mathbf{X} }{\sum_{\mathbf{a}} \int p(\mathbf{X},\mathbf{a},\mathbf{Y})d\mathbf{X} }, \\ \label{eq:MMSE_x}
		&p(x_{n,p,m}|\mathbf{Y}) = \frac{\sum_{\mathbf{a}}\int p(\mathbf{X},\mathbf{a},\mathbf{Y})d\mathbf{X}_{\overline{n,p,m}} }{\int \sum_{\mathbf{a}}\int p(\mathbf{X},\mathbf{a},\mathbf{Y})d\mathbf{X}_{\overline{n,p,m}} d x_{n,p,m}}, \\ \nonumber
		&{\color{black}p(h_{n,p,m}|\mathbf{Y},a_n) = p(x_{n,p,m}|\mathbf{Y},a_n)}\\ \label{eq:MMSE_h}
		&{\color{black}=\frac{\sum_{\mathbf{a}_{\bar{n}}}\int p(\mathbf{X},\mathbf{a},\mathbf{Y})d\mathbf{X}_{\overline{n,p,m}} }{\int \sum_{\mathbf{a}_{\bar{n}}}\int p(\mathbf{X},\mathbf{a},\mathbf{Y})d\mathbf{X}_{\overline{n,p,m}} d x_{n,p,m}}}, 
	\end{align}
where $n\in\mathcal{N}, p \in \mathcal{P}, m \in\mathcal{M}$, and 
\begin{align} \nonumber
	&p(\mathbf{X},\mathbf{a},\mathbf{Y}) =\prod_{l \in \mathcal{L}}\prod_{m \in \mathcal{M}}p(y_{l,m}|\mathbf{X}_{:,m})\prod_{n \in \mathcal{N}}p(a_n)\\ \label{eq:joint_distribution}&\times \prod_{p \in \mathcal{P}}\prod_{m \in \mathcal{M}} p(x_{n,p,m}|a_n).
\end{align}
Here, $p(y_{l,m}|\mathbf{X}_{:,m}) \triangleq f_{\mathcal{CN}}(y_{l,m};\mathbf{A}_{l,:}\mathbf{X}_{:,m},\sigma^2)$.
\end{lemma} 
\begin{figure}[!t]
	\centering
	\includegraphics[width=4.5cm]{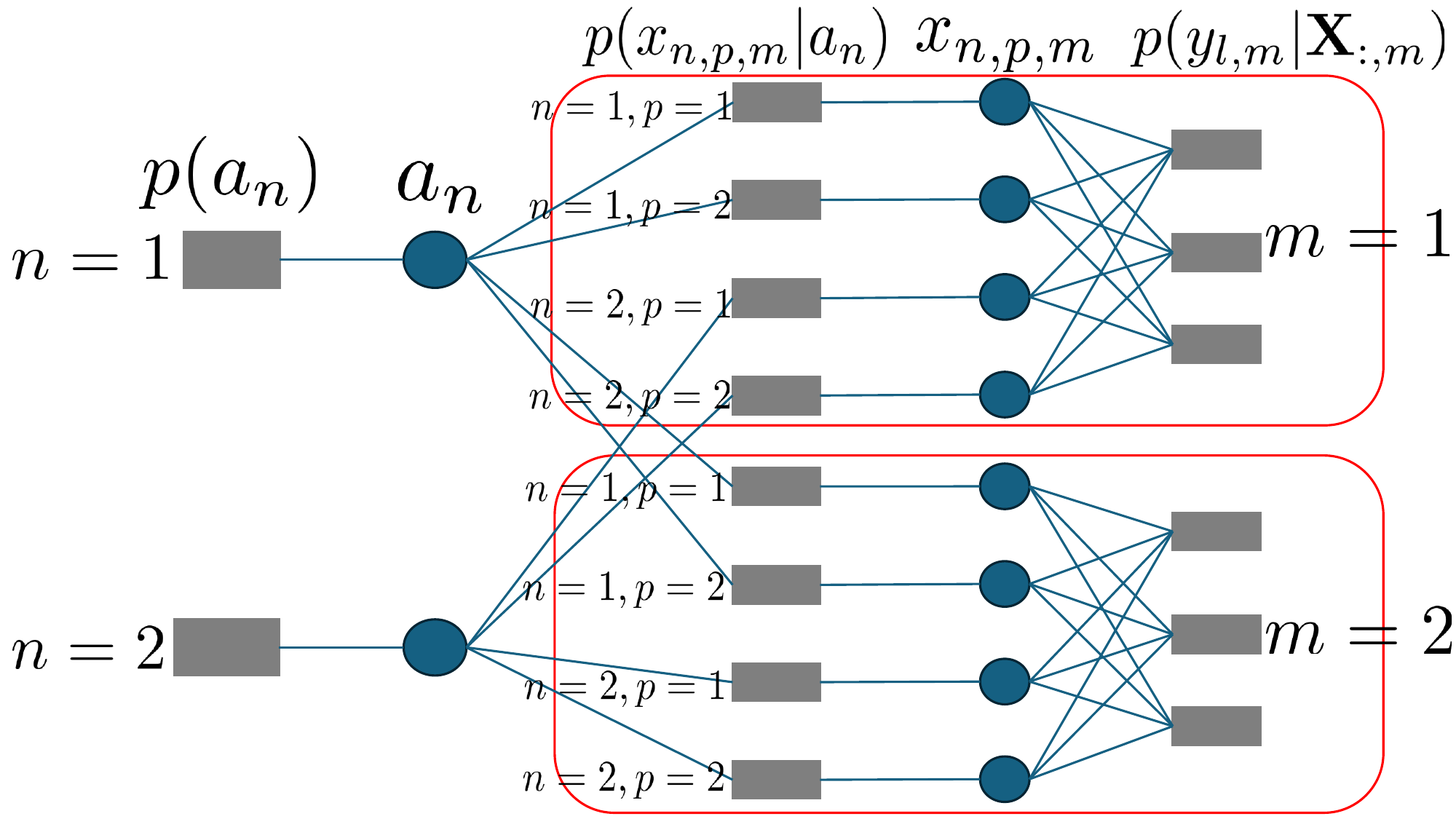}
	\caption{Illustration of the factor graph of $p(\mathbf{X},\mathbf{a},\mathbf{Y})$ in (\ref{eq:joint_distribution}).}
	\label{fig_factorgraph}
\end{figure}
\begin{proof}
See Appendix~\ref{appendix:marginal}.
\end{proof}

Lemma \ref{lemma:MMSE} indicates that multiple integrals have to be computed to successfully obtain  $p(a_n|\mathbf{Y})$ in (\ref{eq:mmse_a}), $p(x_{n,p,m}|\mathbf{Y})$ in (\ref{eq:MMSE_x}), and $p(h_{n,p,m}|\mathbf{Y},a_n)$ in (\ref{eq:MMSE_h}). The computation costs for the multiple integrals are significant, especially for large $N$ and $M$, rendering the exact computation of $\hat{a}^{\star}_n(\mathbf{Y})$, $\hat{x}^{\star}_{n,p,m}(\mathbf{Y})$, and $\hat{h}^{\star}_{n,p,m}(\mathbf{Y})$ nearly impractical{\footnote{The  closed-form expressions of $\hat{a}^{\star}_n(\mathbf{Y})$, $\hat{x}^{\star}_{n,p,m}(\mathbf{Y})$, and $\hat{h}^{\star}_{n,p,m}(\mathbf{Y})$ do not exist, primarily due to the complexities of deriving closed-form expression of $p(a_n|\mathbf{Y})$, $p(x_{n,p,m}|\mathbf{Y})$, and $p(h_{n,p,m}|\mathbf{Y},a_n)$ and the expectation w.r.t $x_{n,p,m}$ and $h_{n,p,m}$ \cite{Chen18TSP,Liu18TSP,Senel18TCOM,Shao19IoTJ,Qiu,QiuSPL,Zhang23TSP,GaoAMP_OFDM_TWC}.}.


To reduce the \textcolor{black}{computational} cost, we compute the multiple integrals in (\ref{eq:mmse_a}), (\ref{eq:MMSE_x}), and (\ref{eq:MMSE_h}) in a parallel and distributed manner, following the MP method. Specifically, we first graphically represent the factorable integrand of the multiple integrals in (\ref{eq:mmse_a}), (\ref{eq:MMSE_x}), and (\ref{eq:MMSE_h}), i.e., $p(\mathbf{X},\mathbf{a},\mathbf{Y})$, with a factor graph according to the exact factorization in (\ref{eq:joint_distribution}), as illustrated in Fig. \ref{fig_factorgraph}. 
Note that this factor graph reflects the facts that, for all $n \in \mathcal{N}$, $x_{n,p,m}, p\in \mathcal{P},m \in \mathcal{M}$ are conditionally independent given $a_n$, (i.e., unconditionally dependent), and $\mathbf{X}_{(n-1)P+1:nP,1:M},n \in \mathcal{N}$ are independent.
Then, we iteratively compute the exchanging messages between each variable node (circle) and each function node (rectangle) in the factor graph, which captures the exact relationship of $x_{n,p,m}, p\in \mathcal{P},m \in \mathcal{M}$ for all $n \in \mathcal{N}$. The definitions of the messages are given in Table \ref{tab:table_message}. We first initialize $\mu^{(0)}_{f_{l,m}\rightarrow x_{n,p,m}}(x_{n,p,m})$, $\mu^{(0)}_{x_{n,p,m} \rightarrow q_{n,p,m}}(x_{n,p,m})$, and $\mu^{(0)}_{q_{n,p,m}\rightarrow a_{n}}(a_{n})$. At each iteration $t$, we calculate:
\begin{align}  \nonumber
	&\mu^{(t)}_{f_{l,m} \rightarrow x_{n,p,m}}(x_{n,p,m}) \propto \int p(y_{l,m}|\mathbf{X}_{:,m}) \\ \label{eq:message_f_x}
	&\times \prod_{ \substack{i \in \mathcal{N}, j \in \mathcal{P},\\(i-1)\times P +j \neq (n-1)\times P +p}}  \mu^{(t)}_{x_{i,j,m} \rightarrow f_{l,m}}(x_{i,j,m}) d (\mathbf{X}_{:,m})_ {\overline{n,p,m}},\\
	\label{eq:message_product_f_x}
	&\mu^{(t)}_{x_{n,p,m} \rightarrow q_{n,p,m}}(x_{n,p,m}) 	\propto	 \prod_{l \in \mathcal{L}} \mu^{(t)}_{f_{l,m} \rightarrow x_{n,p,m}}(x_{n,p,m}),\\ \nonumber
	&\mu^{(t)}_{q_{n,p,m} \rightarrow a_{n}}(a_{n})\\ \label{eq:message_integral_x}
	&\propto \int  p(x_{n,p,m}|a_n) \mu^{(t)}_{x_{n,p,m} \rightarrow q_{n,p,m}}(x_{n,p,m}) d x_{n,p,m},
\end{align}
\begin{align}\nonumber
	&\mu^{(t+1)}_{q_{n,p,m} \rightarrow x_{n,p,m}}(x_{n,p,m})\\ \label{eq:message_q_x}
	&\propto \sum_{a_n = 0}^{1} \mu^{(t)}_{a_{n} \rightarrow q_{n,p,m}}(a_{n}) p(x_{n,p,m}|a_n) , \\ \label{eq:message_a_q}
	&\mu^{(t+1)}_{a_{n} \rightarrow q_{n,p,m}}(a_{n}) \propto p(a_n)\prod_{\substack{k\in \mathcal{M},\\ k \neq m}}\prod_{\substack{j \in \mathcal{P},\\j \neq p}} \mu^{(t)}_{q_{n,j,k} \rightarrow a_{n}}(a_{n}),\\ \nonumber
	&\mu^{(t+1)}_{x_{n,p,m} \rightarrow f_{l,m}}(x_{n,p,m})\propto \mu^{(t)}_{q_{n,p,m} \rightarrow x_{n,p,m}}(x_{n,p,m})\\ \label{eq:posterior_x_except}
	&\times \prod_{b \in \mathcal{L},b \neq l} \mu^{(t)}_{f_{b,m} \rightarrow x_{n,p,m}}(x_{n,p,m}). 
\end{align}
The notation ``$\propto$'' implies that the right-hand side differs from the left-hand side by a normalized constant to ensure that the integral result on the left-hand side equals 1. Moreover, since each message is derived from the multiplication, integration, or summation of probability distributions, it is inherently non-negative. Therefore, these messages can be considered as probability distributions. Besides, note that the messages in (\ref{eq:message_f_x}), (\ref{eq:message_product_f_x}), and (\ref{eq:message_integral_x}) are updated in parallel, and then the messages in (\ref{eq:message_q_x}), (\ref{eq:message_a_q}), and (\ref{eq:posterior_x_except}) are updated in parallel.
{
The approximations of (\ref{eq:mmse_a}), (\ref{eq:MMSE_x}), and (\ref{eq:MMSE_h}) are given by:
\begin{align}\label{eq:posterior_a}
	&p^{(t)}_{n}(a_{n}) \triangleq C^{-1}_1 p(a_n)\prod_{m \in \mathcal{M}} \prod_{p\in \mathcal{P}} \mu^{(t)}_{q_{n,p,m} \rightarrow a_{n}}(a_{n}),\\  \nonumber
	&p^{(t)}_{n,p,m}(x_{n,p,m}) \triangleq C^{-1}_2 \mu^{(t)}_{q_{n,p,m} \rightarrow x_{n,p,m}}(x_{n,p,m})\\ \label{eq:posterior_x} 
	& \times \mu^{(t)}_{x_{n,p,m} \rightarrow q_{n,p,m}}(x_{n,p,m}),\\ \nonumber
	&\tilde{p}^{(t)}_{n,p,m}(x_{n,p,m};a_n) \triangleq  C^{-1}_3 p(x_{n,p,m}|a_n)\\ \label{eq:inte_posterior_xa}
	&\times \mu^{(t)}_{x_{n,p,m} \rightarrow q_{n,p,m}}(x_{n,p,m}),
\end{align}
 at the $t$-th iteration, respectively, with normalizing constants: 
\begin{align*}
	C_1 &\triangleq \sum_{a_n=0}^{1} p(a_n) \prod_{m \in \mathcal{M}} \prod_{p\in \mathcal{P}} \mu^{(t)}_{q_{n,p,m} \rightarrow a_{n}}(a_{n}),\\
C_2 &\triangleq \int \mu^{(t)}_{q_{n,p,m} \rightarrow x_{n,p,m}}(x_{n,p,m}) \\
&\times \mu^{(t)}_{x_{n,p,m} \rightarrow q_{n,p,m}}(x_{n,p,m}) d x_{n,p,m},\\
	C_3 &\triangleq \int p(x_{n,p,m}|a_n)\mu^{(t)}_{x_{n,p,m} \rightarrow q_{n,p,m}}(x_{n,p,m}) d x_{n,p,m}.
\end{align*} 
Therefore, the approximations of $\hat{a}^{\star}_n(\mathbf{Y})$, $\hat{x}^{\star}_{n,p,m}(\mathbf{Y})$, and $\hat{h}^{\star}_{n,p,m}(\mathbf{Y})$ at the $t$-th iteration are given by:
\begin{align} \label{eq:a_detect}
	\hat{a}^{(t)}_n&=\begin{cases}
		0,\ \theta^{(t)}_n < 0,\\
		1,\ \theta^{(t)}_n \geq 0,\\
	\end{cases}\\ \label{eq:inte_posterior_x} 
	\hat{x}^{(t)}_{n,p,m} &\triangleq \int x_{n,p,m}  p^{(t)}_{n,p,m}(x_{n,p,m}) d x_{n,p,m}, 
\end{align}
\begin{align}
 \label{eq:inte_posterior_h} 
	\hat{h}^{(t)}_{n,p,m} \triangleq \int x_{n,p,m}  \tilde{p}^{(t)}_{n,p,m}(x_{n,p,m};a_n=1) d x_{n,p,m},
\end{align} 
respectively, where $n\in\mathcal{N}, p \in \mathcal{P}, m \in\mathcal{M}$, and 
\begin{align} \label{eq:detector_MP}
	\theta^{(t)}_n = \log \frac{p^{(t)}_{n}(a_{n}=1)}{p^{(t)}_{n}(a_{n}=0)}.
\end{align}
Here $\theta^{(t)}_n$ is the approximation of $\theta_n$ in (\ref{eq:theta_n}) at the $t$-th iteration.
Note that, for all $n\in\mathcal{N}, p \in \mathcal{P}, m \in\mathcal{M}$, if the message updates converge, $p^{(t)}_{n}(a_{n})$ in (\ref{eq:posterior_a}), $p^{(t)}_{n,p,m}(x_{n,p,m})$ in (\ref{eq:posterior_x}), and $\tilde{p}^{(t)}_{n,p,m}(x_{n,p,m})$ in (\ref{eq:inte_posterior_xa}) converge to $p(a_n|\mathbf{Y})$ in (\ref{eq:mmse_a}), $p(x_{n,p,m}|\mathbf{Y})$ in (\ref{eq:MMSE_x}), and $p(h_{n,p,m}|\mathbf{Y},a_n)$ in (\ref{eq:MMSE_h}), respectively, and $\hat{a}^{(t)}_n$ in (\ref{eq:a_detect}), $\hat{x}^{(t)}_{n,p,m}$ in (\ref{eq:inte_posterior_x}), and $\hat{h}^{(t)}_{n,p,m}$ in (\ref{eq:inte_posterior_h}) converge to $\hat{a}^{\star}_n(\mathbf{Y})$ in (\ref{eq:detector_true}), $\hat{x}^{\star}_{n,p,m}(\mathbf{Y})$ in (\ref{eq:MMSE_x_equation}), and $\hat{h}^{\star}_{n,p,m}(\mathbf{Y})$ in (\ref{eq:MMSE_h_equation}), respectively.

However, the MP method is still not tractable. Because keeping track of the above messages in (\ref{eq:message_f_x})-(\ref{eq:posterior_x_except}), which are functions of $a_n,n\in\mathcal{N}$ or $x_{n,p,m},n\in\mathcal{N},p\in\mathcal{P},m\in\mathcal{M}$, is impractical. Thus, further approximations of the MP method are inevitable.

}

 \begin{table}[!t]
	\caption{Message definitions. \label{tab:table_message}}
	\centering
	\scriptsize
	\begin{tabular}{|c||c|}
		\hline
		\makecell[c]{$\mu^{(t)}_{f_{l,m} \rightarrow x_{n,p,m}}(x_{n,p,m})$} & \makecell[c]{message from $p(y_{l,m}|\mathbf{X}_{:,m})$ to $x_{n,p,m}$} \\
		\hline
		\makecell[c]{$\mu^{(t)}_{x_{n,p,m} \rightarrow f_{l,m}}(x_{n,p,m})$} & \makecell[c]{message from $x_{n,p,m}$ to $p(y_{l,m}|\mathbf{X}_{:,m})$} \\
		\hline
		\makecell[c]{$\mu^{(t)}_{x_{n,p,m} \rightarrow q_{n,p,m}}(x_{n,p,m})$} &\makecell[c]{message from $x_{n,p,m}$ to $ p(x_{n,p,m}|a_n)$} \\
		\hline
		\makecell[c]{$\mu^{(t)}_{q_{n,p,m} \rightarrow x_{n,p,m}}(x_{n,p,m})$} & \makecell[c]{message from $p(x_{n,p,m}|a_n)$ to $x_{n,p,m}$}\\
		\hline
		\makecell[c]{$\mu^{(t)}_{a_{n} \rightarrow q_{n,p,m}}(a_{n})$} & \makecell[c]{ message from $a_{n}$ to $ p(x_{n,p,m}|a_n)$}\\
		\hline
		\makecell[c]{$\mu^{(t)}_{q_{n,p,m} \rightarrow a_{n}}(a_{n})$} & \makecell[c]{message from $p(x_{n,p,m}|a_n)$ to $a_{n}$}\\
		\hline
	\end{tabular}
\end{table}

\section{Approximations of Messages and Estimates}\label{sec:approx}
In this section, we approximate the message updates in (\ref{eq:message_f_x})-(\ref{eq:posterior_x_except}) and obtain  tractable updates of the estimates $\hat{a}^{(t)}_n$ in (\ref{eq:a_detect}) (i.e., $\theta^{(t)}_n$ in (\ref{eq:detector_MP})), $\hat{x}^{(t)}_{n,p,m}$ in (\ref{eq:inte_posterior_x}), and $\hat{h}^{(t)}_{n,p,m}$ in (\ref{eq:inte_posterior_h}), following the AMP method \cite{DonohoAMP}, which is an approximation of the MP method. 
	For all $l\in\mathcal{L}, n\in\mathcal{N}, p \in \mathcal{P}, m \in\mathcal{M}$, let $\hat{x}^{(t)}_{ x_{n,p,m} \rightarrow f_{l,m}}$ and $\hat{\nu}^{(t)}_{ x_{n,p,m} \rightarrow f_{l,m}}$ denote the mean and variance of $\mu^{(t)}_{x_{n,p,m} \rightarrow f_{l,m}}(x_{n,p,m})$ in (\ref{eq:posterior_x_except}).
	First, we approximate the message $\mu^{(t)}_{f_{l,m} \rightarrow x_{n,p,m}}(x_{n,p,m})$ in (\ref{eq:message_f_x}). 
	\begin{lemma}\label{lemma:version1_1}
		For all $t$, as $N \rightarrow \infty$, we have:
		\begin{align} \nonumber
			&\mu^{(t)}_{f_{l,m} \rightarrow x_{n,p,m}}(x_{n,p,m}) \\ \nonumber
			&- f_{\mathcal{CN}}(x_{n,p,m}; \frac{\hat{z}^{(t)}_{f_{l,m} \rightarrow x_{n,p,m}}}{A_{l,n,p}}, \frac{\hat{\gamma}^{(t)}_{f_{l,m} \rightarrow x_{n,p,m}}}{|A_{l,n,p}|^2})\xrightarrow{\|\cdot\|_K} 0, 
		\end{align}
		where 
		\begin{align} \nonumber
			&\hat{z}^{(t)}_{f_{l,m} \rightarrow x_{n,p,m}}  \\ \label{eq:mean_Z_f_x}
			&\triangleq y_{l,m} - \sum\limits_{\substack{i \in \mathcal{N}, j \in \mathcal{P},\\(i-1)\times P +j \neq (n-1)\times P +p}} A_{l,i,j}\hat{x}^{(t)}_{ x_{i,j,m} \rightarrow f_{l,m}},
				\end{align}
		\begin{align}	\nonumber
			&\hat{\gamma}^{(t)}_{f_{l,m} \rightarrow x_{n,p,m}} \\ \label{eq:mean_V_f_x}
			&\triangleq \sigma^{2} + \sum\limits_{\substack{i \in \mathcal{N}, j \in \mathcal{P},\\(i-1)\times P +j \neq (n-1)\times P +p}} |A_{l,i,j}|^2\hat{\nu}^{(t)}_{ x_{i,j,m} \rightarrow f_{l,m}}.
		\end{align}
	\end{lemma}
	\begin{proof}
		The proof follows \textcolor{black}{from} that of \cite[Lemma~3.1]{DonohoAMP}.
	\end{proof}
	
	Lemma \ref{lemma:version1_1} states that the message $\mu^{(t)}_{f_{l,m} \rightarrow x_{n,p,m}}(x_{n,p,m})$ in (\ref{eq:message_f_x}) can be approximated as a complex Gaussian PDF at large $N$.
	Define
		$\eta_n\left(x,y,z\right) \triangleq \frac{\beta_n x}{\left(1+ \frac{(1-z)f_{\mathcal{CN}}(0;x,y)}{zf_{\mathcal{CN}}(0;x,y+ \beta_n)}\right)\left(y + \beta_n\right)}$
	and $\eta^{'}_n\left(x,y,z\right) \triangleq \frac{\partial \eta_n\left(x,y,z\right)}{\partial x}$. Next, we approximate the remaining messages in (\ref{eq:message_product_f_x})-(\ref{eq:posterior_x_except}).
	
	\begin{lemma} \label{lemma:consistent_message}
		For all $t$, suppose
		\begin{align} \nonumber
			&\mu^{(t)}_{f_{l,m} \rightarrow x_{n,p,m}}(x_{n,p,m}) \\ \label{lemma2:eq}
			&= f_{\mathcal{CN}}(x_{n,p,m};\frac{\hat{z}^{(t)}_{f_{l,m} \rightarrow x_{n,p,m}}}{A_{l,n,p}},\frac{\hat{\gamma}^{(t)}_{f_{l,m} \rightarrow x_{n,p,m}}}{|A_{l,n,p}|^2}),
		\end{align} 
		where $\hat{z}^{(t)}_{f_{l,m} \rightarrow x_{n,p,m}}$ and $\hat{\gamma}^{(t)}_{f_{l,m} \rightarrow x_{n,p,m}}$ are given by (\ref{eq:mean_Z_f_x}) and (\ref{eq:mean_V_f_x}), respectively. (i) We have:
		\begin{align} \label{eq:productL}
			&\mu^{(t)}_{x_{n,p,m} \rightarrow q_{n,p,m}}(x_{n,p,m})= f_{\mathcal{CN}}(x_{n,p,m};\hat{r}^{(t)}_{x_{n,p,m}},\hat{\sigma}^{(t)}_{x_{n,p,m}})	,	  
		\end{align}
		\begin{align} \label{eq:mes_q_a}
			&\mu^{(t)}_{q_{n,p,m} \rightarrow a_{n}}(a_{n}) = f_{\mathcal{B}}(a_{n};{\hat{\rho}}^{(t)}_{n,p,m}),\\ \label{eq:mes_a_q}
			&\mu^{(t+1)}_{a_{n} \rightarrow q_{n,p,m}}(a_{n}) = f_{\mathcal{B}}(a_{n};\lambda^{(t+1)}_{n,p,m}), \\ \label{eq:approximate_q_x}
			&\mu^{(t+1)}_{q_{n,p,m} \rightarrow x_{n,p,m}}(x_{n,p,m}) =  f_{\mathcal{B-CN}}(x_{n,p,m};0,\beta_n,\lambda^{(t+1)}_{n,p,m}),\\ \nonumber
			&\mu^{(t+1)}_{x_{n,p,m} \rightarrow f_{l,m}}(x_{n,p,m}) = C^{-1}_4 f_{\mathcal{B-CN}}(x_{n,p,m};0,\beta_n,\lambda^{(t+1)}_{n,p,m}) \\ \label{eq:approximate_x_f}
			&\times f_{\mathcal{CN}}(x_{n,p,m}; \hat{r}^{(t)}_{x_{n,p,m} \rightarrow f_{l,m}}, \hat{\sigma}^{(t)}_{x_{n,p,m}\rightarrow f_{l,m}}), \\ \label{eq:mean_x_f_x}
			&\hat{x}^{(t+1)}_{ x_{n,p,m} \rightarrow f_{l,m}} =  \eta_n(\hat{r}^{(t)}_{x_{n,p,m} \rightarrow f_{l,m}},\hat{\sigma}^{(t)}_{x_{n,p,m}\rightarrow f_{l,m}},\lambda^{(t+1)}_{n,p,m}),\\ \nonumber
			&\hat{\nu}^{(t+1)}_{ x_{n,p,m} \rightarrow f_{l,m}} = \hat{\sigma}^{(t)}_{x_{n,p,m}\rightarrow f_{l,m}}\\ \label{eq:mean_nu_f_x}
			&\times \eta^{'}_n(\hat{r}^{(t)}_{x_{n,p,m} \rightarrow f_{l,m}},\hat{\sigma}^{(t)}_{x_{n,p,m}\rightarrow f_{l,m}},\lambda^{(t+1)}_{n,p,m}),
		\end{align}
		where 
		\begin{align}
			\label{eq:lemma_rho_npm}
			&\hat{\rho}^{(t)}_{n,p,m} \triangleq\frac{1}{1 + \frac{f_{\mathcal{CN}}(0;\hat{r}^{(t)}_{x_{n,p,m}},\hat{\sigma}^{(t)}_{x_{n,p,m}} )}{f_{\mathcal{CN}}(0;\hat{r}^{(t)}_{x_{n,p,m}},\hat{\sigma}^{(t)}_{x_{n,p,m}} + \beta_n)} },\\ \label{eq:lemma_lambda_npm}
			&\lambda^{(t+1)}_{n,p,m} \triangleq \frac{1}{1 + \frac{(1-\rho_n)}{\rho_n}\prod\limits_{k\in \mathcal{M},k \neq m}\prod\limits_{j \in \mathcal{P}, j \neq p} \frac{(1-\hat{\rho}^{(t)}_{n,j,k})}{\hat{\rho}^{(t)}_{n,j,k}}}, 
\\ \label{eq:lemma_sigma_x}
			&\hat{\sigma}^{(t)}_{x_{n,p,m}} \triangleq \Big(\sum_{l \in \mathcal{L}}\frac{|A_{l,n,p}|^2}{\hat{\gamma}^{(t)}_{f_{l,m} \rightarrow x_{n,p,m}}}\Big)^{-1},\\ \label{eq:lemma_r_x}
			&\hat{r}^{(t)}_{x_{n,p,m} } \triangleq \hat{\sigma}^{(t)}_{x_{n,p,m}} \Big(\sum_{l \in \mathcal{L}}\frac{A^{*}_{l,n,p}\hat{z}^{(t)}_{f_{l,m} \rightarrow x_{n,p,m}}}{\hat{\gamma}^{(t)}_{f_{l,m} \rightarrow x_{n,p,m}}}\Big),\\\label{eq:lemma_sigma_f_x}
			&\hat{\sigma}^{(t)}_{x_{n,p,m}\rightarrow f_{l,m}} \triangleq \Big( \sum_{b \in \mathcal{L}, b \neq l}\frac{|A_{b,n,p}|^2}{\hat{\gamma}^{(t)}_{f_{b,m} \rightarrow x_{n,p,m}}}\Big)^{-1},\\   \label{eq:lemma_r_f_x}
			&\hat{r}^{(t)}_{x_{n,p,m} \rightarrow f_{l,m}} \triangleq \hat{\sigma}^{(t)}_{x_{n,p,m}\rightarrow f_{l,m}} \Big(\sum_{\substack{b \in \mathcal{L},\\ b \neq l}}\frac{A^{*}_{b,n,p}\hat{z}^{(t)}_{f_{b,m} \rightarrow x_{n,p,m}}}{\hat{\gamma}^{(t)}_{f_{b,m} \rightarrow x_{n,p,m}}}\Big) ,\\ \nonumber
			&C_4 \triangleq \int f_{\mathcal{B-CN}}(x_{n,p,m};0,\beta_n,\lambda^{(t+1)}_{n,p,m})\\
			&\times f_{\mathcal{CN}}(x_{n,p,m}; \hat{r}^{(t)}_{x_{n,p,m} \rightarrow f_{l,m}}, \hat{\sigma}^{(t)}_{x_{n,p,m}\rightarrow f_{l,m}}) d x_{n,p,m}.
		\end{align}
		(ii) We have:
		\begin{align} \label{eq:lemma_x}
			&\hat{x}^{(t+1)}_{ n,p,m} = \eta_n(\hat{r}^{(t)}_{x_{n,p,m}},\hat{\sigma}^{(t)}_{x_{n,p,m}},\lambda^{(t+1)}_{n,p,m}),\\ \label{eq:lemma_h}
			&\hat{h}^{(t+1)}_{ n,p,m} = \frac{\beta_n \hat{r}^{(t)}_{x_{n,p,m}}}{\hat{\sigma}^{(t)}_{x_{n,p,m}} + \beta_n},\\ \label{eq:lemma_llr}
			&\theta^{(t+1)}_n = \log \frac{\rho_n \prod\limits_{m \in \mathcal{M}} \prod\limits_{p \in \mathcal{P}} f_{\mathcal{CN}}(0;\hat{r}^{(t)}_{x_{n,p,m}},\hat{\sigma}^{(t)}_{x_{n,p,m}} + \beta_n)}{(1-\rho_n ) \prod\limits_{m \in \mathcal{M}} \prod\limits_{p \in \mathcal{P}} f_{\mathcal{CN}}(0;\hat{r}^{(t)}_{x_{n,p,m}},\hat{\sigma}^{(t)}_{x_{n,p,m}})}.
		\end{align}
	\end{lemma}
	\begin{proof}
	See Appendix~\ref{appendix:a}.
	\end{proof}
	
	Lemma~\ref{lemma:consistent_message} (i) indicates that under the assumption that $\mu^{(t)}_{f_{l,m} \rightarrow x_{n,p,m}}(x_{n,p,m})$ in (\ref{eq:message_f_x}) takes the form of a complex Gaussian PDF, the other messages take the forms of complex Gaussian, Bernoulli, and Bernoulli-Gaussian distributions, determined by parameters. Therefore, the message updates in (\ref{eq:message_f_x})-(\ref{eq:posterior_x_except}) can be \textcolor{black}{made} through the parameter updates in (\ref{eq:mean_Z_f_x})-(\ref{eq:mean_V_f_x}) and (\ref{eq:mean_x_f_x})-(\ref{eq:lemma_r_f_x}) with reduced computational complexity. Consequently, Lemma~\ref{lemma:consistent_message} (ii) indicates that, under the same assumption, the estimates $\hat{x}^{(t+1)}_{ n,p,m}$, $\hat{h}^{(t+1)}_{ n,p,m}$, and $\theta^{(t+1)}_n$ can be obtained based on the parameter updates, as shown in (\ref{eq:lemma_x}), (\ref{eq:lemma_h}), and (\ref{eq:lemma_llr}), respectively.
	Specifically, we first update the parameters $\hat{z}^{(t)}_{f_{l,m} \rightarrow x_{n,p,m}}$ and $\hat{\gamma}^{(t)}_{f_{l,m} \rightarrow x_{n,p,m}}$ in (\ref{eq:mean_Z_f_x}) and (\ref{eq:mean_V_f_x}), respectively. Then, we calculate the other parameters in (\ref{eq:mean_x_f_x})-(\ref{eq:lemma_r_f_x}) based on $\hat{z}^{(t)}_{f_{l,m} \rightarrow x_{n,p,m}}$ and $\hat{\gamma}^{(t)}_{f_{l,m} \rightarrow x_{n,p,m}}$. At last, we obtain $\hat{x}^{(t+1)}_{ n,p,m}$, $\hat{h}^{(t+1)}_{ n,p,m}$, and $\theta^{(t+1)}_n$ based on these parameters. 
	
	Since there are $6LNPM + 3NPM$ parameters, the parameter updates and calculations are still computationally expensive for large $L$ and $N$. This motivates us to further simplify the parameter operations and estimate updates. Specifically, we first introduce a new parameter:
	\begin{align} \label{eq:new_v_sec4}
		\tau^{(t)}_{l,m}  \triangleq \sigma^{2} + \sum_{n \in \mathcal{N}} \sum_{p \in \mathcal{P}} |A_{l,n,p}|^2 \hat{\nu}^{(t)}_{ x_{n,p,m} \rightarrow f_{l,m}}.
	\end{align}
	 Then, by (\ref{eq:new_v_sec4}) and Lemma~\ref{lemma:consistent_message}, we have the following theorem.
	 \begin{lemma} \label{theo:tau}
	 	For all $t$, suppose that (\ref{lemma2:eq}) holds,
	 		\begin{align} \nonumber
	 			&\sum_{n \in \mathcal{N}} \sum_{p \in \mathcal{P}} |A_{l,n,p}|^2 \hat{\nu}^{(t)}_{ x_{n,p,m} \rightarrow f_{l,m}} 
	 			\\ \label{eq:appro_tau}
	 			&= \frac{1}{L}\sum_{n \in \mathcal{N}} \sum_{p \in \mathcal{P}} \hat{\nu}^{(t)}_{ x_{n,p,m} \rightarrow f_{l,m}} + \mathcal{O}\left(\frac{1}{L}\right), \text{as $L \rightarrow \infty$}
	 		\end{align}
	 		holds, where $\mathcal{O}(\cdot)$ is uniformly in $l$, and
	 		\begin{align} \nonumber
	 			&\hat{\nu}^{(t+1)}_{ x_{n,p,m} \rightarrow f_{l,m}} =  \hat{\sigma}^{(t)}_{x_{n,p,m}}\eta^{'}_n(\hat{r}^{(t)}_{x_{n,p,m} },\hat{\sigma}^{(t)}_{x_{n,p,m}},\lambda^{(t+1)}_{n,p,m})\\ \label{eq:taylor_v}
	 			& + \mathcal{O}\left(\frac{1}{N}\right), \text{as $N \rightarrow \infty$}
	 		\end{align}
	 		holds, where $\mathcal{O}(\cdot)$ is uniformly in $(l,n)$. Then, we have:
	 		\begin{align} \label{eq:theorem_v}
	 			&\tau^{(t)}_{l,m} = \hat{\tau}^{(t)}_{m} + o(1), \text{as $L,N \rightarrow \infty$}, \\ \label{eq:theorem_gamma}
	 			&\hat{\gamma}^{(t)}_{f_{l,m} \rightarrow x_{n,p,m}} = \tau^{(t)}_{l,m} + o(1), \text{as $L \rightarrow \infty$}, \\
	 			 \label{eq:theorem_sigma}
	 			&\hat{\sigma}^{(t)}_{x_{n,p,m}} = \hat{\tau}^{(t)}_{m} + o(1) , \text{as $L,N \rightarrow \infty$}, \\ \label{eq:lemma4_sigma_x_f}
	 			&(\hat{\sigma}^{(t)}_{x_{n,p,m}\rightarrow f_{l,m}})^{-1} = (\hat{\sigma}^{(t)}_{x_{n,p,m}})^{-1} + o(1) , \text{as $L \rightarrow \infty$},
	 		\end{align}
	 		where 
	 		\begin{align} \label{eq:tau_m}
	 			\hat{\tau}^{(t)}_{m} \triangleq \sigma^{2} + \frac{\sum\limits_{n \in \mathcal{N}} \sum\limits_{p \in \mathcal{P}}\hat{\sigma}^{(t-1)}_{x_{n,p,m}}\eta_n^{'}(\hat{r}^{(t-1)}_{x_{n,p,m}},\hat{\sigma}^{(t-1)}_{x_{n,p,m}},\lambda^{(t)}_{n,p,m})}{L}  .
	 		\end{align}
	 \end{lemma}
	 \begin{proof}
	 	See Appendix~\ref{appendix:lemma4}.
	 \end{proof}
	
	Lemma~\ref{theo:tau} indicates that the parameters $\hat{\sigma}^{(t)}_{x_{n,p,m}}$ in (\ref{eq:lemma_sigma_x}) and $\tau^{(t)}_{l,m}$ in (\ref{eq:new_v_sec4}) can be approximated by $\hat{\tau}^{(t)}_{m}$ for large $L$ and $N$, and the parameters $\hat{\gamma}^{(t)}_{f_{l,m} \rightarrow x_{n,p,m}}$ in (\ref{eq:mean_V_f_x}) and $\hat{\sigma}^{(t)}_{x_{n,p,m}\rightarrow f_{l,m}}$ in (\ref{eq:lemma_sigma_f_x}) can be approximated by $\tau^{(t)}_{l,m}$ and $\hat{\sigma}^{(t)}_{x_{n,p,m}}$, respectively, for large $L$. Therefore, for large $L$ and $N$, the computation of parameter $\hat{\sigma}^{(t)}_{x_{n,p,m}}$ no longer depends on the parameters in (\ref{eq:mean_V_f_x}) and (\ref{eq:lemma_sigma_f_x}).
	
	The parameter operations in (\ref{eq:mean_Z_f_x}), (\ref{eq:mean_x_f_x}), (\ref{eq:lemma_r_x}), and (\ref{eq:lemma_r_f_x}) are still computationally intensive for large $L$ and $N$. To address this issue, we will further simplify these parameter operations and the respective estimate updates in Sections \ref{sec:1} and \ref{sec:2}.
}

\section{AMP for Device Activity Detection and Effective Channel Estimation} \label{sec:1}
In this section, we present the \emph{AMP-A-EC} algorithm for activity detection and effective channel estimation, which effectively updates the estimates $\theta^{(t)}_n$ and $\hat{x}^{(t)}_{n,p,m}$ based on the approximation results in Section~\ref{sec:approx}.

\subsection{Simplifications of Parameter and Estimate Updates} \label{sec:MA1}
In this part, we simplify the parameter operations in (\ref{eq:mean_Z_f_x}), (\ref{eq:mean_x_f_x}), (\ref{eq:lemma_r_x}), and (\ref{eq:lemma_r_f_x}) to enable efficient updates of $\hat{x}^{(t)}_{n,p,m}$ in (\ref{eq:lemma_x}) and $\theta^{(t)}_n$ in (\ref{eq:lemma_llr}). First, we introduce \textcolor{black}{a} new parameter:
\begin{align} \label{eq:new_z}
	&z^{(t)}_{f_{l,m} }  \triangleq y_{l,m} - \sum\nolimits_{n \in \mathcal{N}} \sum\nolimits_{p \in \mathcal{P}} A_{l,n,p}\hat{x}^{(t)}_{ x_{n,p,m} \rightarrow f_{l,m}}.
\end{align}
Then, by (\ref{eq:new_z}), Lemma~\ref{lemma:consistent_message}, and Lemma~\ref{theo:tau}, we have Theorem~\ref{theo:algorithm}.


\begin{theorem} \label{theo:algorithm}
	For all $t$, suppose that (\ref{lemma2:eq}) holds, $\hat{\gamma}^{(t)}_{f_{l,m} \rightarrow x_{n,p,m}} = \tau^{(t)}_{l,m}$, $\hat{\sigma}^{(t)}_{x_{n,p,m}\rightarrow f_{l,m}} = \hat{\sigma}^{(t)}_{x_{n,p,m}}$, $\tau^{(t)}_{l,m} = \hat{\tau}^{(t)}_{m}$, $\hat{\sigma}^{(t)}_{x_{n,p,m}} = \hat{\tau}^{(t)}_{m}$, and
	\begin{align} \nonumber
	&\hat{x}^{(t+1)}_{ x_{n,p,m} \rightarrow f_{l,m}} = \hat{x}^{(t+1)}_{ n,p,m} + \eta^{'}_n(\hat{r}^{(t)}_{x_{n,p,m} },\hat{\sigma}^{(t)}_{x_{n,p,m}},\lambda^{(t+1)}_{n,p,m}) \\  \nonumber
	& \times(\hat{r}^{(t)}_{x_{n,p,m} \rightarrow f_{l,m}}-\hat{r}^{(t)}_{x_{n,p,m}}) + \mathcal{O}\left(\frac{1}{L}\right) + \mathcal{O}\left(\frac{1}{N}\right), \\ \label{eq:theorem1_taylor}
	&\text{as $L,N \rightarrow \infty$}
	\end{align}
	hold, where $\mathcal{O}(\cdot)$ is uniformly in $(l,n)$, and
		\begin{align} \nonumber
			&\sum_{n \in \mathcal{N}} \sum_{p \in \mathcal{P}} |A_{l,n,p}|^2 \eta_n^{'}\left(\hat{r}^{(t)}_{x_{n,p,m}}, \hat{\tau}^{(t)}_{m}, \lambda^{(t+1)}_{n,p,m}\right)
			\\ \nonumber
			&= \frac{1}{L}\sum_{n \in \mathcal{N}} \sum_{p \in \mathcal{P}} \eta_n^{'}\left(\hat{r}^{(t)}_{x_{n,p,m}}, \hat{\tau}^{(t)}_{m}, \lambda^{(t+1)}_{n,p,m}\right) + \mathcal{O}\left(\frac{1}{L}\right),\\ \label{eq:appro_A_x}
			 &\text{as $L \rightarrow \infty$}
		\end{align}
		holds, where $\mathcal{O}(\cdot)$ is uniformly in $l$. Then, we have:
	\begin{align}\nonumber
		&z^{(t)}_{f_{l,m}} = y_{l,m} - \sum_{n \in \mathcal{N}}\sum_{p \in \mathcal{P}} A_{l,n,p} \hat{x}^{(t)}_{n,p,m} + \frac{1}{L}   z^{(t-1)}_{f_{l,m}}  \\ \nonumber
		&\times \sum_{n \in \mathcal{N}}\sum_{p \in \mathcal{P}}    \eta_n^{'}\left(\hat{r}^{(t)}_{x_{n,p,m}}, \hat{\tau}^{(t)}_{m}, \lambda^{(t+1)}_{n,p,m}\right) + o(1), \\ \label{eq:theorem_z}
		&\text{as $L,N \rightarrow \infty$}, \\ \label{eq:theorem_r}
		&\hat{r}^{(t)}_{x_{n,p,m}}=\hat{x}^{(t)}_{n,p,m} + \sum_{l \in \mathcal{L}}A^{*}_{l,n,p}z^{(t)}_{f_{l,m}}    + o(1), \text{as $L,N \rightarrow \infty$}.
		\end{align}
\end{theorem}
 \begin{proof}
See Appendix~\ref{appendix:theorem1}.
\end{proof}

Theorem~\ref{theo:algorithm} states that, under certain conditions, $z^{(t)}_{f_{l,m} }$ in (\ref{eq:new_z}) evolves according to (\ref{eq:theorem_z}), $\hat{r}^{(t)}_{x_{n,p,m}}$ in (\ref{eq:lemma_r_x}) can be expressed in terms of $z^{(t)}_{f_{l,m}}$ and $\hat{x}^{(t)}_{n,p,m}$ according to (\ref{eq:theorem_r}). Therefore, for large $L$ and $N$, the computations of the parameters $z^{(t)}_{f_{l,m}}$, $\hat{r}^{(t)}_{x_{n,p,m}}$ and estimates $\theta^{(t+1)}_n$, $\hat{x}^{(t+1)}_{ n,p,m}$ no longer depend on the parameters in (\ref{eq:mean_Z_f_x}), (\ref{eq:mean_x_f_x}), and (\ref{eq:lemma_r_f_x}).

By Lemma~\ref{theo:tau} and Theorem~\ref{theo:algorithm}, i.e., substituting $\hat{\sigma}^{(t)}_{x_{n,p,m}}$ in (\ref{eq:theorem_sigma}) and $\hat{r}^{(t)}_{x_{n,p,m}}$ in (\ref{eq:theorem_r}) into (\ref{eq:lemma_lambda_npm}), (\ref{eq:lemma_x}), (\ref{eq:lemma_llr}), (\ref{eq:tau_m}), and (\ref{eq:theorem_z}), for large $L$ and $N$, we approximately have (ignoring $o(1)$):
	\begin{align} \nonumber
			&\theta^{(t+1)}_n = \log \bigg(\frac{\rho_n}{(1-\rho_n)}\\ \label{eq:theta}
			&\times \frac{\prod\limits_{m \in \mathcal{M}} \prod\limits_{p \in \mathcal{P}} f_{\mathcal{CN}}(0;\hat{x}^{(t)}_{n,p,m} + \sum\limits_{l \in \mathcal{L}}A^{*}_{l,n,p}  z^{(t)}_{f_{l,m}},\hat{\tau}^{(t)}_{m} + \beta_n)}{\prod\limits_{m \in \mathcal{M}} \prod\limits_{p \in \mathcal{P}} f_{\mathcal{CN}}(0;\hat{x}^{(t)}_{n,p,m} + \sum\limits_{l \in \mathcal{L}}A^{*}_{l,n,p}  z^{(t)}_{f_{l,m}},\hat{\tau}^{(t)}_{m})}\bigg),
		\end{align}
		\begin{align}  \label{eq:version1_lambda}
			& \lambda^{(t+1)}_{n,p,m} = \frac{\exp(\theta^{(t+1)}_n)}{\frac{f_{\mathcal{CN}}(0;  \hat{x}^{(t)}_{n,p,m} + \sum\limits_{l \in \mathcal{L}}A^{*}_{l,n,p}  z^{(t)}_{f_{l,m}},\hat{\tau}^{(t)}_{m} + \beta_n)}{f_{\mathcal{CN}}(0;   \hat{x}^{(t)}_{n,p,m} + \sum\limits_{l \in \mathcal{L}}A^{*}_{l,n,p}  z^{(t)}_{f_{l,m}},\hat{\tau}^{(t)}_{m})}+\exp(\theta^{(t+1)}_n)},\\ \label{eq:version1_theorem_x}
		&\hat{x}^{(t+1)}_{n,p,m} =  \eta_n(\hat{x}^{(t)}_{n,p,m} + \sum_{l \in \mathcal{L}}A^{*}_{l,n,p}  z^{(t)}_{f_{l,m}},\hat{\tau}^{(t)}_{m}, \lambda^{(t+1)}_{n,p,m}),\\  \nonumber 
		&z^{(t+1)}_{f_{l,m}} = y_{l,m} - \sum_{n \in \mathcal{N}}\sum_{p \in \mathcal{P}} A_{l,n,p} \hat{x}^{(t+1)}_{n,p,m} + \frac{1}{L}   z^{(t)}_{f_{l,m}}\\ \label{eq:version1_theorem_z} 
		&\times\sum_{n \in \mathcal{N}}\sum_{p \in \mathcal{P}}    \eta_n^{'}(\hat{x}^{(t)}_{n,p,m} + \sum_{l \in \mathcal{L}}A^*_{l,n,p}  z^{(t)}_{f_{l,m}}, \hat{\tau}^{(t)}_{m}, \lambda^{(t+1)}_{n,p,m}),\\ \nonumber
		&\hat{\tau}^{(t+1)}_{m} = \sigma^{2}+ \frac{\hat{\tau}^{(t)}_{m}}{L}\\ \label{eq:long_tau}
		&\times \sum_{n \in \mathcal{N}} \sum_{p \in \mathcal{P}} \eta_n^{'}(\hat{x}^{(t)}_{n,p,m} + \sum_{l \in \mathcal{L}}A^*_{l,n,p}  z^{(t)}_{f_{l,m}},\hat{\tau}^{(t)}_{m},\lambda^{(t+1)}_{n,p,m}).
	\end{align}
Furthermore, for large $L$ and $N$, we approximate $\hat{\tau}^{(t)}_{m}$ by \cite{SchniterAMP}:
\begin{align}\label{eq:tau}
	\hat{\tau}^{(t)}_{m} = \frac{1}{L} \sum_{l \in \mathcal{L}} |z^{(t)}_{f_{l,m}}|^2.
\end{align}
Therefore, we can iteratively update the estimates $\theta^{(t)}_n$ and $\hat{x}^{(t)}_{n,p,m}$ according to (\ref{eq:theta}) and (\ref{eq:version1_theorem_x}), respectively, based on the updates of the $(L+NP)M$ parameters $\lambda^{(t)}_{n,p,m}$, $z^{(t)}_{f_{l,m}}$, and $\hat{\tau}^{(t)}_{m}$ in (\ref{eq:version1_lambda}), (\ref{eq:version1_theorem_z}), and (\ref{eq:tau}), respectively. Let $\boldsymbol{\theta}^{(t)} \triangleq \left(\theta^{(t)}_n\right)_{n\in\mathcal{N}}$, $\hat{\mathbf{X}}^{(t)} \triangleq (\hat{x}^{(t)}_{n,p,m})_{ n\in\mathcal{N}, p \in \mathcal{P}, m \in\mathcal{M}}$, $\boldsymbol{\lambda}^{(t)}\triangleq (\lambda^{(t)}_{n,p,m})_{ n\in\mathcal{N}, p \in \mathcal{P}, m \in\mathcal{M}}$, $\mathbf{Z}^{(t)} \triangleq (z^{(t)}_{f_{l,m}})_{l\in\mathcal{L}, m \in \mathcal{M}}$, and $\hat{\boldsymbol{\tau}}^{(t)} \triangleq (\hat{\tau}^{(t)}_{m})_{ m \in\mathcal{M}}$. Their elements are computed in parallel. As a result, the computation cost can be significantly reduced. Note that the parameters $\mathbf{Z}^{(t)}$ and $\hat{\boldsymbol{\tau}}^{(t)}$  can be interpreted as residuals and residual variances, respectively\cite{Liu18TSP}.

\subsection{AMP-A-EC}\label{sec:al1}
In this part, we formally present the \emph{AMP-A-EC} algorithm based on the updates in (\ref{eq:theta})-(\ref{eq:version1_theorem_z}) and (\ref{eq:tau}). First, we initialize $\hat{\mathbf{X}}^{(0)}$ and $\mathbf{Z}^{(0)}$. 
Then, at each iteration, we update the estimates $\boldsymbol{\theta}^{(t+1)}$ and $\hat{\mathbf{X}}^{(t+1)}$ according to (\ref{eq:theta}) and (\ref{eq:version1_theorem_x}), respectively, and update the parameters $\boldsymbol{\lambda}^{(t)}$, $\mathbf{Z}^{(t+1)}$, and $\hat{\boldsymbol{\tau}}^{(t)}$ according to (\ref{eq:version1_lambda}), (\ref{eq:version1_theorem_z}), and (\ref{eq:tau}), respectively (c.f., Steps 1-5 of Algorithm \ref{algorithm1}). 
Note that this iteration procedure has no convergence guarantee when $L$ is smaller \textcolor{black}{than} or comparable to $N_a$ \cite{RanganAMPconvergence}, and the MSE $\mathbb{E}\left[\|\mathbf{X}-\hat{\mathbf{X}}^{(t)}\|^2_F |\mathbf{Y}\right]$ is not monotonically decreasing with $t$, i.e., $\hat{\mathbf{X}}^{(t)}$ may not be the best estimate among $\hat{\mathbf{X}}^{(1)},...,\hat{\mathbf{X}}^{(t)}$. This motivates us to track the best estimate obtained so far. Since $\mathbb{E}\left[\|\mathbf{X}-\hat{\mathbf{X}}^{(t)}\|^2_F |\mathbf{Y}\right]$ requires computing multiple integrals, which is as computationally expensive as computing the MMSE estimator given in (\ref{eq:MMSE_x_equation}), we use the objective function of GROUP LASSO problem:
\begin{align} \label{eq:gp}
	f(\mathbf{X}) \triangleq \frac{1}{2} \|\mathbf{Y} - \mathbf{A}\mathbf{X}\|_F^2 + \sum_{i=1}^{NP}\|\mathbf{X}_{i,:}\|_2
\end{align}
to measure the quality of $\hat{\mathbf{X}}^{(t)}$. Specifically, we keep tracking $f^{(t)}_{best} \triangleq \min \{f(\hat{\mathbf{X}}^{(1)}),...,f(\hat{\mathbf{X}}^{(t)})\}$ and obtain the estimate $\hat{\mathbf{X}}^{(t)}_{best} \triangleq \hat{\mathbf{X}}^{(t_{best})}$ with $t_{best}$ satisfying $f(\hat{\mathbf{X}}^{(t_{best})})=f^{(t)}_{best}$ and the corresponding $\boldsymbol{\theta}^{(t)}_{best} \triangleq \boldsymbol{\theta}^{(t_{best})}$ (c.f., Steps 6-11 of Algorithm \ref{algorithm1}). When some stopping criteria are satisfied, e.g., the preset number of iterations is reached \cite{Chen18TSP,Senel18TCOM,DonohoAMP}, or $|\frac{1}{M}\sum_{m\in\mathcal{M}} \hat{\tau}^{(t+1)}_{m} -\frac{1}{M}\sum_{m\in\mathcal{M}} \hat{\tau}^{(t)}_{m}|$ is sufficiently small \cite{Liu18TSP,GaoAMP_OFDM_TWC}, we can stop the \emph{AMP-A-EC} algorithm.
Finally, we detect device activities and estimate the active devices' channels based on $\boldsymbol{\theta}^{(t)}_{best}$ and $\hat{\mathbf{X}}^{(t)}_{best}$. Let $\hat{\mathbf{X}}^{(t)}_{best} \triangleq \left(\hat{x}^{best,(t)}_{n,p,m}\right)_{ n\in\mathcal{N}, p \in \mathcal{P}, m \in\mathcal{M}}$ and $\boldsymbol{\theta}^{(t)}_{best} \triangleq \left(\theta^{best,(t)}_n\right)_{n\in\mathcal{N}}$. Specifically, for all $n\in \mathcal{N}$, the activity detection for device $n$ is given by:
\begin{align} \label{eq:detector}
	\hat{a}^{best,(t)}_n=\begin{cases}
		0,\ \theta^{best,(t)}_n < 0,\\
		1,\ \theta^{best,(t)}_n \geq 0,\\
	\end{cases}
\end{align}
and the channel estimates of device $n$ which is detected to be active, i.e., $\hat{a}^{best,(t)}_n=1$, are given by:
\begin{equation} \label{eq:algo1_channel}
	\hat{h}_{n,p,m} = \hat{x}^{best,(t)}_{n,p,m}, p \in \mathcal{P}, m \in \mathcal{M}.
\end{equation}
The details of \emph{AMP-A-EC} are presented in Algorithm \ref{algorithm1}.
\begin{algorithm}[t] { \small\caption{AMP-A-EC} \label{algorithm1}
		{\bf Input:}
		received signals $\mathbf{Y}$ and measurement matrix $\mathbf{A}$.\\
		{\bf Output:}
		estimates of device activities $\hat{a}^{best,(t)}_n, n \in \mathcal{N}$ and active devices' channels $\hat{h}_{n,p,m}, n \in \mathcal{N}, p \in \mathcal{P}, m \in \mathcal{M}$, with $\hat{a}^{best,(t)}_n=1$.\\
		{\bf Initialize\footnotemark:} set $\hat{\mathbf{X}}^{(0)} = \mathbf{0}$,  $\hat{\mathbf{X}}^{(0)}_{best} = \hat{\mathbf{X}}^{(0)}$, $\mathbf{Z}^{(0)} = \mathbf{Y}$, $f^{(0)}_{best}= \frac{1}{2}\|\mathbf{Y}\|^2_F$, $t=0$.\\
		{\bf repeat}
		\begin{algorithmic}[1]
			\STATE \quad Calculate $\hat{\tau}^{(t)}_{m},m\in\mathcal{M}$ in parallel according to (\ref{eq:tau}).
			\STATE \quad Calculate $\theta^{(t+1)}_{n},n\in\mathcal{N}$ in parallel according to (\ref{eq:theta}).	
			\STATE \quad Calculate $\lambda^{(t+1)}_{n,p,m},n\in\mathcal{N},p\in\mathcal{P},m\in\mathcal{M}$ in parallel according to (\ref{eq:version1_lambda}).	
			\STATE \quad Calculate $\hat{x}^{(t+1)}_{n,p,m},n\in\mathcal{N},p\in\mathcal{P},m\in\mathcal{M}$ in parallel according to (\ref{eq:version1_theorem_x}).
			\STATE \quad Calculate $z^{(t+1)}_{f_{l,m}},l\in\mathcal{L},m\in\mathcal{M}$ in parallel according to (\ref{eq:version1_theorem_z}).
			\STATE \quad Calculate $f(\hat{\mathbf{X}}^{(t+1)})$ based on $\hat{x}^{(t+1)}_{n,p,m},n\in\mathcal{N},p\in\mathcal{P},m\in\mathcal{M}$ according to (\ref{eq:gp}).
			\STATE \quad \textbf{If} $f(\hat{\mathbf{X}}^{(t+1)}) < f^{(t)}_{best}$
			\STATE \quad\quad   set $f^{(t+1)}_{best}=f(\hat{\mathbf{X}}^{(t+1)})$, $\hat{\mathbf{X}}^{(t+1)}_{best}=\hat{\mathbf{X}}^{(t+1)}$, and $\boldsymbol{\theta}^{(t)}_{best} = \boldsymbol{\theta}^{(t)}$.
			\STATE \quad \textbf{else}
			\STATE \quad\quad   set $f^{(t+1)}_{best}=f^{(t)}_{best}$ and   $\hat{\mathbf{X}}^{(t+1)}_{best}=\hat{\mathbf{X}}^{(t)}$.
			\STATE \quad\textbf{end}
			\STATE \quad Set $t = t+1$.
		\end{algorithmic}
		{\bf until} some stopping criteria\footnotemark are satisfied.\\
		{\bf Calculate} $\hat{a}^{best,(t)}_n, n \in \mathcal{N}$ in parallel based on $\boldsymbol{\theta}^{(t)}_{best}$ according to (\ref{eq:detector}).\\
		{\bf Calculate} $\hat{h}_{n,p,m}, p \in \mathcal{P}, m \in \mathcal{M}$ in parallel based on $\hat{\mathbf{X}}^{(t+1)}_{best}$ according to (\ref{eq:algo1_channel}), with $\hat{a}^{best,(t)}_n=1, n \in \mathcal{N}$.
	}
\end{algorithm}
\footnotetext[4]{The initialization of $\hat{\mathbf{X}}^{(0)}$ and $\mathbf{Z}^{(0)}$ follow those in \cite{Chen18TSP,Liu18TSP,Senel18TCOM,GaoAMP_OFDM_TWC,DonohoAMP}.}

Now, we analyze the computational complexity of \emph{AMP-A-EC} via \textcolor{black}{floating-point operation (FLOP)} count.
The dominant terms of the per-iteration flop counts of Steps 1, 2, 3, 4, 5, and 6-11 of \emph{AMP-A-EC} are $2LM$, $10NPM$, $3NPM$,  $3NPM$, $2LMNP + 2LM^2 + 8NPM^2$, and $2NPM$, respectively. Therefore, the dominant term and \textcolor{black}{the} order of the per-iteration computational complexity of \emph{AMP-A-EC} are $2LMNP + 2LM^2 + 8NPM^2$ and $\mathcal{O}(LNPM)$, respectively.


\subsection{Performance Analysis}
In this part, we analyze the error probability of device activity detection and the MSE of channel estimation corresponding to Steps 1-5 in \emph{AMP-A-EC}, based on \textcolor{black}{SE} analysis\cite{MontanariSE}.

Given that matrix $\mathbf{A}$ in (\ref{eq:version1_theorem_x}) is a random matrix with all elements i.i.d. according to a complex Gaussian distribution, the arguments $\hat{x}^{(t)}_{n,p,m} + \sum_{l \in \mathcal{L}}A^{*}_{l,n,p}  z^{(t)}_{f_{l,m}}, p \in \mathcal{P},m \in \mathcal{M}$ in $\eta_n(\cdot)$ in (\ref{eq:version1_theorem_x}) converge to a random vector ${\mathbf{r}_t}\triangleq {\mathbf{x}} + \tau^{(t)} {\mathbf{v}}\in \mathbb{C}^{PM \times 1}$, where ${\mathbf{x}} $ denotes a random vector following the distribution $f_{\mathcal{B-CN}}(\mathbf{x};\mathbf{0},\beta \mathbf{I},\rho)$, and $\mathbf{v}$ denotes a vector independent of ${\mathbf{x}}$ and following the distribution $f_{\mathcal{CN}}(\mathbf{v};\mathbf{0},\mathbf{I})$, and 
$\tau^{(t)}$  evolves according to \cite[Theorem 1]{Liu18TSP}:
\begin{align} \label{eq:se}
	\tau^{(t+1)}=\sigma^{2}+\frac{NP}{L}\left( \frac{\rho \beta \tau^{(t)}}{\beta+\tau^{(t)}}+ \phi_{t} \left(\tau^{(t)}\right)\right),
\end{align}	
with $\tau^{(0)}=\sigma^{2}+\frac{NP}{L} \rho \beta$, and
\begin{align} \nonumber
	&\phi_{t} \left(\tau^{(t)}\right)=\frac{1}{MP} \mathbb{E}_{\mathbf{r}_t}\left[ \frac{g(\mathbf{r}_t,\tau^{(t)} \mathbf{I})\left(1-g(\mathbf{r}_t,\tau^{(t)}\mathbf{I})\right)\beta^{2}}{\left(\beta+\tau^{(t)}\right)^{2}}\mathbf{r}_t^{H} \mathbf{r}_t\right],\\
	&g(\mathbf{r}_t,\tau^{(t)} \mathbf{I}) \triangleq \frac{1}{1+ \frac{(1-\rho)f_{\mathcal{CN}}(0;\mathbf{r}_t,\tau^{(t)} \mathbf{I})}{\rho f_{\mathcal{CN}}(0;\mathbf{r}_t,\left(\tau^{(t)} + \beta\right) \mathbf{I})}}.
\end{align}
Here, $\tau^{(t)}$ is referred to as \textcolor{black}{the} ``state'', and (\ref{eq:se}) is called \textcolor{black}{the} state evolution. Besides, the parameters $\hat{\tau}_m^{(t)}, m \in \mathcal{M}$ in (\ref{eq:tau}) converge to $\tau^{(t)}$, as $L,N \rightarrow \infty$ \cite{MontanariSE}.
Based on the above state evolution results, we first analyze the error probability of the activity detection for device $n \in \mathcal{N}$  at the $t$-th iteration, denoted by
		$p_{n}^{\mathrm{e},(t)} \triangleq {\rm Pr}(\hat{a}^{(t)}_n \neq a_n),
$
where $\hat{a}^{(t)}_n$ is given by (\ref{eq:a_detect}), with $\theta^{(t)}_n$ given by (\ref{eq:theta}). Define ${\mathbf{h}}_{n} \triangleq ({h}_{n,p,m})_{p\in \mathcal{P}, m \in \mathcal{M}}$. 
\begin{theorem} \label{theo:error_prob}
For all $t$, suppose that $(\hat{x}^{(t)}_{n,p,m} + \sum_{l \in \mathcal{L}}A^{*}_{l,n,p}  z^{(t)}_{f_{l,m}})_{p \in \mathcal{P},m \in \mathcal{M}} = a_n \mathbf{h}_{n} + \tau^{(t)} \mathbf{v}$, and $\hat{\tau}_m^{(t)}= \tau^{(t)}$. Then, we have:
	\begin{align}\nonumber
	&p_{n}^{\mathrm{e}(t)} = \rho_n \frac{\underline{\Gamma}(PM, \frac{\tau^{(t)}}{\beta_n} \log \frac{1-\rho_n}{\rho_n} +b_{t, n}(\beta_n,\tau^{(t)}) PM)}{\Gamma(PM)} \\ \nonumber
	&+ (1-\rho_n) \frac{\bar{\Gamma}(PM, \frac{\tau^{(t)}+\beta_n}{\beta_n} \log \frac{1-\rho_n}{\rho_n} +c_{t, n}(\beta_n,\tau^{(t)}) PM)}{\Gamma(PM)},
\end{align}
	where $\bar{\Gamma}(\cdot, \cdot)$, $\underline{\Gamma}(\cdot, \cdot)$, and $\Gamma(\cdot)$ denote the upper incomplete Gamma function, lower incomplete Gamma function, and Gamma function, respectively, $\tau^{(t)}$ is given in (\ref{eq:se}), $b_{t, n}(\beta_n,\tau^{(t)}) \triangleq\frac{\tau^{(t)}}{\beta_n} \log \left(1+\frac{\beta_n}{\tau^{(t)}}\right),$ and $c_{t, n}(\beta_n,\tau^{(t)}) \triangleq\frac{\beta_n+\tau^{(t)}}{\beta_n} \log \left(1+\frac{\beta_n}{\tau^{(t)}}\right).$
\end{theorem}
\begin{proof}
	See Appendix~\ref{appendix:theorem2}.
\end{proof}

Theorem~\ref{theo:error_prob} characterizes the error probability analytically in terms of $P$ and $M$, and $\tau^{(t)}$. Since $a>\log (1+a)>\frac{a}{1+a}$ for $a>0$, we have $b_{t, n}(\beta_n,\tau^{(t)})<1$ and $c_{t, n}(\beta_n,\tau^{(t)})>1$. By the property of Gamma functions and Theorem \ref{theo:error_prob}, we know that the error probability for the detection of each device activity decreases with $M$ and $P$, and $p_{n}^{\mathrm{e}(t)} \rightarrow 0$ as $M \rightarrow \infty$.

Next, for each active device $n\in\mathcal{N}$, we analyze the MSE of the channel estimation at the $t$-th iteration, $\frac{1}{PM}\mathbb{E}\left[\|\mathbf{h}_{n} - \hat{\mathbf{h}}^{(t)}_{n}\|^2_2\right]$, where $\hat{\mathbf{h}}^{(t)}_{n} \triangleq (\hat{x}^{(t)}_{n,p,m})_{p\in \mathcal{P}, m \in \mathcal{M}}$. 
\begin{theorem} \label{theorem:channel_error}
For all $t$, suppose that $(\hat{x}^{(t)}_{n,p,m} + \sum_{l \in \mathcal{L}}A^{*}_{l,n,p}  z^{(t)}_{f_{l,m}})_{p \in \mathcal{P},m \in \mathcal{M}} = \mathbf{h}_{n} + \tau^{(t)} \mathbf{v} \triangleq \hat{\mathbf{r}}^{(t)}_{n}$, and $\hat{\tau}_m^{(t)}= \tau^{(t)}$. Then, we have: 
	\begin{align} \nonumber
		&\frac{1}{PM}\mathbb{E}\left[\|\mathbf{h}_{n} - \hat{\mathbf{h}}^{(t)}_{n}\|^2_2\right]=\frac{\beta_n \tau^{(t)}}{\beta_n + \tau^{(t)}}\\ \nonumber
		 &+ \frac{1}{PM}\mathbb{E}_{\hat{\mathbf{r}}^{(t)}_{n}}\Big[\frac{(1 -g(\hat{\mathbf{r}}^{(t)}_{n},\tau^{(t)}))^2 \beta^2_n}{(\beta_n + \tau^{(t)})^2} (\hat{\mathbf{r}}^{(t)}_{n})^H\hat{\mathbf{r}}^{(t)}_{n}\Big].
	\end{align}
\end{theorem}
\begin{proof}
	See Appendix~\ref{appendix:h}.
\end{proof}

Theorem~\ref{theorem:channel_error} characterizes the MSE of the channel estimation analytically in terms of $\tau^{(t)}$ in (\ref{eq:se}). 

\section{AMP for Device Activity Detection and Actual Channel Estimation} \label{sec:2}
In this section, we present the \emph{AMP-A-AC} algorithm for activity detection and actual channel estimation, which effectively updates the estimates $\theta^{(t)}_n$ and $\hat{h}^{(t)}_{n,p,m}$ based on the approximation results in Section~\ref{sec:approx}.

\subsection{Simplifications of Parameter and Estimate Updates}
In this part, we simplify the parameter operations in (\ref{eq:mean_Z_f_x}), (\ref{eq:mean_x_f_x}), (\ref{eq:lemma_r_x}), and (\ref{eq:lemma_r_f_x}) to enable efficient updates of $\theta^{(t)}_n$ in (\ref{eq:lemma_llr}) and $\hat{h}^{(t)}_{n,p,m}$ in (\ref{eq:lemma_h}). First, we introduce \textcolor{black}{a new} parameter:
	\begin{align} \label{eq:new_z2}
		&\tilde{z}^{(t)}_{f_{l,m} }  \triangleq y_{l,m} - \sum_{n \in \mathcal{N}} \sum_{p \in \mathcal{P}} A_{l,n,p}\epsilon^{(t)}_{ x_{n,p,m} \rightarrow f_{l,m}} \hat{h}^{(t)}_{ x_{n,p,m} \rightarrow f_{l,m}},
	\end{align}
	where
	\begin{align}
		&\hat{h}^{(t)}_{ x_{n,p,m} \rightarrow f_{l,m}} \triangleq \frac{\beta_n \hat{r}^{(t-1)}_{x_{n,p,m} \rightarrow f_{l,m}}}{\hat{\sigma}^{(t-1)}_{x_{n,p,m}\rightarrow f_{l,m}} + \beta_n},\\
		\nonumber
		&\epsilon^{(t)}_{ x_{n,p,m} \rightarrow f_{l,m}}\\ \label{eq:epsilon_origin}
		&\triangleq \frac{1}{1+\frac{(1-\lambda^{(t)}_{n,p,m})f_{\mathcal{CN}}(0;\hat{r}^{(t-1)}_{x_{n,p,m} \rightarrow f_{l,m}},\hat{\sigma}^{(t-1)}_{x_{n,p,m}\rightarrow f_{l,m}})}{\lambda^{(t)}_{n,p,m}f_{\mathcal{CN}}(0;\hat{r}^{(t-1)}_{x_{n,p,m} \rightarrow f_{l,m}},\hat{\sigma}^{(t-1)}_{x_{n,p,m}\rightarrow f_{l,m}}+ \beta_n)}}.
	\end{align}
Then, by (\ref{eq:new_z2}), Lemma~\ref{lemma:consistent_message}, and Lemma~\ref{theo:tau}, we have Theorem~\ref{theo:algorithm2}.

\begin{theorem} \label{theo:algorithm2}
	For all $t$, suppose that (\ref{lemma2:eq}) holds, $\hat{\gamma}^{(t)}_{f_{l,m} \rightarrow x_{n,p,m}} = \tau^{(t)}_{l,m}$, $\hat{\sigma}^{(t)}_{x_{n,p,m}\rightarrow f_{l,m}} = \hat{\sigma}^{(t)}_{x_{n,p,m}}$, $\tau^{(t)}_{l,m} = \hat{\tau}^{(t)}_{m}$, $\hat{\sigma}^{(t)}_{x_{n,p,m}} = \hat{\tau}^{(t)}_{m}$, and
		\begin{align} \nonumber
			\hat{h}^{(t+1)}_{ x_{n,p,m} \rightarrow f_{l,m}} &= \hat{h}^{(t+1)}_{ n,p,m} + \frac{\beta_n (\hat{r}^{(t)}_{x_{n,p,m} \rightarrow f_{l,m}}-\hat{r}^{(t)}_{x_{n,p,m}})}{\hat{\sigma}^{(t)}_{x_{n,p,m}} + \beta_n}\\ \label{eq:taylor_h}
			&+\mathcal{O}\left(\frac{1}{L}\right)+\mathcal{O}\left(\frac{1}{N}\right), \text{as $L,N \rightarrow \infty$}
		\end{align}
		hold, where $\mathcal{O}(\cdot)$ is uniformly in $(l,n)$, and
		\begin{align} \nonumber
			&\sum_{n \in \mathcal{N}} \sum_{p \in \mathcal{P}} \frac{\beta_n |A_{l,n,p}|^2}{\hat{\sigma}^{(t)}_{x_{n,p,m}} + \beta_n} 
			\\ \label{eq:appro_A_h}
			&= \frac{1}{L}\sum_{n \in \mathcal{N}} \sum_{p \in \mathcal{P}} \frac{\beta_n}{\hat{\sigma}^{(t)}_{x_{n,p,m}} + \beta_n}+ \mathcal{O}\left(\frac{1}{L}\right),\text{as $L \rightarrow \infty$}
		\end{align}
		holds, where $\mathcal{O}(\cdot)$ is uniformly in $l$. Then, we have:
		\begin{align}\nonumber
			&\tilde{z}^{(t)}_{f_{l,m}} = y_{l,m} - \sum_{n \in \mathcal{N}}\sum_{p \in \mathcal{P}} A_{l,n,p} \epsilon^{(t)}_{ x_{n,p,m} \rightarrow f_{l,m}}\hat{h}^{(t)}_{ n,p,m}  + \frac{1}{L}   \tilde{z}^{(t-1)}_{f_{l,m}}  \\ \label{eq:theorem_z2}
			&\times \sum_{n \in \mathcal{N}}\sum_{p \in \mathcal{P}}  \frac{\epsilon^{(t)}_{ x_{n,p,m} \rightarrow f_{l,m}} \beta_n}{\hat{\tau}^{(t)}_{m} + \beta_n} + o(1), \text{as $L,N \rightarrow \infty$}, \\  \nonumber
			&\hat{r}^{(t)}_{x_{n,p,m}}=\sum_{l \in \mathcal{L}}A^{*}_{l,n,p}\left(\tilde{z}^{(t)}_{f_{l,m}} + A_{l,n,p}\epsilon^{(t)}_{ x_{n,p,m} \rightarrow f_{l,m}} \hat{h}^{(t)}_{n,p,m}\right)    \\ \label{eq:theorem_r2}
			&+ o(1), \text{as $L,N \rightarrow \infty$},\\ \nonumber
&\epsilon^{(t)}_{ x_{n,p,m} \rightarrow f_{l,m}} = \frac{1}{1+\frac{(1-\lambda^{(t)}_{n,p,m})f_{\mathcal{CN}}(0;\hat{r}^{(t-1)}_{x_{n,p,m}}+ o(1),\hat{\tau}^{(t)}_{m})}{\lambda^{(t)}_{n,p,m}f_{\mathcal{CN}}(0;\hat{r}^{(t-1)}_{x_{n,p,m} }+ o(1),\hat{\tau}^{(t)}_{m}+ \beta_n)}},\\ \label{eq:epsilon}
			&\text{as $L\rightarrow \infty$}.
		\end{align}
\end{theorem}
\begin{proof}
See Appendix~\ref{appendix:second}.
\end{proof}

Theorem~\ref{theo:algorithm2} states that, under certain conditions, the parameter $\tilde{z}^{(t)}_{f_{l,m} }$ evolves according to (\ref{eq:theorem_z2}), the parameter $\hat{r}^{(t)}_{x_{n,p,m}}$ in (\ref{eq:lemma_r_x}) can be expressed in terms of $\tilde{z}^{(t)}_{f_{l,m}}$ and estimate $\hat{h}^{(t)}_{ n,p,m}$ according to (\ref{eq:theorem_r2}). Therefore, the \textcolor{black}{calculations} of the parameters $\tilde{z}^{(t)}_{f_{l,m}}$, $\hat{r}^{(t)}_{x_{n,p,m}}$ and estimates $\theta^{(t+1)}_n$, $\hat{h}^{(t+1)}_{ n,p,m}$ no longer depend on the parameters in (\ref{eq:mean_Z_f_x}), (\ref{eq:mean_x_f_x}), and (\ref{eq:lemma_r_f_x}).
	
By Theorem~\ref{theo:algorithm2}, for large $L$ and $N$, we approximately have (ignoring $o(1)$):
\begin{align} \label{eq:lambda_n}
	\epsilon^{(t)}_{ x_{n,p,m} \rightarrow f_{l,m}} = \frac{\exp(\theta^{(t)}_n)}{1+\exp(\theta^{(t)}_n)} \triangleq \hat{\lambda}^{(t)}_n.
\end{align}	
Furthermore, by Lemma~\ref{theo:tau} and Theorem~\ref{theo:algorithm2}, i.e., substituting $\hat{\sigma}^{(t)}_{x_{n,p,m}}$ in (\ref{eq:theorem_sigma}) and $\hat{r}^{(t)}_{x_{n,p,m}}$ in (\ref{eq:theorem_r2}) into (\ref{eq:lemma_h}), (\ref{eq:lemma_llr}), (\ref{eq:theorem_z2}), and (\ref{eq:lambda_n}), for large $L$ and $N$, we approximately have (ignoring $o(1)$):

\begin{align}
\nonumber
&\theta^{(t+1)}_n = \log \bigg(\frac{\rho_n}{(1-\rho_n)}\\ \label{eq:theta2}
&\times \frac{\prod\limits_{m \in \mathcal{M}} \prod\limits_{p \in \mathcal{P}} f_{\mathcal{CN}}(0;\hat{\lambda}^{(t)}_n   \hat{h}^{(t)}_{n,p,m} + \sum\limits_{l \in \mathcal{L}}A^{*}_{l,n,p}  \tilde{z}^{(t)}_{f_{l,m}},\hat{\tau}^{(t)}_m + \beta_n)}{\prod\limits_{m \in \mathcal{M}} \prod\limits_{p \in \mathcal{P}} f_{\mathcal{CN}}(0;\hat{\lambda}^{(t)}_n \hat{h}^{(t)}_{n,p,m} + \sum\limits_{l \in \mathcal{L}}A^{*}_{l,n,p}  \tilde{z}^{(t)}_{f_{l,m}},\hat{\tau}^{(t)}_m)}\bigg),\\
\label{eq:version2_lambda}
& \hat{\lambda}^{(t+1)}_{n} = \frac{\exp(\theta^{(t+1)}_n)}{1+\exp(\theta^{(t+1)}_n)},\\  \label{eq:version2_theorem_x}
	&\hat{h}^{(t+1)}_{n,p,m}=  \frac{\beta_n( \hat{\lambda}^{(t)}_n\hat{h}^{(t)}_{n,p,m} + \sum_{l \in \mathcal{L}}A^{*}_{l,n,p}\tilde{z}^{(t)}_{f_{l,m}})}{\hat{\tau}^{(t)}_m + \beta_n},\\ \nonumber 
	&\tilde{z}^{(t+1)}_{f_{l,m}} = y_{l,m} - \sum_{n \in \mathcal{N}}\sum_{p \in \mathcal{P}}A_{l,n,p} \hat{\lambda}^{(t+1)}_{n}\hat{h}^{(t+1)}_{n,p,m}\\ \label{eq:version2_theorem_z}
	&+ \frac{1}{L} \tilde{z}^{(t)}_{f_{l,m}} \sum_{n \in \mathcal{N}}\sum_{p \in \mathcal{P}} \hat{\lambda}^{(t+1)}_{n} \frac{\beta_n}{\hat{\tau}^{(t)}_m + \beta_n} , 
\end{align}
\begin{align} \label{eq:tau2}
	&\hat{\tau}^{(t+1)}_{m} = \frac{1}{L} \sum_{l \in \mathcal{L}} |\tilde{z}^{(t+1)}_{f_{l,m}}|^2.
\end{align}
Therefore, we can iteratively update the estimates $\theta^{(t)}_n$ and $\hat{h}^{(t)}_{n,p,m}$, according to (\ref{eq:theta2}) and (\ref{eq:version2_theorem_x}), respectively, based on the updates of parameters $\hat{\lambda}^{(t)}_{n}$, $\tilde{z}^{(t)}_{f_{l,m}}$ and $\hat{\tau}^{(t)}_{m}$.  Let $\boldsymbol{\theta}^{(t)} \triangleq \left(\theta^{(t)}_n\right)_{n\in\mathcal{N}}$, $\hat{\boldsymbol{\lambda}}^{(t)}\triangleq (\hat{\lambda}^{(t)}_{n})_{ n\in\mathcal{N}}$, $\hat{\mathbf{H}}^{(t)} \triangleq (\hat{h}^{(t)}_{n,p,m})_{ n\in\mathcal{N}, p \in \mathcal{P}, m \in\mathcal{M}}$, $\tilde{\mathbf{Z}}^{(t)} \triangleq (\tilde{z}^{(t)}_{f_{l,m}})_{l\in\mathcal{L}, m \in \mathcal{M}}$, and $\hat{\boldsymbol{\tau}}^{(t)} \triangleq (\hat{\tau}^{(t)}_{m})_{ m \in\mathcal{M}}$. Their elements are computed in parallel. 

\subsection{AMP-A-AC}\label{sec:al2}
In this part, we formally present the \emph{AMP-A-AC} algorithm based on the updates in (\ref{eq:theta2})-(\ref{eq:tau2}). First, we initialize the estimates $\hat{\mathbf{H}}^{(0)}$ and parameters $\hat{\boldsymbol{\lambda}}^{(0)}$ and $\tilde{\mathbf{Z}}^{(0)}$. 
Then, at each iteration, we update the estimates $\boldsymbol{\theta}^{(t+1)}$ and $\hat{\mathbf{H}}^{(t+1)}$ according to (\ref{eq:theta2}) and (\ref{eq:version2_theorem_x}), respectively, and update the parameters $\hat{\boldsymbol{\lambda}}^{(t)}$, $\tilde{\mathbf{Z}}^{(t)}$, and $\hat{\boldsymbol{\tau}}^{(t)}$ according to (\ref{eq:version2_lambda}), (\ref{eq:version2_theorem_z}), and (\ref{eq:tau2}), respectively (c.f., Steps 1-6 of Algorithm \ref{algorithm2}). 
Let 
\begin{align} \label{eq:theo:a}
	\hat{a}^{(t)}_n\triangleq\begin{cases}
		0,\ \theta^{(t)}_n < 0,\\
		1,\ \theta^{(t)}_n \geq 0,\\
	\end{cases}
\end{align}
and $\hat{\mathbf{X}}^{(t)} \triangleq (\hat{a}^{(t)}_n \hat{h}^{(t)}_{n,p,m})_{ n\in\mathcal{N}, p \in \mathcal{P}, m \in\mathcal{M}}$.
Then, as in Section~\ref{sec:al1}, we track the best estimates $\hat{\mathbf{a}}^{(t)}_{best} \triangleq \left(\hat{a}^{best,(t)}_{n}\right)_{n\in\mathcal{N}}$ and $\hat{\mathbf{H}}^{(t)}_{best} \triangleq \left(\hat{h}^{best,(t)}_{n,p,m}\right)_{ n\in\mathcal{N}, p \in \mathcal{P}, m \in\mathcal{M}}$ via $f(\hat{\mathbf{X}}^{(t)})$ (c.f., Steps 7-12 of Algorithm \ref{algorithm2}).
When some stopping criteria are satisfied, we can stop the \emph{AMP-A-AC} algorithm.
For all $n\in \mathcal{N}$, the activity detection for device $n$ is given by $\hat{a}^{best,(t)}_{n}$. 
For all $n \in \mathcal{N}$ with $\hat{a}^{best,(t)}_{n}=1$, the channel estimates of device $n$ are given by $\hat{h}^{best,(t)}_{n,p,m}, p \in \mathcal{P}, m \in \mathcal{M}$.
The details of \emph{AMP-A-AC} are presented in Algorithm~\ref{algorithm2}.
\begin{algorithm}[t] {\small\caption{\emph{AMP-A-AC}} \label{algorithm2}
		{\bf Input:}
		received signals $\mathbf{Y}$ and measurement matrix $\mathbf{A}$.\\
		{\bf Output:}
		estimates of activities $\mathbf{a}^{(t)}_{best}$ and actual channels $\hat{\mathbf{H}}^{(t)}_{best}$\\
		{\bf Initialize:} set $\hat{\mathbf{a}}^{(0)} = \mathbf{0}$, $\hat{\mathbf{H}}^{(0)} = \mathbf{0}$, $\hat{\mathbf{a}}^{(0)}_{best} = \hat{\mathbf{a}}^{(0)}$, $\hat{\mathbf{H}}^{(0)}_{best}=\hat{\mathbf{H}}^{(0)}$, $\tilde{\mathbf{Z}}^{(0)} = \mathbf{Y}$, $f^{(t)}_{best}=\frac{1}{2}\|\mathbf{Y}\|^2_F$, $t=0$.\\
		{\bf repeat}
		\begin{algorithmic}[1]
			\STATE \quad Calculate $\hat{\tau}^{(t)}_{m},m\in\mathcal{M}$ in parallel according to (\ref{eq:tau2}).
			\STATE \quad Calculate $\theta^{(t+1)}_{n},n\in\mathcal{N}$ in parallel according to (\ref{eq:theta2}).			
			\STATE \quad Calculate $\hat{\lambda}^{(t+1)}_{n},n\in\mathcal{N}$ in parallel according to (\ref{eq:version2_lambda}).
			\STATE \quad Calculate $\hat{h}^{(t+1)}_{n,p,m},n\in\mathcal{N},p\in\mathcal{P},m\in\mathcal{M}$ in parallel according to (\ref{eq:version2_theorem_x}).
			\STATE \quad Calculate $\tilde{z}^{(t+1)}_{f_{l,m}},l\in\mathcal{L},m\in\mathcal{M}$ in parallel according to (\ref{eq:version2_theorem_z}).
			\STATE \quad Calculate $\hat{a}_n^{(t+1)},n\in\mathcal{N}$ in parallel according to (\ref{eq:theo:a}).	
			\STATE \quad Calculate $f(\hat{\mathbf{X}}^{(t+1)})$ based on $\hat{\mathbf{X}}^{(t)} \triangleq (\hat{a}^{(t)}_n \hat{h}^{(t)}_{n,p,m})_{ n\in\mathcal{N}, p \in \mathcal{P}, m \in\mathcal{M}}$ according to (\ref{eq:gp}).
			\STATE \quad \textbf{If} $f(\hat{\mathbf{X}}^{(t+1)}) < f^{(t)}_{best}$
			\STATE \quad\quad   set $f^{(t+1)}_{best}=f(\hat{\mathbf{X}}^{(t+1)})$, $\hat{\mathbf{a}}^{(t+1)}_{best}=\hat{\mathbf{a}}^{(t+1)}$, and $\hat{\mathbf{H}}^{(t+1)}_{best}=\hat{\mathbf{H}}^{(t+1)}$.
			\STATE \quad \textbf{else}
			\STATE \quad\quad   set $f^{(t+1)}_{best}=f^{(t)}_{best}$, $\hat{\mathbf{a}}^{(t+1)}_{best}=\hat{\mathbf{a}}^{(t)}$, and $\hat{\mathbf{H}}^{(t+1)}_{best}=\hat{\mathbf{H}}^{(t)}$.
			\STATE \quad\textbf{end}
			\STATE \quad Set $t = t+1$.
		\end{algorithmic}
		{\bf until} some stopping criteria are satisfied.}
\end{algorithm}


Now, we analyze the computational complexity of \emph{AMP-A-AC} via FLOP count.
	The dominant terms of the per-iteration flop counts of Steps 1, 2, 3, 4, 5, and 6-12 of \emph{AMP-A-AC} are $2LM$, $10NPM$, $2N$,  $3NPM$, $2LMNP + NPM^2$, and $2NPM$, respectively. Therefore, the
	dominant term and order of the per-iteration computational complexity of \emph{AMP-A-AC} are $2LMNP + NPM^2$ and $\mathcal{O}(LNPM)$, respectively.


\section{Comparisons}\label{secVII}
In this section, we compare the proposed \emph{AMP-A-EC} and \emph{AMP-A-AC} with  existing algorithms and their extensions \cite{Chen18TSP,Liu18TSP,Senel18TCOM,JiaTWC,GaoAMP_OFDM_TWC,OMP}. 

\subsection{Comparisons between AMP-A-EC and AMP-A-AC}\label{sec:complexity}
In this part, we compare \emph{AMP-A-EC} and \emph{AMP-A-AC}. 
Their main differences lie in the residual calculation and channel estimation. The residual definition and calculation of \emph{AMP-A-EC}, given in (\ref{eq:new_z}) and (\ref{eq:version1_theorem_z}), respectively, follow those of the classic AMP method \cite[Eq. (1.1)]{MontanariSE} and thus resemble those of existing AMP-based algorithms \cite{DonohoAMP,Chen18TSP,Liu18TSP,Senel18TCOM,GaoAMP_OFDM_TWC}, enabling the analysis based on the classic SE technique \cite{MontanariSE}.
In contrast, \emph{AMP-A-AC} employs a slightly different residual definition and calculation, given in (\ref{eq:new_z2}) and (\ref{eq:version2_theorem_z}), and hence does not follow the SE \cite[Eq. (1.4)]{MontanariSE}, making the analysis intractable.

Furthermore, these differences lead to different computation complexities for residual calculation and channel estimation. 
Compared to \emph{AMP-A-AC}, \emph{AMP-A-EC} additionally requires calculating $\eta_n(\cdot)$ and $\eta_n^{'}(\cdot)$. In addition, the $NPM$ parameters of \emph{AMP-A-EC}, $\lambda^{(t)}_{n,p,m}, n\in\mathcal{N}, p\in\mathcal{P}, m\in\mathcal{M}$ in (\ref{eq:version1_lambda}), play a role \textcolor{black}{similar to} the $N$ parameters of \emph{AMP-A-AC}, $\hat{\lambda}^{(t)}_n, n\in\mathcal{N}$ in (\ref{eq:version2_lambda}), but have more complex expressions and a larger number. 
As a result, \emph{AMP-A-EC} exhibits a higher computational complexity than \emph{AMP-A-AC}, even though their computational complexities in order are identical.

\subsection{Comparisons between Proposed and Existing AMP-based Algorithms} \label{sec:comb}
In this part, we compare \emph{AMP-A-EC} and \emph{AMP-A-AC} with \emph{AMP-FL-ext} \cite{Chen18TSP,Liu18TSP,Senel18TCOM} and \emph{AMP-FS} \cite{GaoAMP_OFDM_TWC}.
\emph{AMP-FL-ext} represents the extension of the AMP-based algorithms in \cite{Chen18TSP,Liu18TSP,Senel18TCOM} for activity detection and time-domain effective channel estimation in a narrowband system with flat fading to an OFDM-based wideband system with frequency-selective fading. \emph{AMP-FS} \cite{GaoAMP_OFDM_TWC} employs the AMP method for activity detection and frequency-domain effective channel estimation for each subcarrier.

\emph{AMP-A-EC} and \emph{AMP-A-AC} capture the dependence of $x_{n,p,m}, p\in \mathcal{P},m \in \mathcal{M}$, for all $n\in \mathcal{N}$, since it is obtained based on the exact factorization of the joint distribution in (\ref{eq:joint_distribution}), as illustrated in Section \ref{sec:3c}.
In contrast, \emph{AMP-FS} ignores the dependence of frequency-domain effective channels caused by the shared device activities and the underlying time-domain frequency-selective fading channels.
Specifically, for all $n \in \mathcal{N}$, the frequency-domain channels $\tilde{h}_{n,k,m} = \frac{1}{\sqrt{K}} \sum_{p=1}^{P}h_{n,p,m}e^{-\frac{j2\pi (k-1) (p-1)}{K}}, k\in \mathcal{K},m \in \mathcal{M}$ and frequency-domain effective channels $\tilde{x}_{n,k,m}=a_n \tilde{h}_{n,k,m}, k\in \mathcal{K},m \in \mathcal{M}$ of device $n$ are dependent, and \emph{AMP-FS} is obtained based on an approximate factorization $p(\tilde{\mathbf{X}})\approx \prod_{n \in \mathcal{N}}\prod_{k \in \mathcal{K}}\prod_{m \in \mathcal{M}} p(\tilde{x}_{n,k,m})$ assuming the independence of $\tilde{x}_{n,k,m}, k\in \mathcal{K},m \in \mathcal{M}$, where $\tilde{\mathbf{X}}\triangleq(\tilde{x}_{n,k,m})_{n \in \mathcal{N}, k\in \mathcal{K},m \in \mathcal{M}}$.
Besides, \emph{AMP-FL-ext} ignores the dependence of time-domain effective channels caused by the shared device activities. In particular, for all $n \in \mathcal{N}$, the time-domain 
effective channels $\mathbf{X}_{(n-1)P+p,1:M},p \in \mathcal{P}$ are dependent, and \emph{AMP-FL-ext} 
is obtained based on an approximate factorization $p(\mathbf{X})\approx \prod_{n \in \mathcal{N}}\prod_{p \in \mathcal{P}} p(\mathbf{X}_{(n-1)P+p,1:M})$ assuming the independence of $\mathbf{X}_{(n-1)P+p,1:M},p \in \mathcal{P}$.

In addition, \emph{AMP-A-EC} and \emph{AMP-A-AC} approximately obtain the MAP-based device activity detection $\hat{a}^{(t)}_n$ by replacing $\theta_{n}$ in $\hat{a}^{\star}_n(\mathbf{Y})$ in (\ref{eq:detector_true}) with $\theta^{(t)}_{n}$ in (\ref{eq:theta}) and (\ref{eq:theta2}), respectively, for all $n\in\mathcal{N}$. The relationship between $(\hat{a}^{\star}_n(\mathbf{Y}))_{n\in\mathcal{N}}$ and $(\hat{a}^{(t)}_n)_{n\in\mathcal{N}}$ has yet to be explicitly demonstrated. On the contrary, \emph{AMP-FL-ext} and \emph{AMP-FS} approximately obtain the MMSE-based effective channel estimation using the AMP method and then obtain each device's activity detection based on a log-likelihood ratio test and the corresponding channel estimation of each device that is detected to be active \cite{Chen18TSP,Liu18TSP,Senel18TCOM,GaoAMP_OFDM_TWC}. 
	
Finally, \emph{AMP-A-EC} and \emph{AMP-A-AC} are different from \emph{AMP-FL-ext} and \emph{AMP-FS} in the following aspect. First, \emph{AMP-A-EC} and \emph{AMP-A-AC}'s updates for device activity detection involve prior information on device activities, as shown in (\ref{eq:theta}) and (\ref{eq:theta2}), respectively, whereas \emph{AMP-FL-ext} and \emph{AMP-FS}’s updates do not. Second, \emph{AMP-A-EC} and \emph{AMP-A-AC} track the best estimates among all iterates and possibly achieve higher accuracies, whereas \emph{AMP-FL-ext} and \emph{AMP-FS} do not.

\subsection{Computational Complexity Comparisons} \label{sec:comc}
	In this part, we compare the per-iteration computational complexities of \emph{AMP-A-EC}, \emph{AMP-A-AC}, \emph{AMP-FL-ext}, \emph{AMP-FS}, \emph{ML-MMSE}, and \emph{OMP-ext.}, as shown in Table. \ref{tab:table_revised_complexity}. Here, \emph{ML-MMSE} \cite{JiaTWC} is the combination of the extension of MLE-based device activity detection for $Q=1$ in \cite{JiaTWC} to $Q>1$ and {\em MMSE}-based channel estimation for detected active devices. \emph{OMP-ext.} \cite{OMP} is the extension of OMP for activity detection and frequency-domain channel estimation with a single-antenna BS to a multi-antenna BS. From Table. \ref{tab:table_revised_complexity}, we can observe that the per-iteration computational complexities of \emph{AMP-A-EC}, \emph{AMP-A-AC}, and \emph{AMP-FL-ext}, all performing in the time-domain, are independent of $K$ and scale linearly with $P$, and those of \emph{AMP-FS} and \emph{OMP-ext.}, both working in the frequency-domain, are independent of $P$ and scale linearly with $K$.
	Moreover, \emph{AMP-A-EC}, \emph{AMP-A-AC}, and \emph{AMP-FL-ext} have lower per-iteration computational complexities than \emph{AMP-FS} and \emph{OMP-ext.} since $P\ll K$, and \emph{ML-MMSE} has the highest per-iteration computational complexity.
 \begin{table}[t]
	\caption{Per-Iteration Computational Complexity Comparisons. \label{tab:table_revised_complexity}}
	\centering
	\scriptsize
	\begin{tabular}{|c|c|c|}
		\hline
		\textbf{Algorithms} & \textbf{Computational Complexity} \\
		\hline
		\emph{ML-MMSE} \cite{JiaTWC}  & $\mathcal{O}(L^2NP+L^3)$ \\
		\hline
		\emph{AMP-FS} \cite{GaoAMP_OFDM_TWC}, \emph{OMP-ext.} \cite{OMP} & $\mathcal{O}(LNKM)$ \\
		\hline
		\emph{AMP-FL-ext.} \cite{Chen18TSP,Liu18TSP,Senel18TCOM}, \emph{AMP-A-EC}, \emph{AMP-A-AC}& $\mathcal{O}(LNPM)$\\
		\hline
	\end{tabular}
\end{table}

\begin{figure*}[!t]
	\begin{center}
		\subfigure[\scriptsize{Error probability at $L=96$.}\label{fig:iter_error_L96}]
		{\resizebox{5.0cm}{!}{\includegraphics{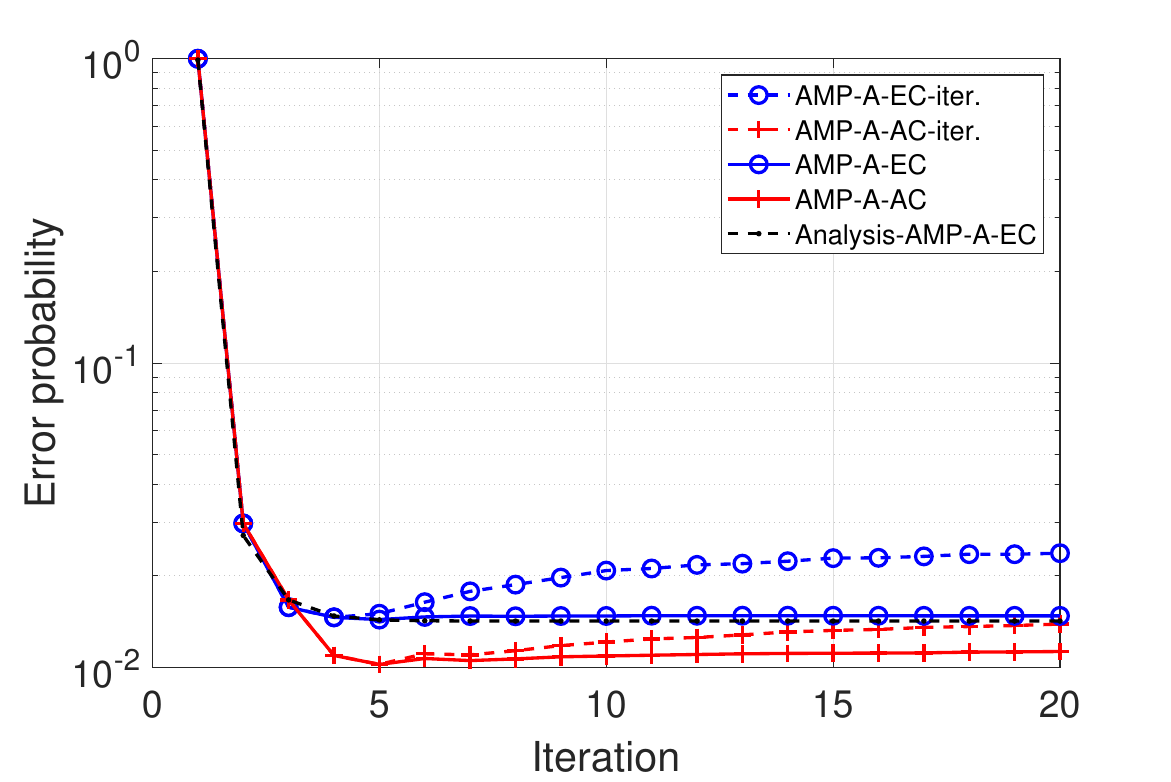}}}\quad
		\subfigure[\scriptsize{Error probability at $L=128$.}\label{fig:iter_error_L128}]
		{\resizebox{5.0cm}{!}{\includegraphics{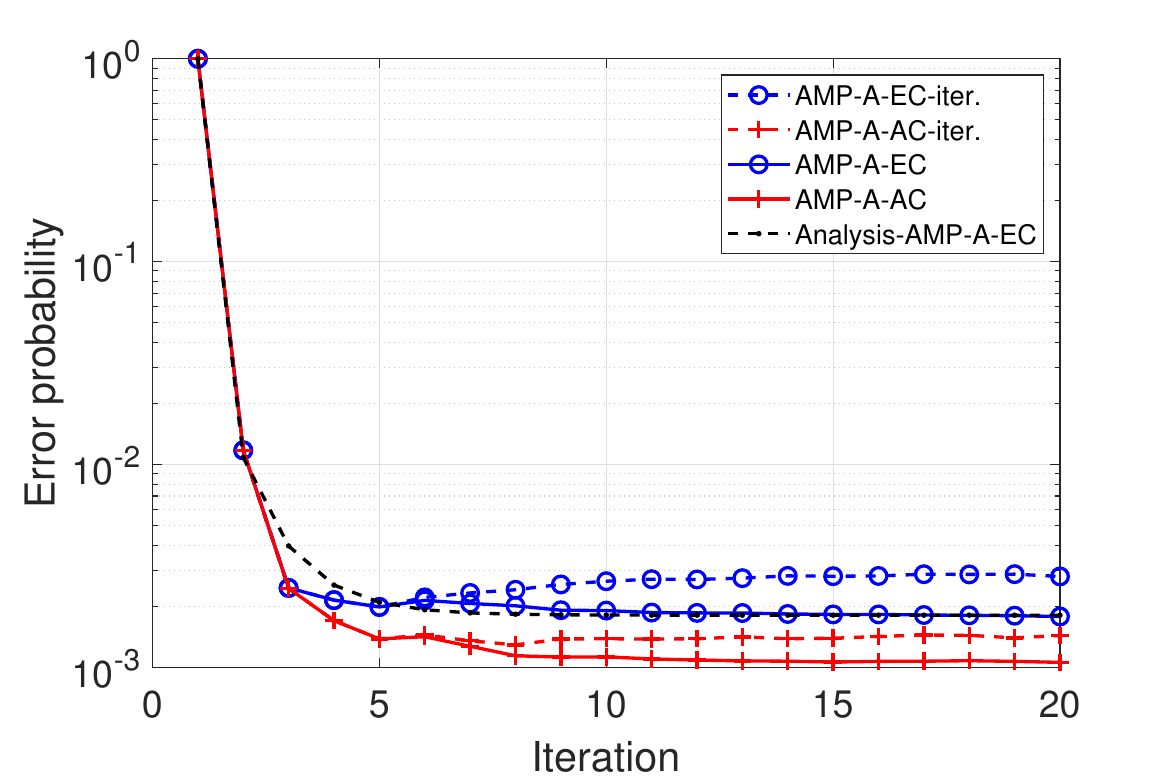}}}
		\subfigure[\scriptsize{Error probability at $L=192$.}\label{fig:iter_error_L192}]
		{\resizebox{5.0cm}{!}{\includegraphics{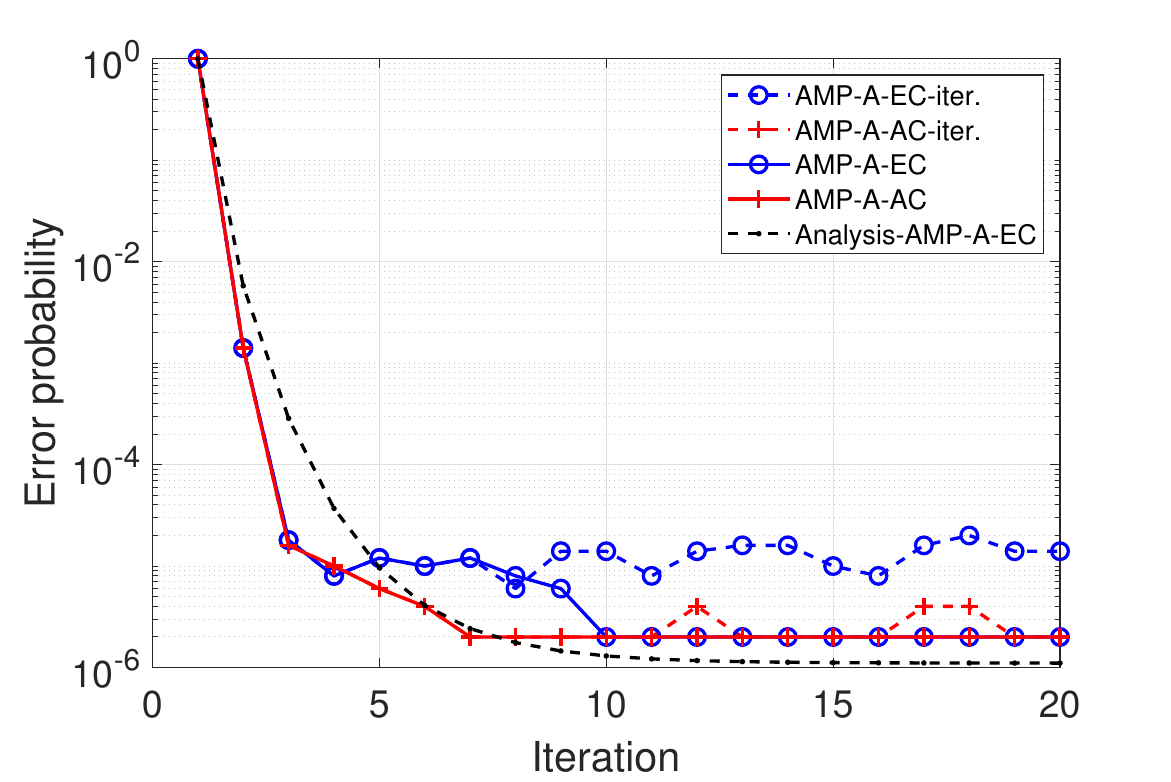}}}
	\end{center}
	\vspace{-3mm}
	\caption{\small{Error probability versus the iteration number at ${\color{black}P_t}=10$ dBm, $M=64$, $N=1000$, and $P=3$ for the c.d. case.}}
	\label{fig:iter_error}
\end{figure*}
\begin{figure*}[!t]
	\begin{center}
		\subfigure[\scriptsize{MSE at $L=96$.}\label{fig:iter_mse_L96}]
		{\resizebox{5.0cm}{!}{\includegraphics{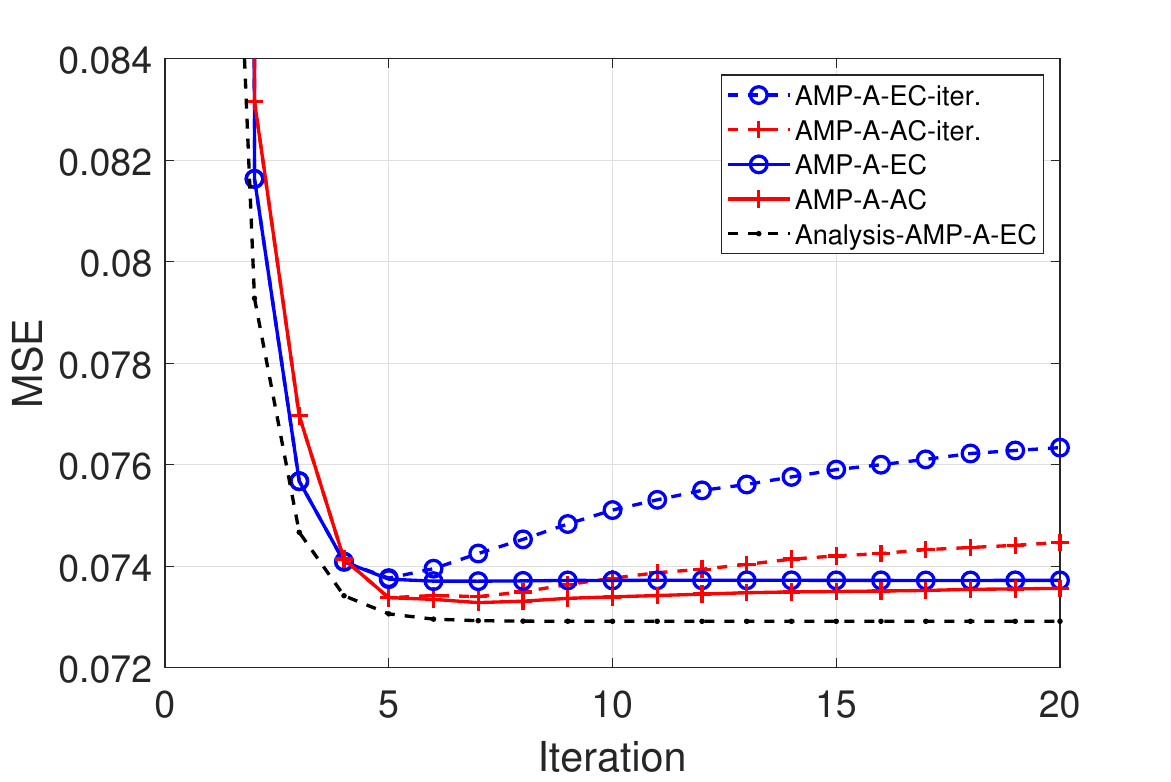}}}\quad
		\subfigure[\scriptsize{MSE at $L=128$.}\label{fig:iter_mse_L128}]
		{\resizebox{5.0cm}{!}{\includegraphics{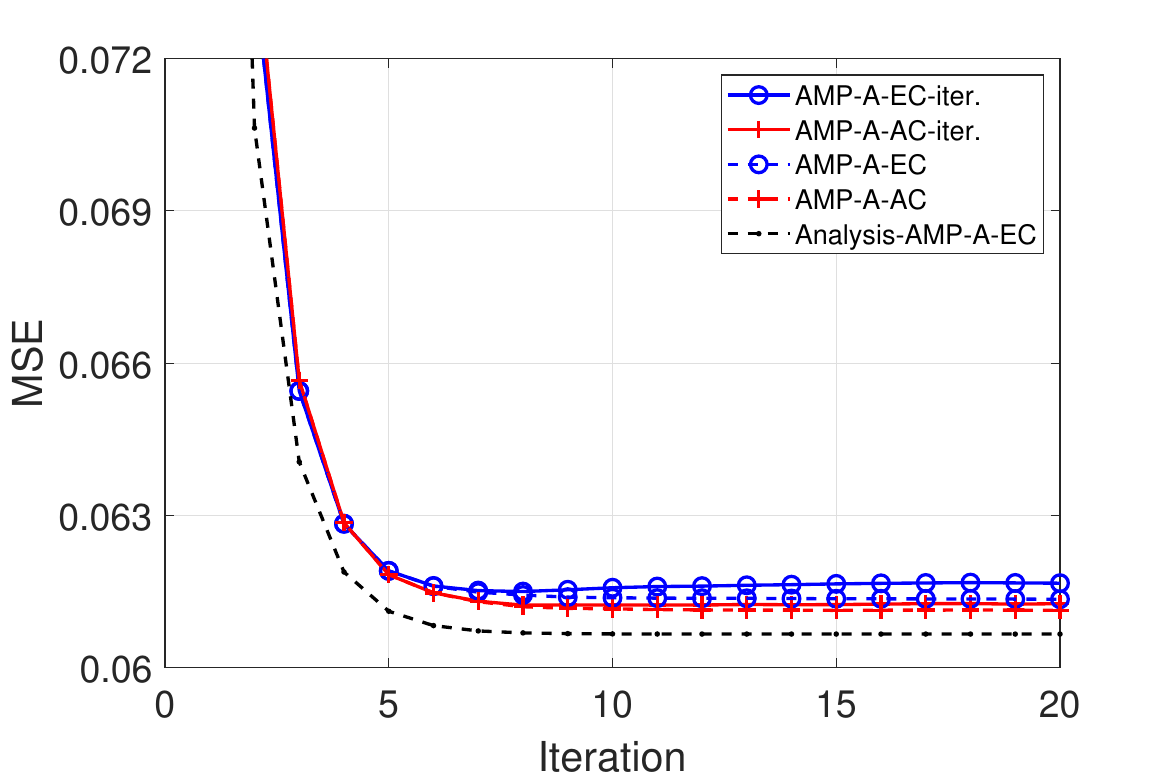}}}
		\subfigure[\scriptsize{MSE at $L=192$.}\label{fig:iter_mse_L192}]
		{\resizebox{5.0cm}{!}{\includegraphics{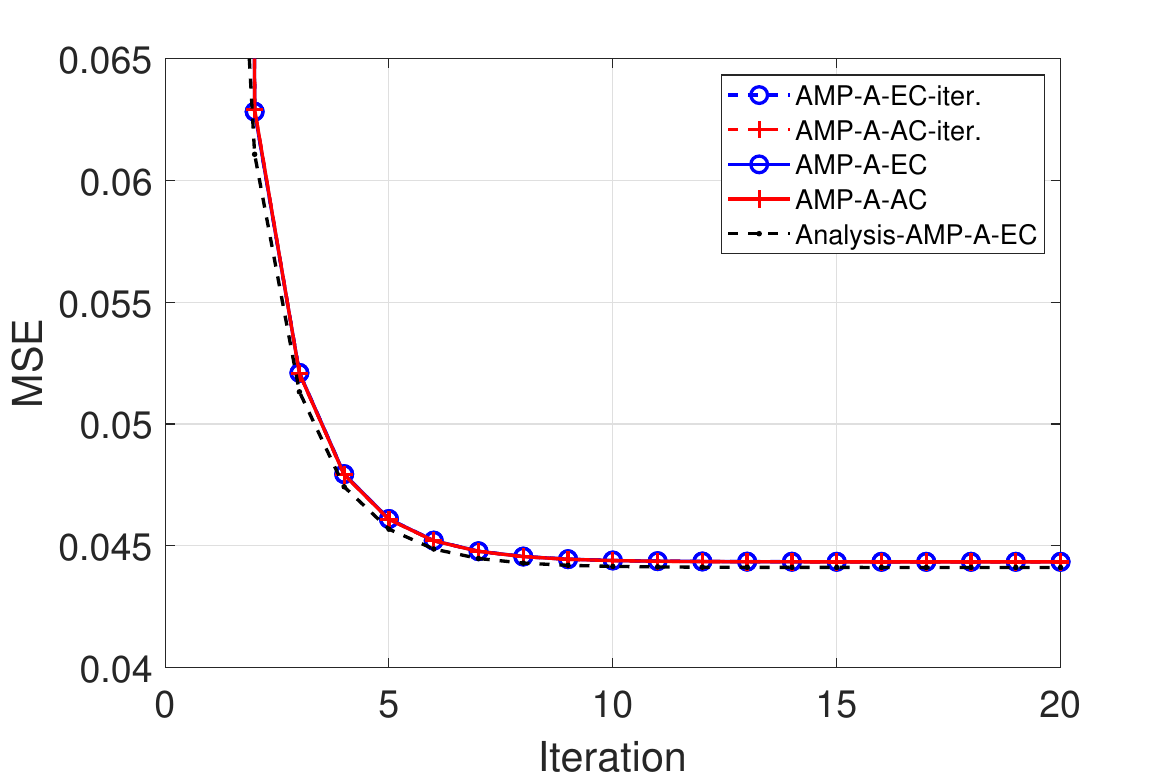}}}
	\end{center}
	\vspace{-3mm}
	\caption{\small{Error probability and MSE versus the iteration number at ${\color{black}P_t}=10$ dBm, $M=64$, $N=1000$, and $P=3$ for the c.d. case.}}
	\label{fig:iter_mse}
\end{figure*}

\section{Numerical results} \label{sec:result}
In this section, we evaluate the accuracies and computation time of \emph{AMP-A-EC} and \emph{AMP-A-AC}. We consider four baseline schemes, i.e., \emph{AMP-FL-ext} \cite{Liu18TSP}, \emph{AMP-FS} \cite{GaoAMP_OFDM_TWC},  \emph{OMP-ext.} \cite{OMP}, and \emph{ML-MMSE} \cite{JiaTWC}, which have been illustrated in Section \ref{sec:comb} and Section \ref{sec:comc}. Note that, \emph{AMP-A-EC}, \emph{AMP-A-AC}, \emph{AMP-FL-ext}, \emph{AMP-FS}, and \emph{ML-MMSE} stop at the 20-th iteration, and \emph{OMP-ext.} adopts the stopping criteria in \cite{OMP}.   

In the simulation, the noise power spectral density is set as -174 dBm/Hz. The number of subcarriers and the subcarrier spacing are set as $K=32$ and 30 kHz, respectively. Thus, the noise power is -114.18 dBm, i.e., $\sigma^2 = 10^{-11.418}$. We assume that large-scale fading is dominated by pathloss, and the pathloss is modeled as $\beta_n =10^{\left(-\eta \log_{10}\left(\frac {4\pi d_n}{\lambda}\right)\right)}$ \cite{TANsca}, where $d_n$ represents the distance between device $n$ and the BS, the pathloss exponent $\eta$ is 2.85, and the wavelength $\lambda$ is 0.086m.
We consider two distance cases: constant distance (c.d.) case and random distance (r.d.) case. In the c.d. case, $d_n=70$ m, $n \in \mathcal{N}$. In the r.d. case, $d_n, n \in \mathcal{N}$ are randomly distributed in [50, 100] m. We denote the transmitted power by \textcolor{black}{$P_t$}.
The average received signal-to-noise ratio (SNR) is given by $\frac{1}{N}\sum_{n \in \mathcal{N}} \frac{{\color{black}P_t} \beta_n }{\sigma^2}$. 

Furthermore, in the simulation, we adopt the following setup. We generate pilots $\tilde{\mathbf s}_{q,n}$, $q\in \mathcal{Q}, n \in \mathcal{N}$ according to i.i.d. $\mathcal{CN}(0,1)$ and normalize $\sum_{q \in \mathcal{Q}} \|\tilde{\mathbf s}_{q,n}\|^2_2$ to $K$ for all $n \in \mathcal{N}$ \cite{Chen18TSP,Liu18TSP,JiaTWC}.
We generate $500$ realizations for $\alpha_n,n\in\mathcal N$, $ h_{n,p,m}$, $n\in\mathcal N$, $p\in\mathcal P$, $m\in\mathcal M$, Gaussian pilots, and $d_n, n \in \mathcal{N}$ only for the r.d. case and evaluate the average error probability and MSE over all $500$ realizations. We set $\rho_n=\rho=0.1,n\in\mathcal{N}$. Unless otherwise stated, we choose ${\color{black}P_t}=10$ dBm, $L=128$, $M=64$, $N=1000$, and $P=3$.
Our simulation environment is MATLAB R2022b. For the hardware, the CPU is Intel i7-13700KF. 

\begin{figure*}[!t]
	\begin{center}
		\subfigure[\scriptsize{Error probability.}\label{fig:error_L}]
		{\resizebox{5.0cm}{!}{\includegraphics{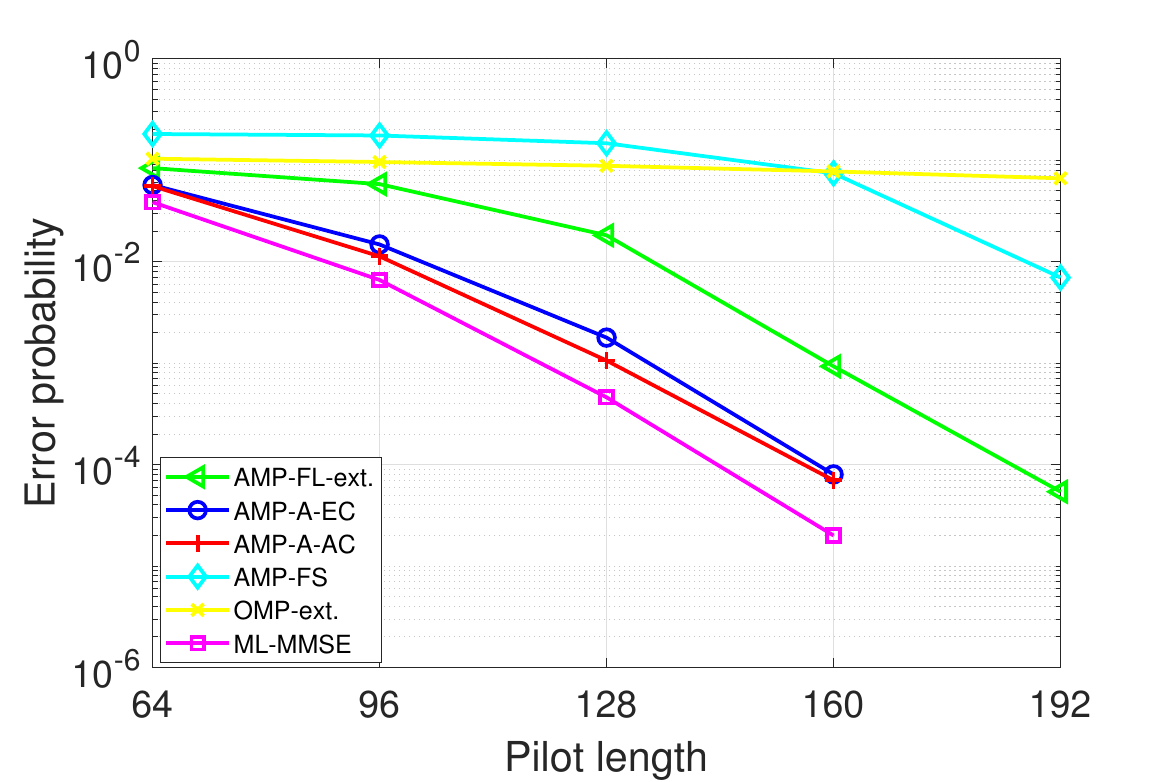}}}
		\subfigure[\scriptsize{False alarm probability.}\label{fig:FA_L}]
		{\resizebox{5.0cm}{!}{\includegraphics{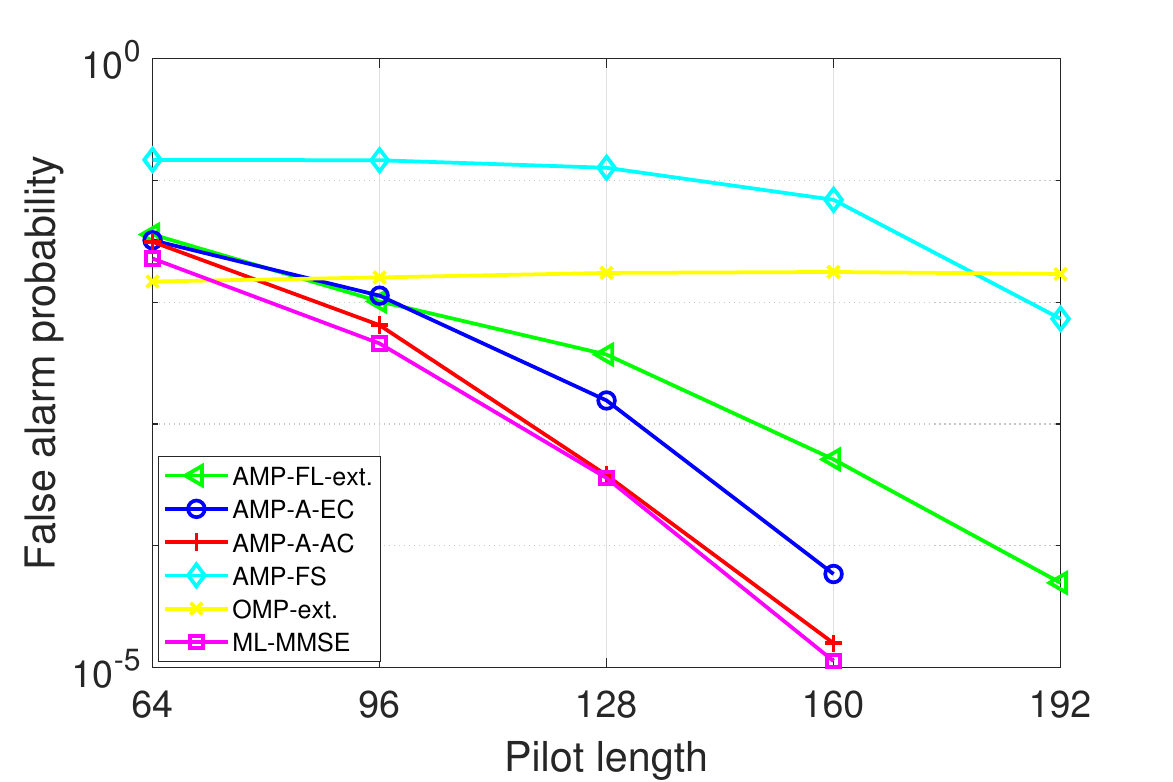}}}
		\subfigure[\scriptsize{Missed detection probability.}\label{fig:MD_L}]
		{\resizebox{5.0cm}{!}{\includegraphics{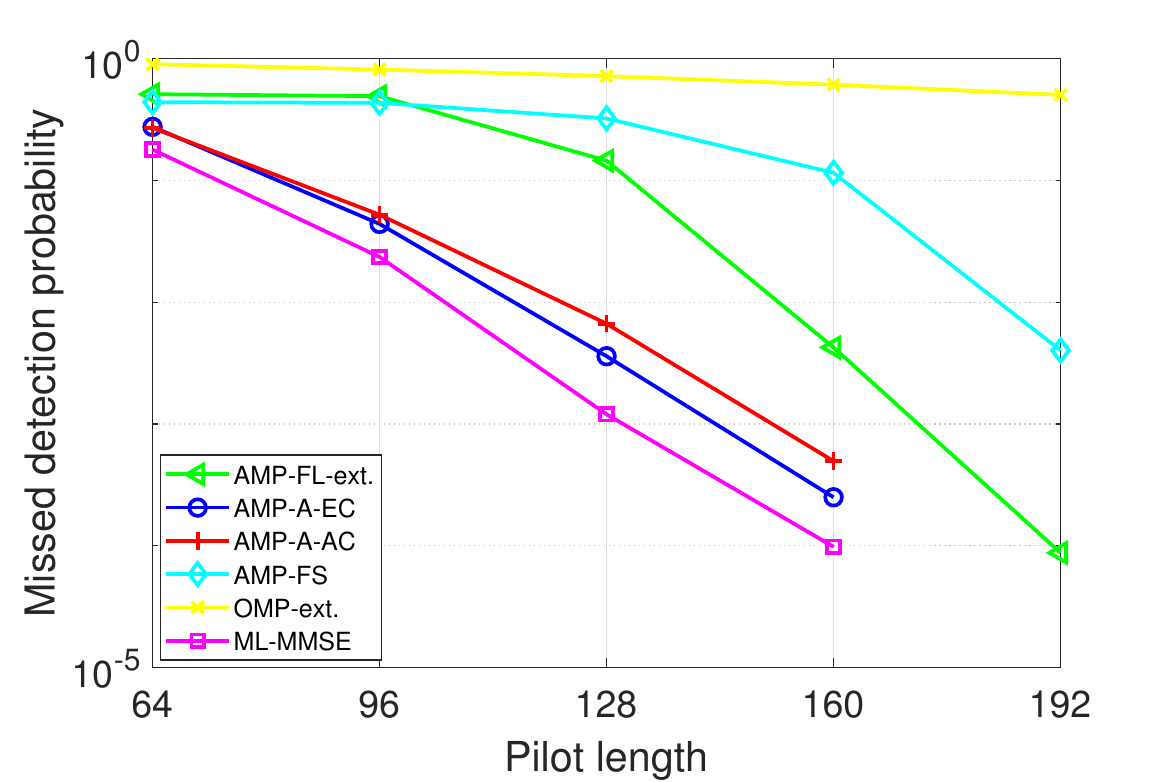}}}
		\subfigure[\scriptsize{MSE.}\label{fig:mse_L}]
		{\resizebox{5.0cm}{!}{\includegraphics{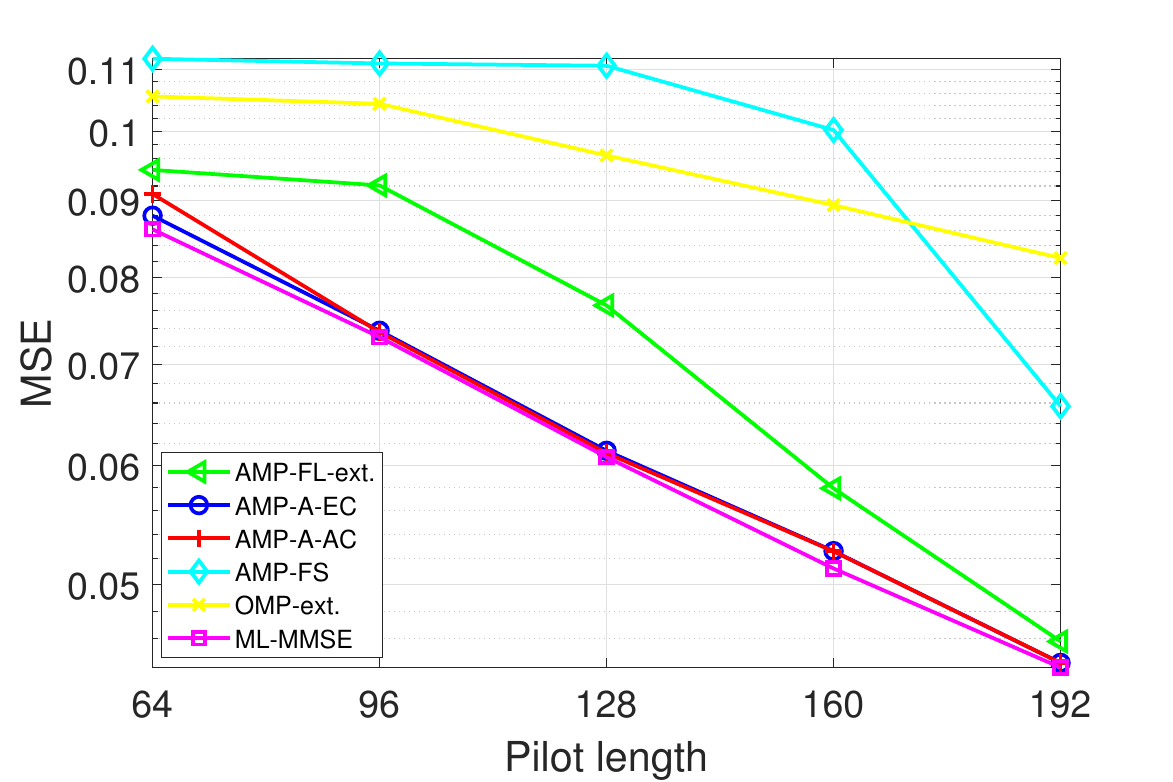}}}
		\subfigure[\scriptsize{Computation time.}\label{fig:time_L}]
		{\resizebox{5.0cm}{!}{\includegraphics{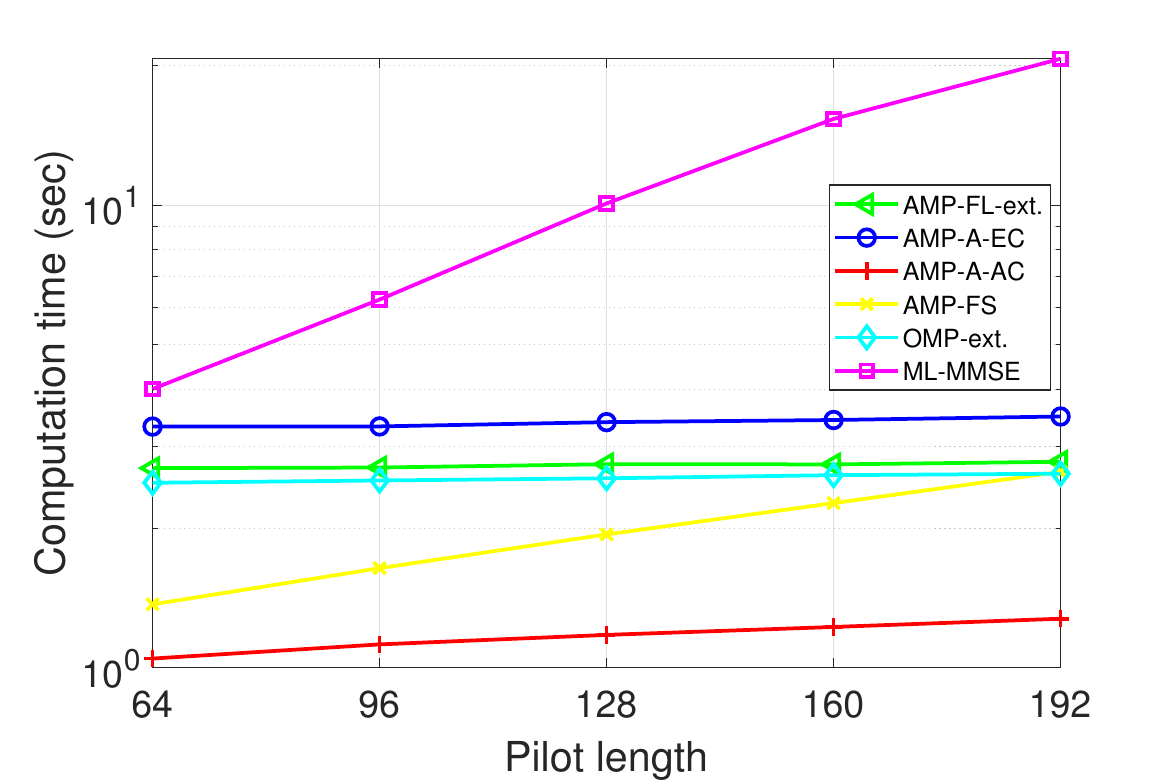}}}
	\end{center}
	\vspace{-5mm}
	\caption{\small{Error probability, false alarm probability, missed detection probability, MSE, and computation time versus pilot length at ${\color{black}P_t}=10$~dBm, $M=64$, $N=1000$, and $P=3$ for the c.d. case.}}
	\label{fig:variousL}
\end{figure*}
\begin{figure*}[!t]
	\begin{center}
		\subfigure[\scriptsize{Error probability.}\label{fig:error_M}]
		{\resizebox{5.0cm}{!}{\includegraphics{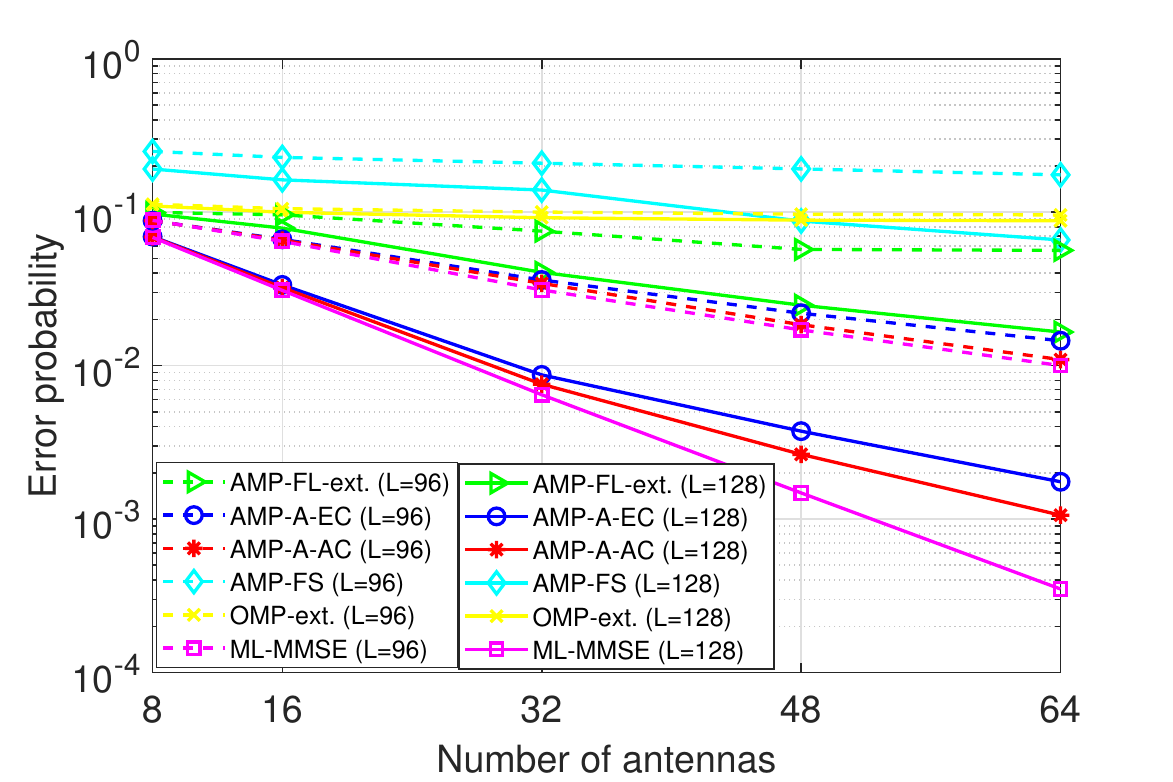}}}
		\subfigure[\scriptsize{False alarm probability.}\label{fig:FA_M}]
		{\resizebox{5.0cm}{!}{\includegraphics{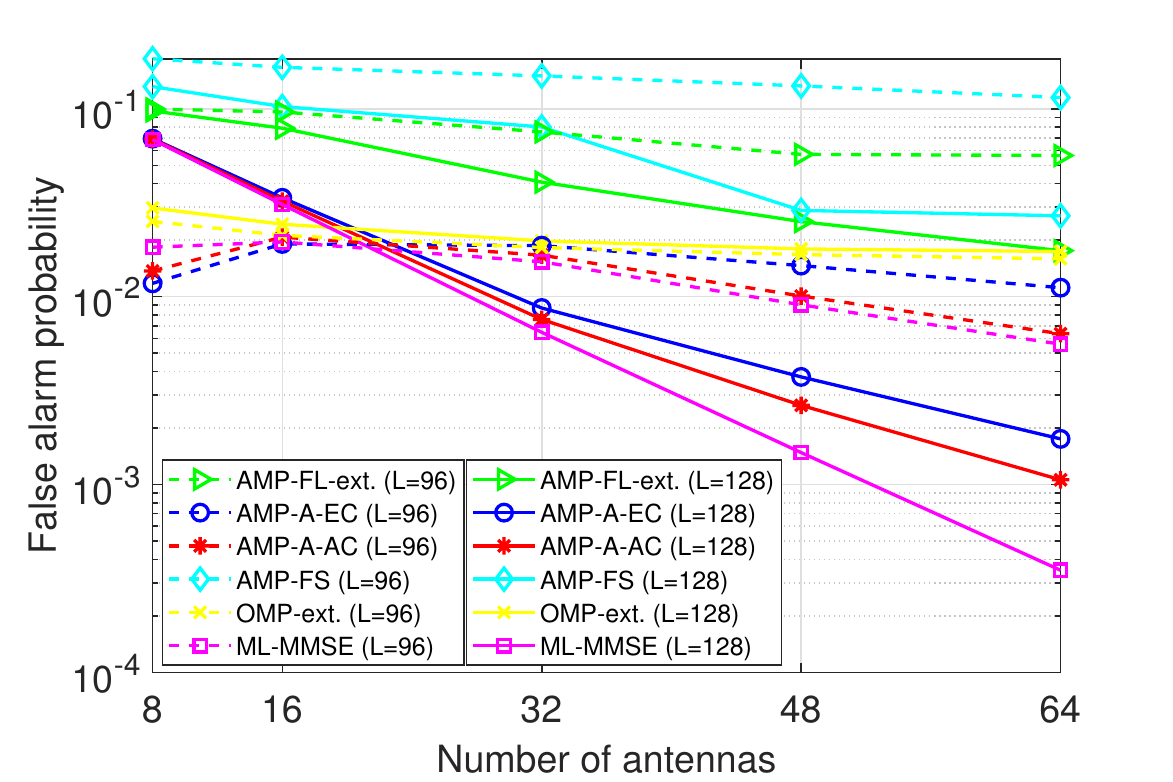}}}
		\subfigure[\scriptsize{Missed detection probability.}\label{fig:MD_M}]
		{\resizebox{5.0cm}{!}{\includegraphics{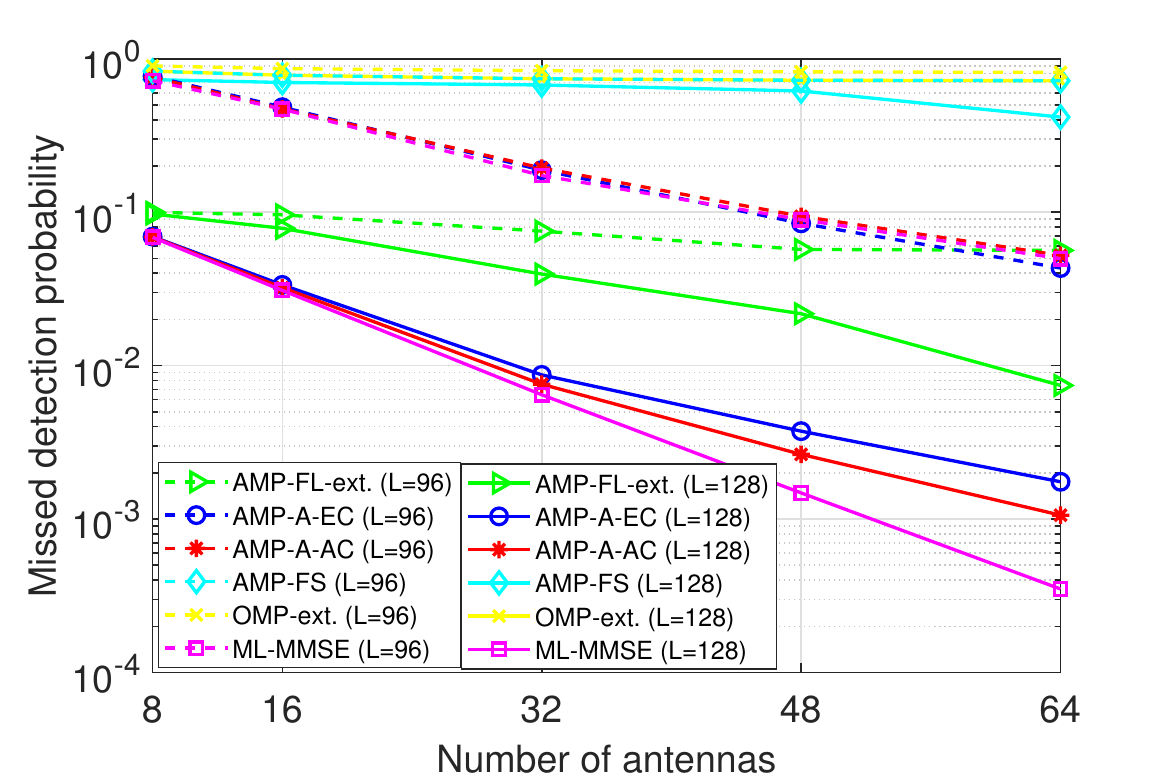}}}
		\subfigure[\scriptsize{MSE.}\label{fig:mse_M}]
		{\resizebox{5.0cm}{!}{\includegraphics{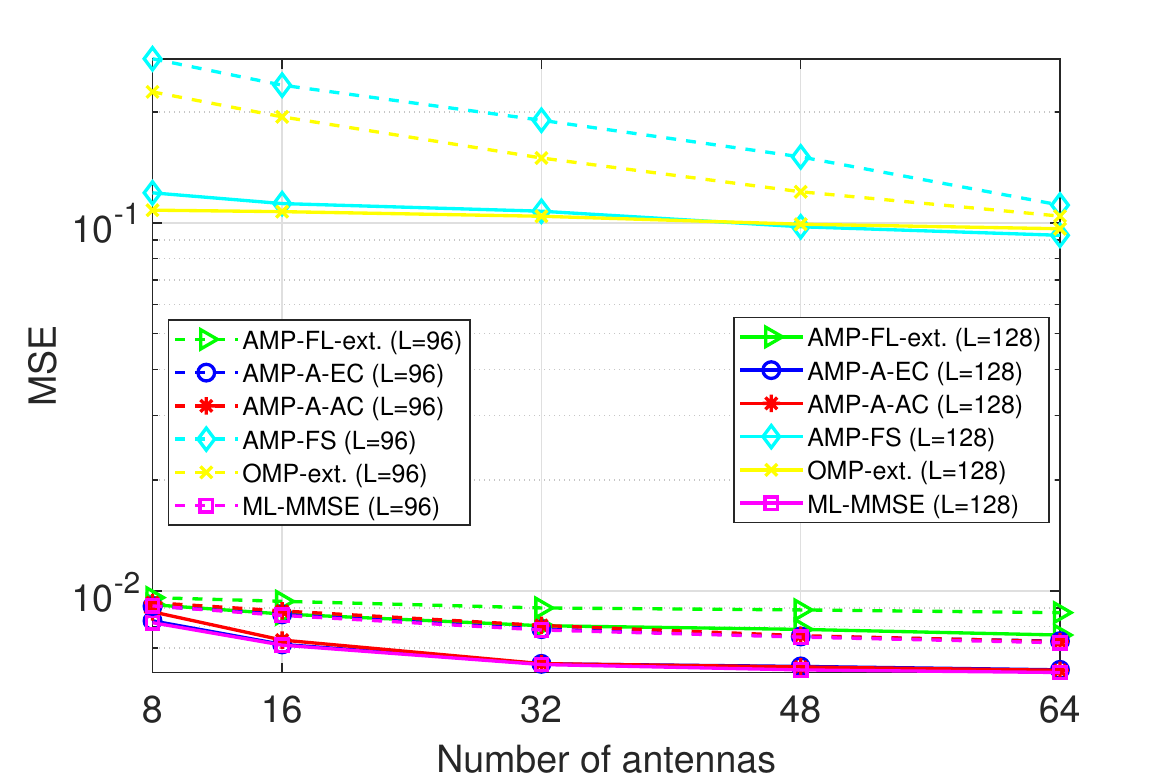}}}
		\subfigure[\scriptsize{Computation time.}\label{fig:time_M}]
		{\resizebox{5.0cm}{!}{\includegraphics{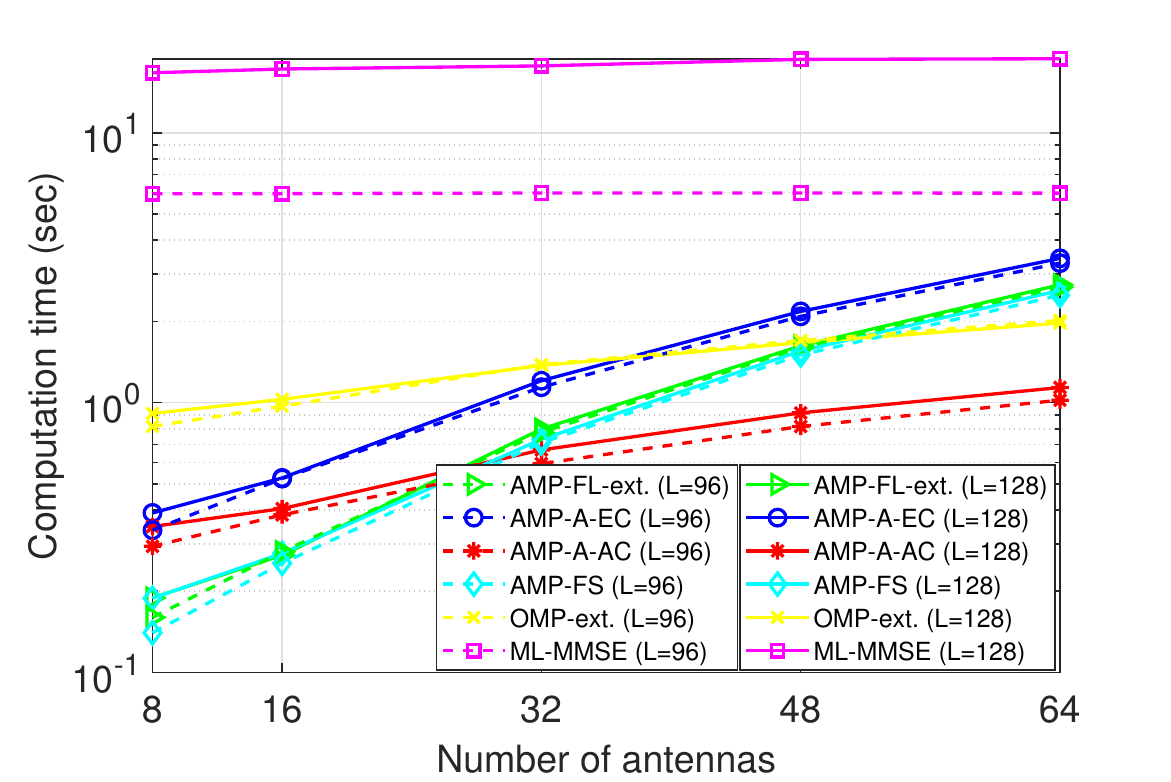}}}
	\end{center}
	\vspace{-5mm}
	\caption{\small{Error probability, false alarm probability, missed detection probability, MSE, and computation time versus the number of antennas at {\color{black}$L=96$ ($L=128$)}, ${\color{black}P_t}=10$ dBm, $N=1000$, and $P=3$ for the c.d. case.}}
	\label{fig:variousM}
\end{figure*}
\begin{figure*}[!t]
	\begin{center}
		\subfigure[\scriptsize{Error probability.}\label{fig:error_Tap}]
		{\resizebox{5.0cm}{!}{\includegraphics{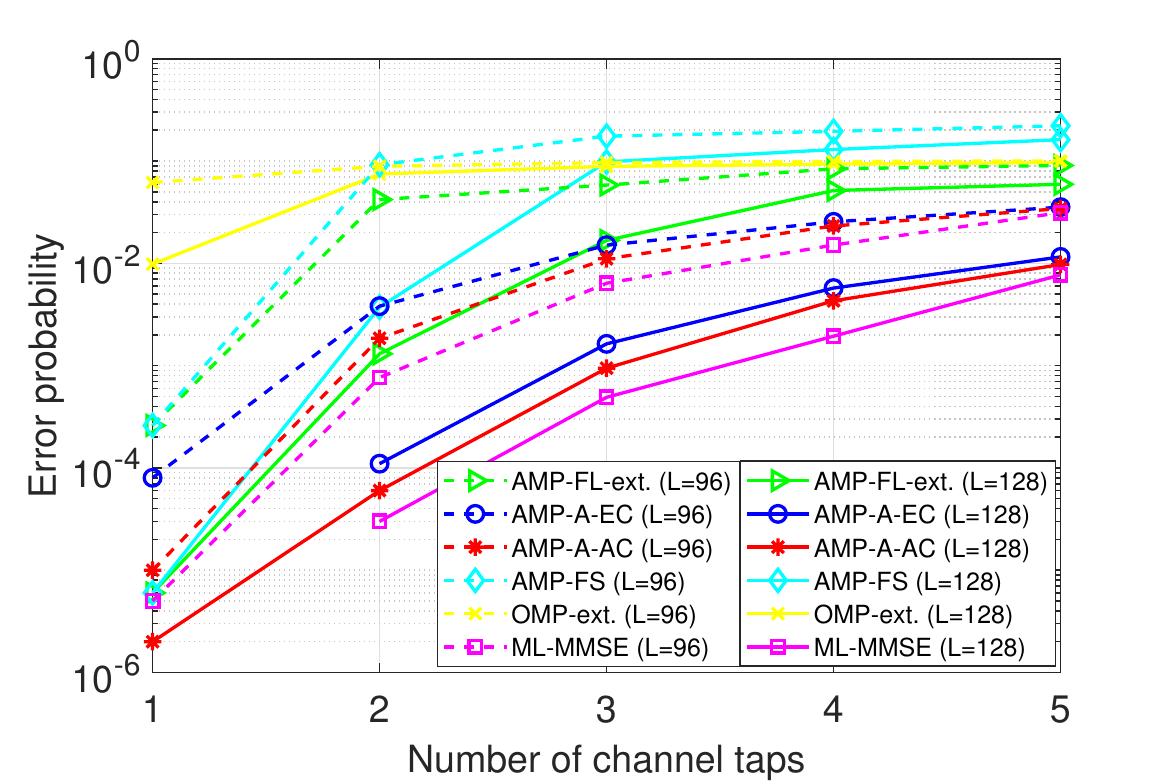}}}
		\subfigure[\scriptsize{False alarm probability.}\label{fig:FA_Tap}]
		{\resizebox{5.0cm}{!}{\includegraphics{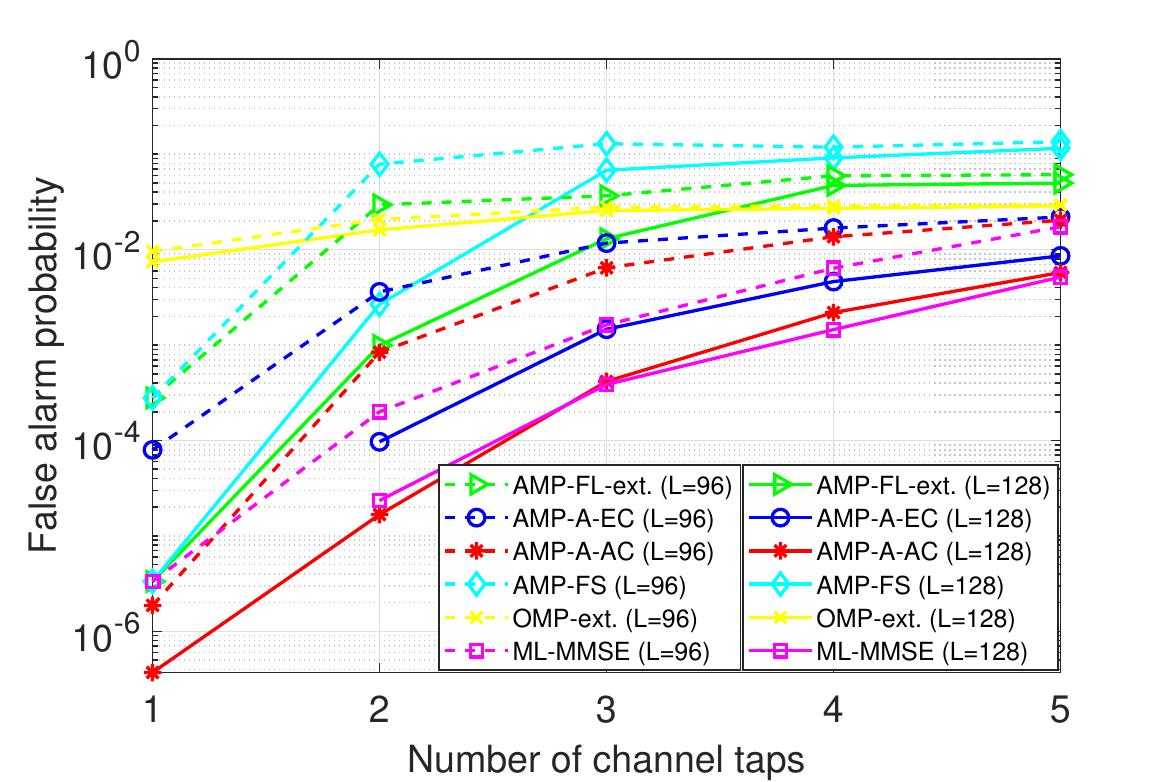}}}
		\subfigure[\scriptsize{Missed detection probability.}\label{fig:MD_Tap}]
		{\resizebox{5.0cm}{!}{\includegraphics{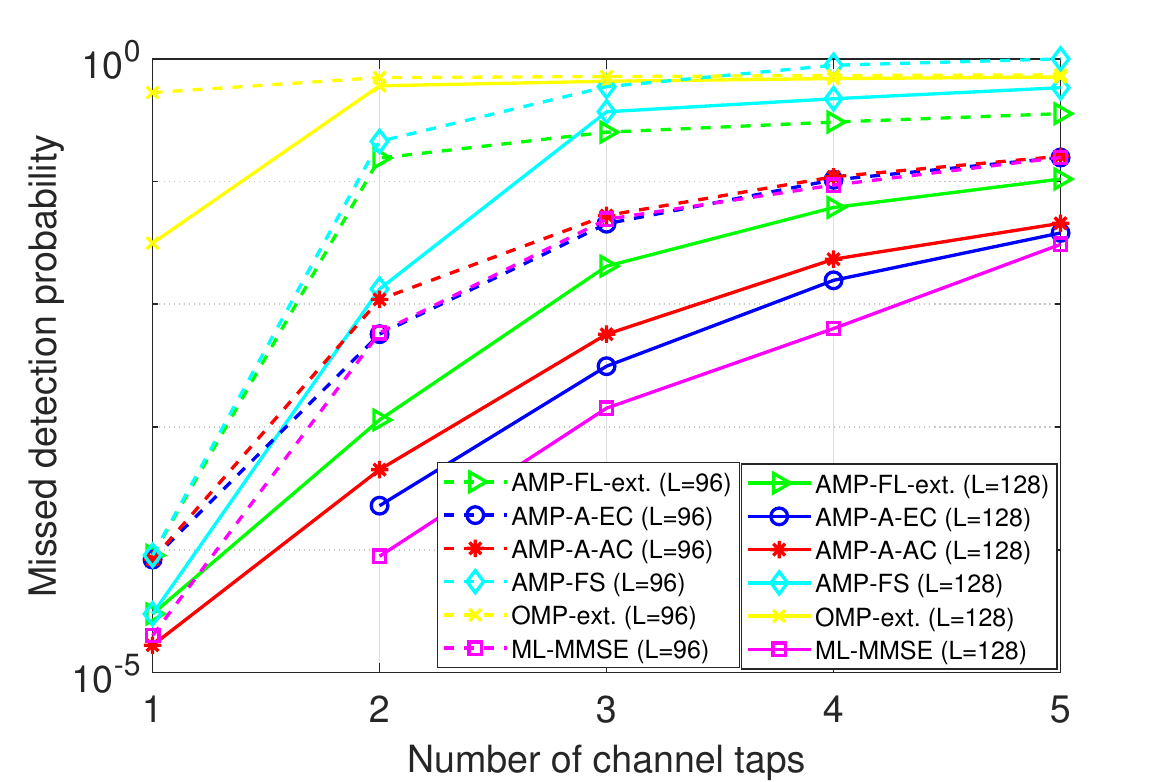}}}
		\subfigure[\scriptsize{MSE.}\label{fig:mse_Tap}]
		{\resizebox{5.0cm}{!}{\includegraphics{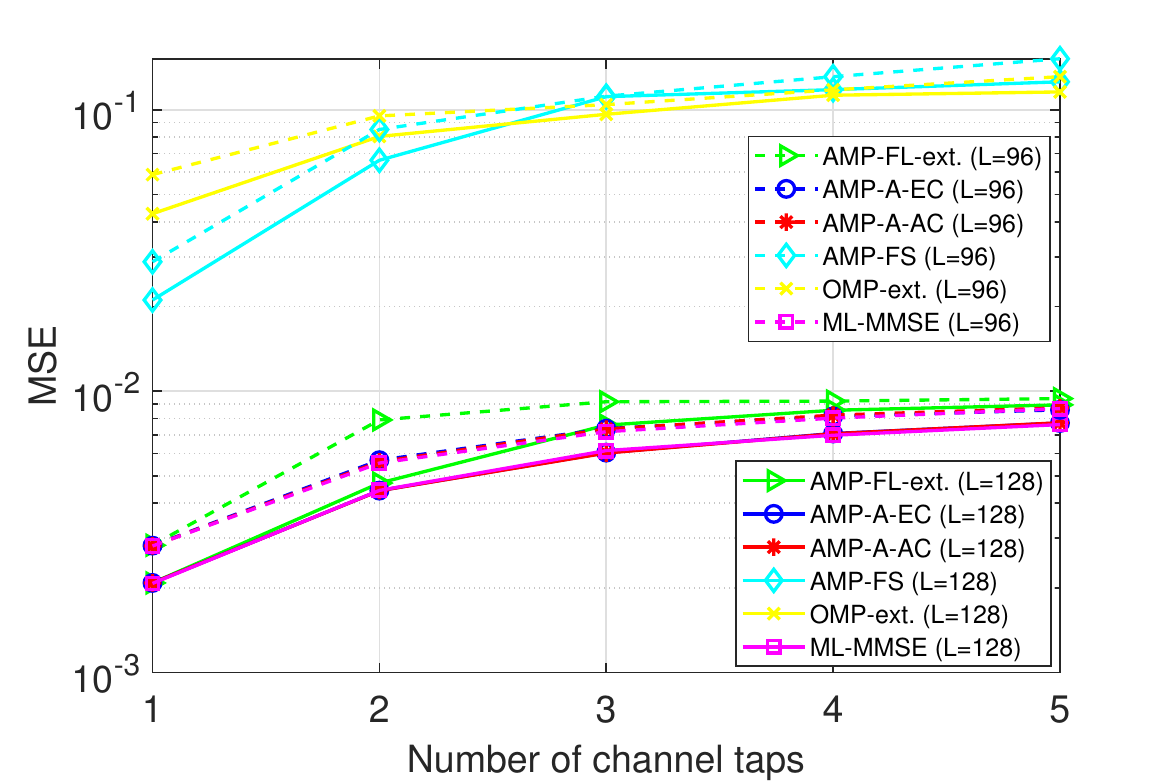}}}
		\subfigure[\scriptsize{Computation time.}\label{fig:time_Tap}]
		{\resizebox{5.0cm}{!}{\includegraphics{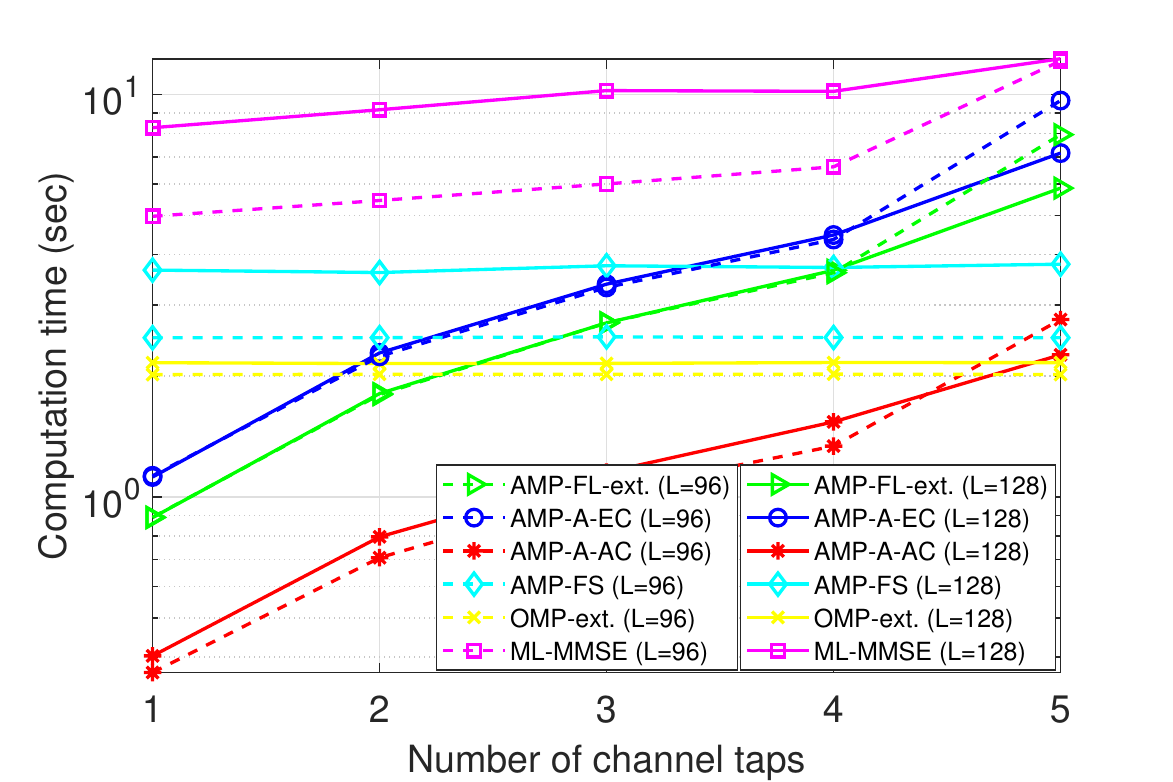}}}
	\end{center}
	\vspace{-5mm}
	\caption{\small{Error probability, false alarm probability, missed detection probability, MSE, and computation time versus the number of channel taps at {\color{black}$L=96$ ($L=128$),} ${\color{black}P_t}=10$ dBm, $M=64$, and $N=1000$ for the c.d. case.}}
	\label{fig:variousTap}
\end{figure*}
\begin{figure*}[t]
	\begin{center}
		\subfigure[\scriptsize{\color{black}Error probability.}\label{fig:error_user}]
		{\resizebox{5.0cm}{!}{\includegraphics{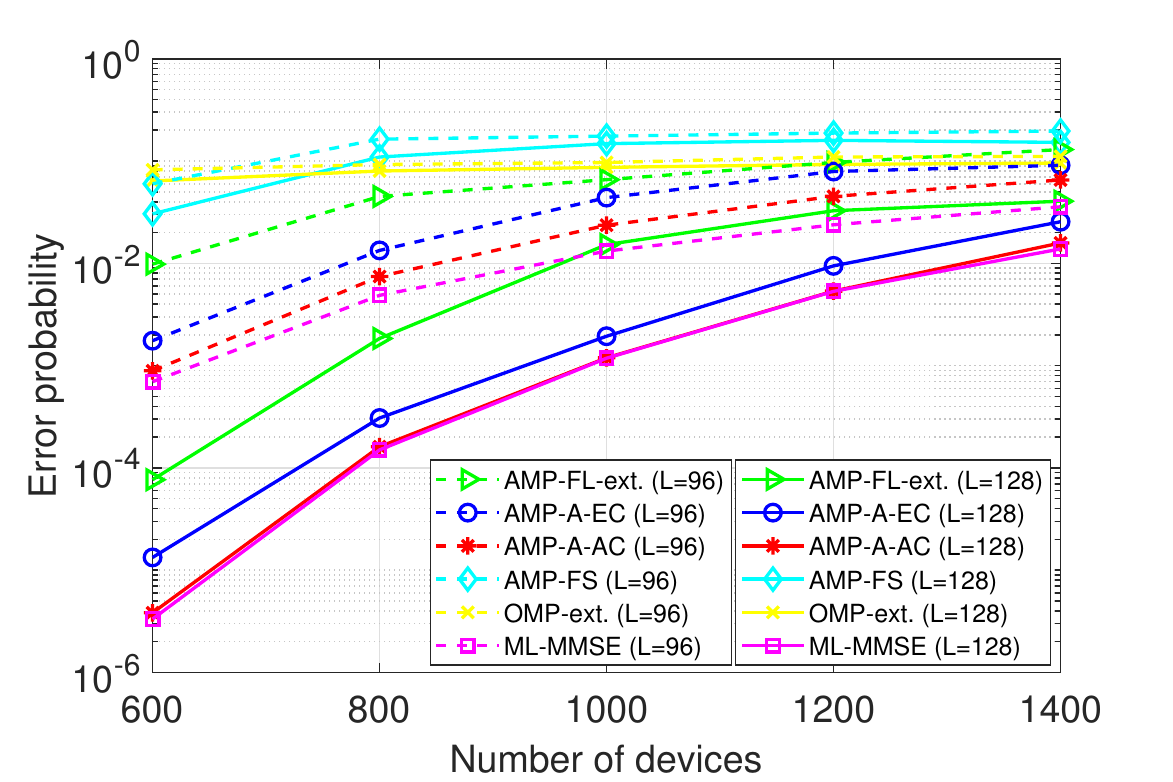}}}	
		\subfigure[\scriptsize{\color{black}MSE.}\label{fig:mse_user}]
		{\resizebox{5.0cm}{!}{\includegraphics{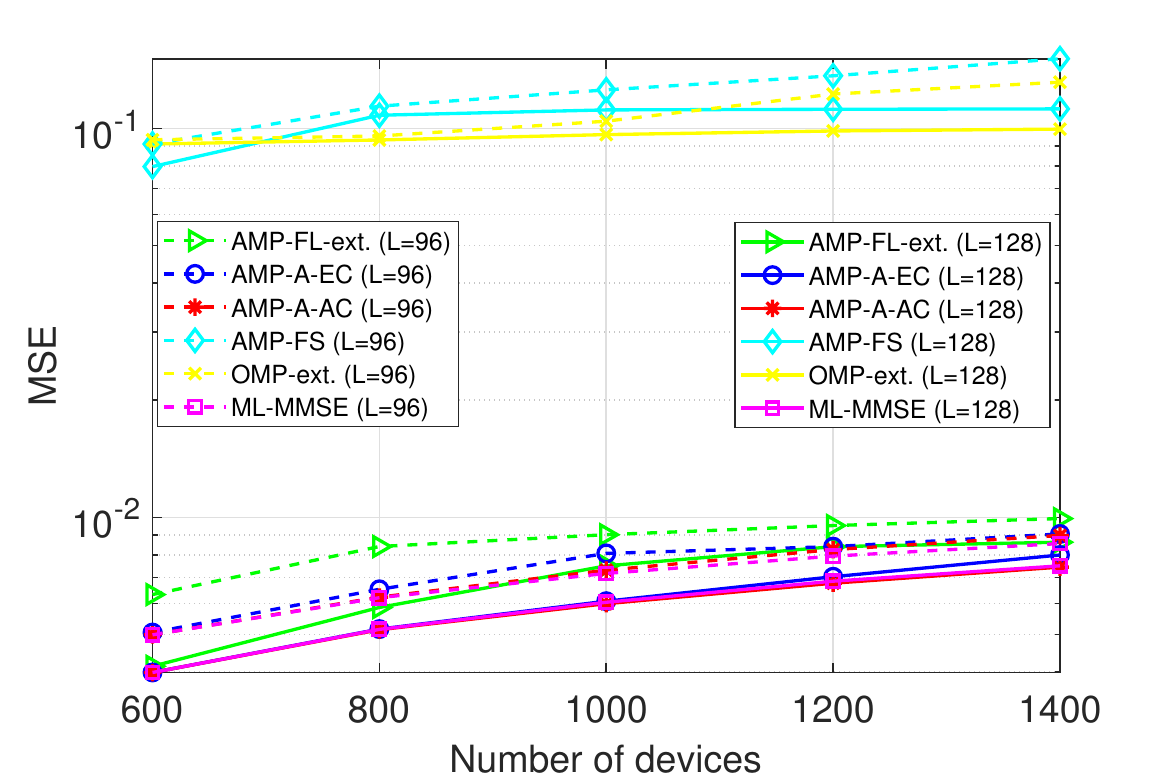}}}
		\subfigure[\scriptsize{\color{black}Computation time.}\label{fig:time_user}]
		{\resizebox{5.0cm}{!}{\includegraphics{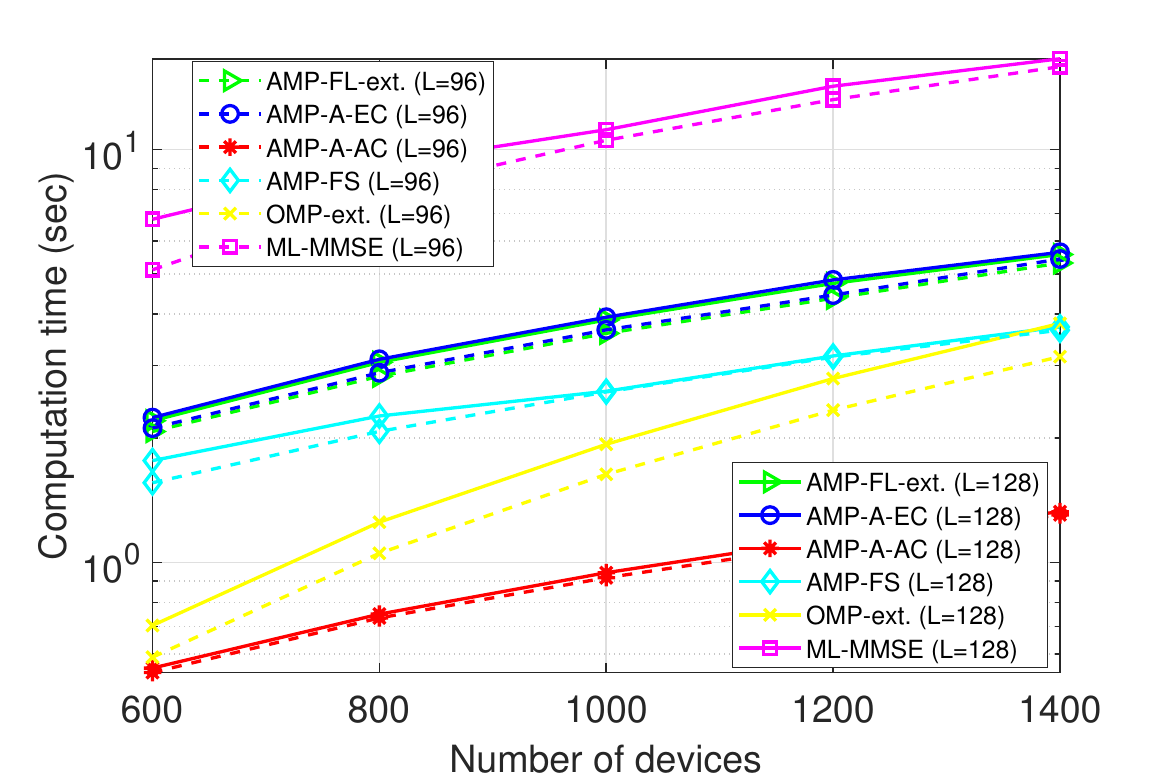}}}
	\end{center}
	\vspace{-0.2cm}
	\caption{\small{{\color{black}Error probability, MSE, and computation time versus the number of devices at $L=96$ ($L=128$), $P_t=10$ dBm, $M=64$, and $P=3$ for the c.d. case.}}}
	\label{fig:user}
\end{figure*}
\begin{figure*}[t]
	\begin{center}
		\subfigure[\scriptsize{\color{black}Error probability.}\label{fig:error_snr_r2}]
		{\resizebox{5.0cm}{!}{\includegraphics{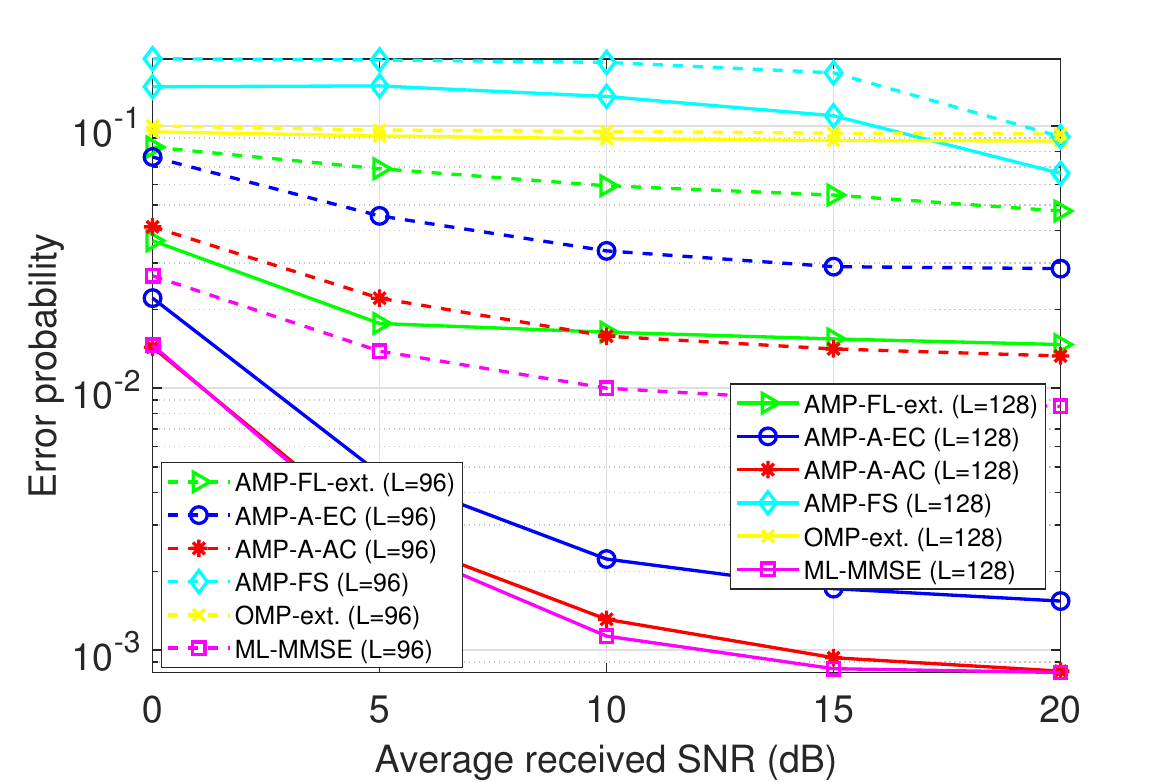}}}	
		\subfigure[\scriptsize{\color{black}MSE.}\label{fig:mse_snr_r2}]
		{\resizebox{5.0cm}{!}{\includegraphics{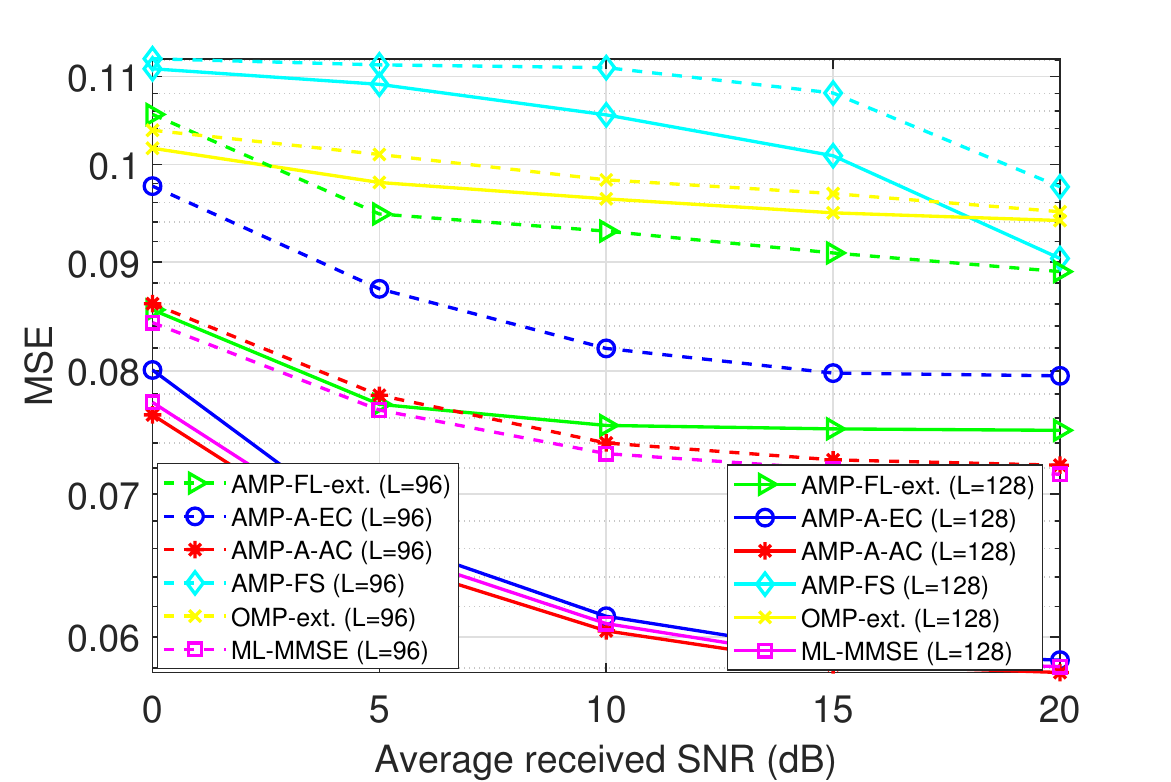}}}
		\subfigure[\scriptsize{\color{black}Computation time.}\label{fig:time_snr_r2}]
		{\resizebox{5.0cm}{!}{\includegraphics{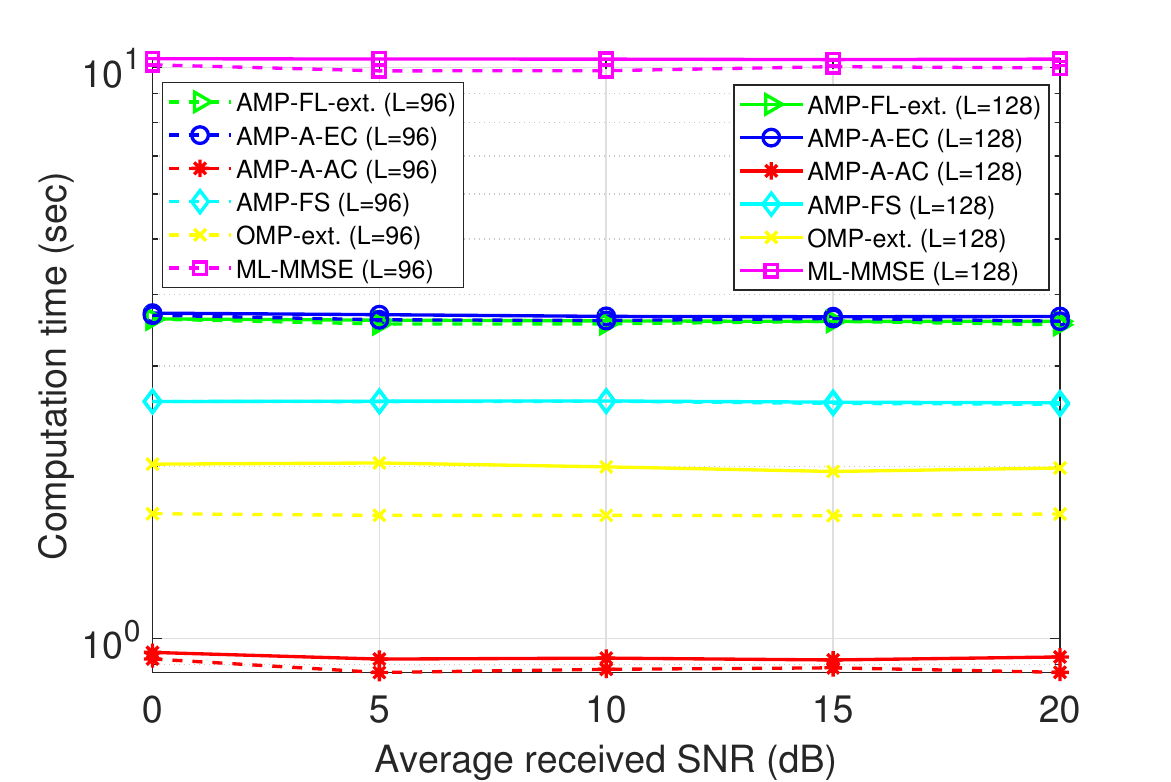}}}
	\end{center}
	\vspace{-0.2cm}
	\caption{\small{\color{black}Error probability, MSE, and computation time versus the average received SNR at $L=96$ ($L=128$), $N=1000$, $M=64$, and $P=3$ for the c.d. case.}}
	\label{fig:snr_r2}
\end{figure*}
\begin{figure*}[t]
	\begin{center}
		\subfigure[\scriptsize{\color{black}Error probability.}\label{fig:error_snr_rd}]
		{\resizebox{5.0cm}{!}{\includegraphics{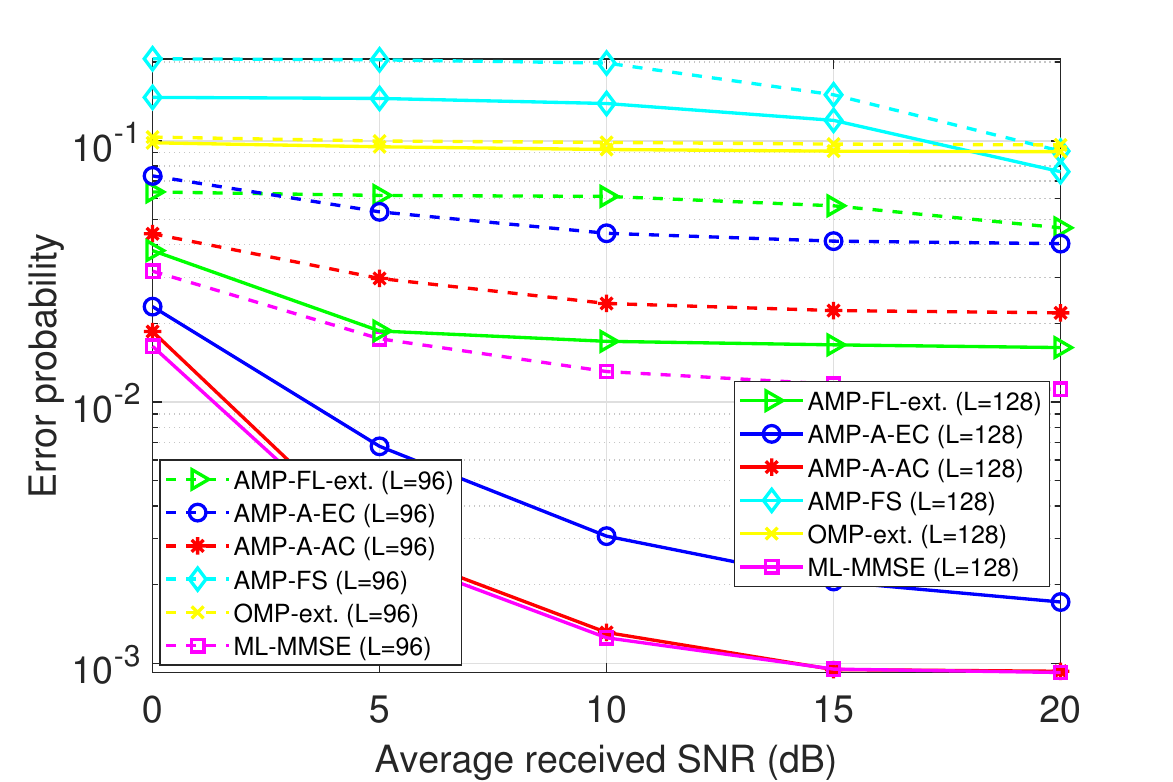}}}	
		\subfigure[\scriptsize{\color{black}MSE.}\label{fig:mse_snr_rd}]
		{\resizebox{5.0cm}{!}{\includegraphics{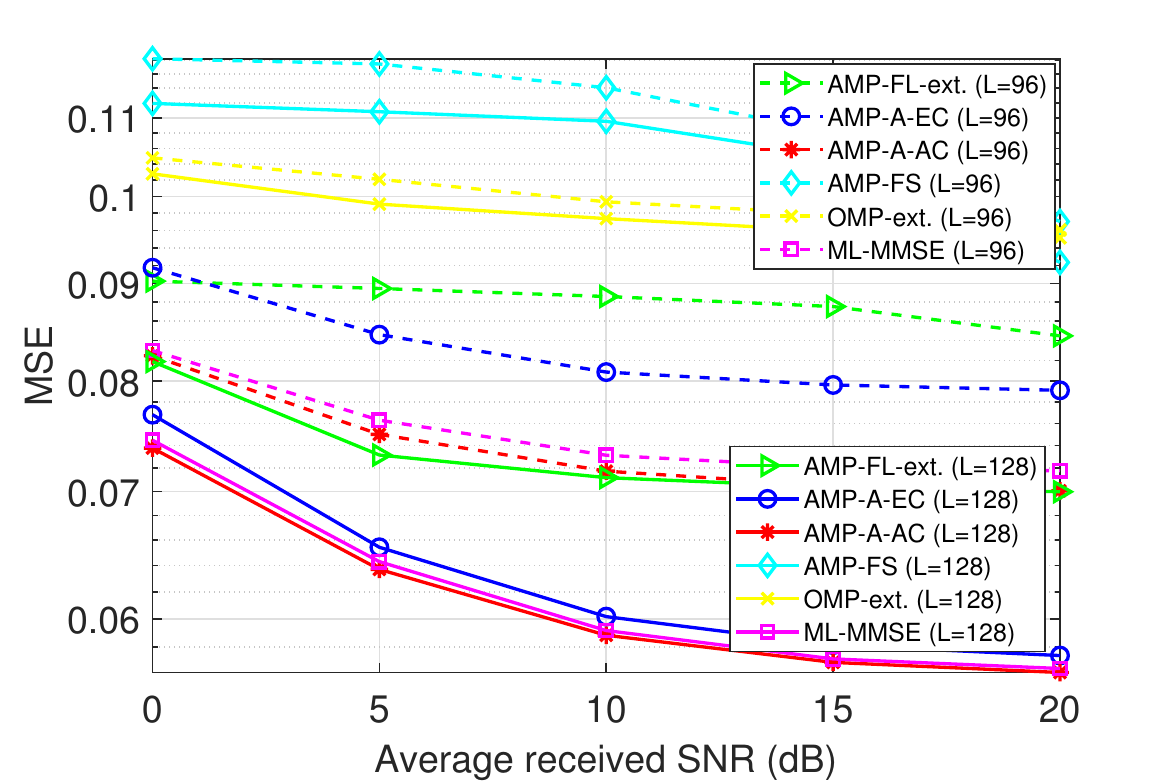}}}
		\subfigure[\scriptsize{\color{black}Computation time.}\label{fig:time_snr_rd}]
		{\resizebox{5.0cm}{!}{\includegraphics{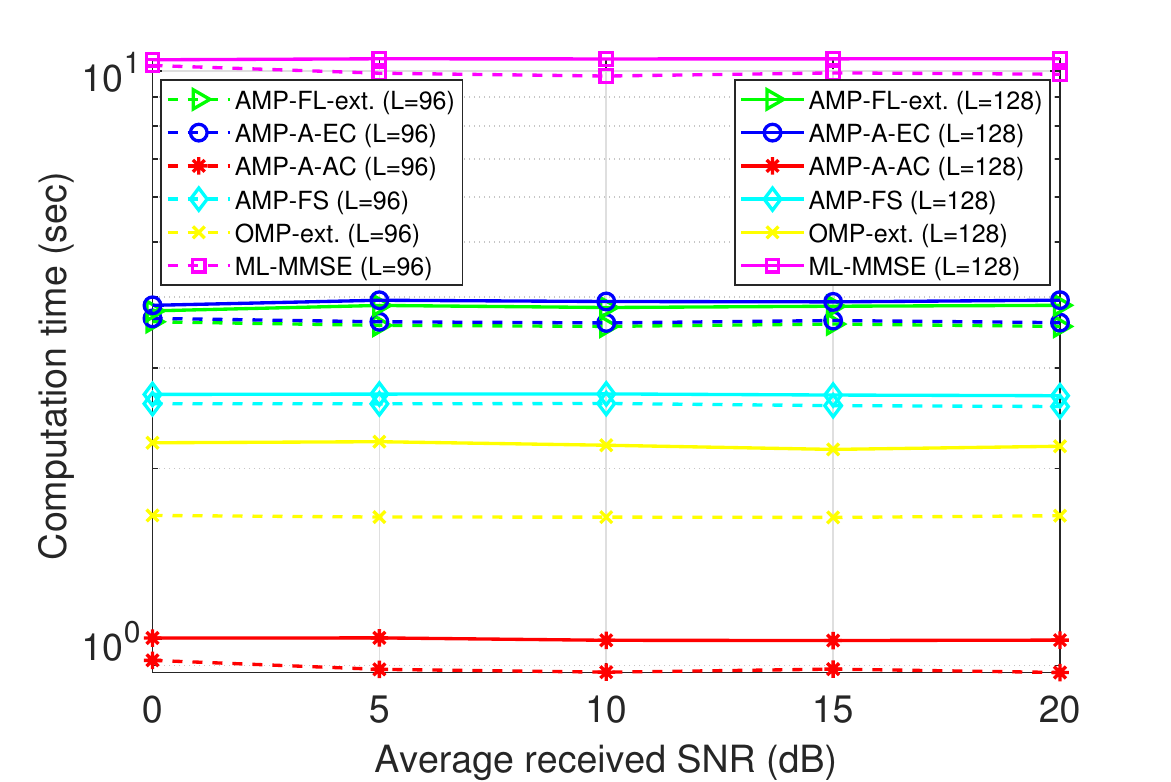}}}
	\end{center}
	\vspace{-0.2cm}
	\caption{\small{\color{black}Error probability, MSE, and computation time versus the average received SNR at $L=96$ ($L=128$), $N=1000$, $M=64$, and $P=3$ for the r.d. case.}}
	\label{fig:snr_rd}
\end{figure*}

 Fig.~\ref{fig:iter_error} and Fig.~\ref{fig:iter_mse} illustrate the error probabilities of activity detection and MSEs of channel estimation of the proposed algorithms, i.e., \emph{AMP-A-EC} and \emph{AMP-A-AC}, and the counterparts without tracking the best estimates over iterations, i.e., \emph{AMP-A-EC-iter.} and \emph{AMP-A-AC-iter.}, versus the number of iterations in the cases of $L=96$ ($<N\rho$), $L=128$ ($L$ slightly larger than $N\rho$), and $L=192$ ($>N\rho$), respectively, where $N\rho=100$. The error probability and MSE of \emph{AMP-A-EC-iter.} (\emph{AMP-A-AC-iter.}) are measured according to $\hat{\mathbf{a}}^{(t)}$ and $\hat{\mathbf{X}}^{(t)}$ ($\hat{\mathbf{a}}^{(t)}$ and $\hat{\mathbf{H}}^{(t)}$). 
From Fig.~\ref{fig:iter_error_L96} and Fig.~\ref{fig:iter_mse_L96} for $L=96$, and {Fig.~\ref{fig:iter_error_L128} and Fig.~\ref{fig:iter_mse_L128}} for $L=128$, the error probabilities and MSEs of \emph{AMP-A-EC-iter.} and \emph{AMP-A-AC-iter.}, which do not track the best estimates over iterations, first decrease and then increase with the increase in the number of iterations due to the lack of convergence guarantee of the underlying AMP method in these regimes \cite{donohoAMPsparse}, as discussed in Section~\ref{sec:al1}, whereas the error probabilities and MSEs of \emph{AMP-A-EC} and \emph{AMP-A-AC} first decease and then remain unchanged with the increase in the number of iterations. 
At the 20-th iteration, \emph{AMP-A-EC} and \emph{AMP-A-AC} achieve smaller error probabilities ($33\%\text{--}37.5\%$ and $18.1\%\text{--}19.4\%$ smaller) and MSEs ($1.2\%\text{--}3\%$ and $0.8\%\text{--}1.4\%$ smaller) than \emph{AMP-A-EC-iter.} and \emph{AMP-A-AC-iter.}, respectively.
This confirms the benefit of tracking the best estimates over iterations at $L<N\rho$ and $L$ slightly greater than $N\rho$.
From Fig.~\ref{fig:iter_error_L192} and Fig.~\ref{fig:iter_mse_L192} for $L=192$, we can observe that the error probability and MSE of each algorithm are monotonically decreasing with the increase in the number of iterations due to the improvement of the convergence of the underlying AMP method \cite{Liu18TSP}.
Fig.~\ref{fig:iter_error} and Fig. \ref{fig:iter_mse} also indicate that the analytical error probability and MSE of \emph{AMP-A-EC}, represented by \emph{Analysis-AMP-A-EC}, are close to the numerical ones. 

Fig.~\ref{fig:variousL}, Fig.~\ref{fig:variousM}, and Fig.~\ref{fig:variousTap} plot the error probabilities, false alarm probabilities, missed detection probabilities, MSEs, and computation times of all algorithms versus the pilot length $L$, number of antennas $M$, number of channel taps $P$, respectively, in the cases of $L=96$ and $L=128$. 
Fig.~\ref{fig:user}, Fig.~\ref{fig:snr_r2}, {\color{black}and Fig.~\ref{fig:snr_rd}} plot the error probabilities, MSEs, and computation times of all algorithms versus the number of devices $N$ and average received SNR
{\color{black} in the cases of $L=96$ and $L=128$.}
From Fig.~\ref{fig:variousL} and Fig.~\ref{fig:variousM}, we can observe that each algorithm's error probability, false alarm probability, missed detection probability, and MSE decrease with $L$ and $M$, mainly due to more measurement vectors adopted and more observations available, respectively. From Fig. \ref{fig:variousTap} and Fig. \ref{fig:user}, we can see that each algorithm's error probability and MSE increase with $P$ and $N$ due to more parameters to be estimated.
Fig. \ref{fig:snr_r2} {\color{black}and Fig.~\ref{fig:snr_rd} show} that each algorithm's error probability and MSE decrease with the average received SNR \textcolor{black}{and that} each algorithm in the c.d. and r.d. \textcolor{black}{cases} exhibits similar error probability and MSE at the same average received SNR. 
From Fig.~\ref{fig:time_L}, Fig.~\ref{fig:time_M}, and Fig. \ref{fig:time_user}, we can see that each algorithm's computation time increases with $L$, $M$, and $N$. From Fig.~\ref{fig:time_Tap}, we can observe that {\em AMP-FS} and {\em OMP-ext.}'s computation times remain unchanged with the increase of $P$ since they rely on the frequency-domain signal model (irrespective of $P$). In contrast, the other algorithms' computation times increase with $P$ since they rely on the time-domain signal model (dependent on $P$). Fig.~\ref{fig:time_snr_r2} \textcolor{black}{and Fig.~\ref{fig:time_snr_rd} show} that each algorithm's computation time is irrelevant to the average received SNR. These observations are in accordance with the computational complexity analysis in Table \ref{tab:table_revised_complexity}.

Furthermore, we can make the following observations from \textcolor{black}{Figs.~\ref{fig:variousL}-\ref{fig:snr_rd}}. Firstly, \emph{AMP-A-EC} and \emph{AMP-A-AC} achieve comparable error probabilities (false alarm probabilities and missed detection probabilities) and MSEs with significantly shorter ($96\%$ shorter) computation times than ML-MMSE. Secondly, \emph{AMP-A-EC} and \emph{AMP-A-AC} achieve substantially smaller error probabilities (false alarm probabilities and missed detection probabilities) and MSEs ($94\%$ and $33\%$ smaller, respectively, at $M=64$) than \emph{AMP-FL-ext.}, \emph{AMP-FS}, and \emph{OMP-ext}. The significant gains in accuracy over \emph{AMP-FL-ext.}, \emph{AMP-FS}, and \emph{OMP-ext} come from the fact that they precisely capture the dependency of the actual frequency-domain channels and select the best estimate among the iterates, as discussed in Section \ref{sec:comb}.
 Thirdly, \emph{AMP-A-AC} achieves the lowest computation among all algorithms in most cases, whereas \emph{AMP-A-EC} requires a slightly longer computation time than \emph{AMP-FL-ext.}, \emph{AMP-FS}, and \emph{OMP-ext.} due to the calculation of additional parameters $\theta^{(t)}_n,n \in \mathcal{N}$ in (\ref{eq:theta}) and the process of tracking the best estimate over iterations. Fourthly, \emph{AMP-A-EC} slightly outperforms \emph{AMP-A-AC} at small $L$ and $M$ but slightly underperforms \emph{AMP-A-AC} at large $L$ and $M$, as the accuracies of the respective approximations vary with the values of the parameters. Additionally, \emph{AMP-A-AC} has a shorter computation time than \emph{AMP-A-EC} due to its lower computational complexity for channel estimation and residual calculation, as discussed in Section \ref{sec:complexity}.

In summary, the gains in accuracy and computation time demonstrate the practical values of the proposed algorithms for OFDM-based grant-free access under frequency-selective fading. Moreover, \emph{AMP-A-EC} is preferable at small $L$ and $M$, and \emph{AMP-A-AC} is applicable at large $L$ and $M$.

\section{Conclusion}
This paper investigated AMP-based joint device activity detection and channel estimation for OFDM-based grant-free access under frequency-selective fading. First, we presented the general exact time-domain signal model and built a new factor graph that captures the precise statistics of time-domain channels and device activities.
Then, we proposed two AMP-based algorithms, i.e., \emph{AMP-A-EC} and \emph{AMP-A-AC}, which approximately solve the MAP-based device activity detection problem and two MMSE-based channel estimation problems and always return the best estimates over iterations. Besides, we analyzed \emph{AMP-A-EC}'s error probability of device activity detection and the MSE of channel estimation via state evolution and showed \emph{AMP-A-AC}'s advantage in computational complexity. Numerical results demonstrated \emph{AMP-A-EC} and \emph{AMP-A-AC}'s substantial gains over the existing algorithms and the respective preferable regions. Numerical results also validated the performance analysis of \emph{AMP-A-EC}.

{\appendices 
	\section{Proof of the Lemma \ref{lemma:MMSE}} \label{appendix:marginal}
	We have:
	\begin{align}\nonumber
		&p(a_n|\mathbf{Y}) \overset{(a)}{=}\frac{p(a_n,\mathbf{Y})}{p(\mathbf{Y})}\overset{(b)}{=} \frac{\sum_{\mathbf{a}_{\bar{n}}} \int p(\mathbf{X},\mathbf{a},\mathbf{Y})d\mathbf{X} }{\sum_{\mathbf{a}} \int p(\mathbf{X},\mathbf{a},\mathbf{Y})d\mathbf{X} },\ n\in\mathcal{N},
\\ \nonumber
		&p(x_{n,p,m}|\mathbf{Y}) \overset{(a)}{=}  \frac{p(x_{n,p,m},\mathbf{Y})}{p(\mathbf{Y})},
\\ \nonumber
		&\overset{(b)}{=}\frac{\sum_{\mathbf{a}}\int p(\mathbf{X},\mathbf{a},\mathbf{Y})d\mathbf{X}_{\overline{n,p,m}} }{\int \sum_{\mathbf{a}}\int p(\mathbf{X},\mathbf{a},\mathbf{Y})d\mathbf{X}_{\overline{n,p,m}} d x_{n,p,m}}, \\ \nonumber
		& n\in\mathcal{N}, p \in \mathcal{P}, m \in\mathcal{M},\\  \nonumber
		&{\color{black} p(x_{n,p,m}|\mathbf{Y},a_n)} \overset{(a)}{=}  \frac{p(x_{n,p,m},a_n,\mathbf{Y})}{p(a_n,\mathbf{Y})} \\ \nonumber
		&{\color{black}\overset{(b)}{=}\frac{\sum_{\mathbf{a}_{\bar{n}}}\int p(\mathbf{X},\mathbf{a},\mathbf{Y})d\mathbf{X}_{\overline{n,p,m}} }{\int \sum_{\mathbf{a}_{\bar{n}}}\int p(\mathbf{X},\mathbf{a},\mathbf{Y})d\mathbf{X}_{\overline{n,p,m}} d x_{n,p,m}}}, \\ \nonumber
		&n\in\mathcal{N}, p \in \mathcal{P}, m \in\mathcal{M},  
	\end{align}
	where $(a)$ is due to the Bayes' theorem, and $(b)$ is due to the relationship between marginal distribution and joint distribution.

	\section{Proof of the Lemma \ref{lemma:consistent_message}} \label{appendix:a}
	First, we prove of Lemma~\ref{lemma:consistent_message} (i). Firstly, by substituting (\ref{lemma2:eq}) into (\ref{eq:message_product_f_x}) and by Gaussian product theorem\footnote{Gaussian product theorem states that: $f_{\mathcal{CN}}(x;a_1,b_1)f_{\mathcal{CN}}(x;a_2,b_2) = f_{\mathcal{CN}}(0;a_1-a_2,b_1+b_2)f_{\mathcal{CN}}(x;c,d)$ where $d = \frac{b_1b_2}{b_1+b_2}$, $c = d (\frac{a_1}{b_1}+\frac{a_2}{b_2})$.}, we have (\ref{eq:productL}).
Secondly, by substituting (\ref{eq:prior_xa}) and (\ref{eq:productL}) into (\ref{eq:message_integral_x}) and by the Gaussian product theorem, we have
\begin{align}\nonumber
	&\mu^{(t)}_{q_{n,p,m} \rightarrow a_{n}}(a_{n}) =  C_5 \Big(a_n f_{\mathcal{CN}}(0;\hat{r}^{(t)}_{x_{n,p,m}},\hat{\sigma}^{(t)}_{x_{n,p,m}} + \beta_n) \\ \label{eq:appendixA_qa}
	&\quad+ (1-a_n) f_{\mathcal{CN}}(0;\hat{r}^{(t)}_{x_{n,p,m}},\hat{\sigma}^{(t)}_{x_{n,p,m}}) \Big),
\end{align}
where $C_5 \triangleq f_{\mathcal{CN}}(0;\hat{r}^{(t)}_{x_{n,p,m}},\hat{\sigma}^{(t)}_{x_{n,p,m}} + \beta_n)+ f_{\mathcal{CN}}(0;\hat{r}^{(t)}_{x_{n,p,m}},\hat{\sigma}^{(t)}_{x_{n,p,m}})$ is the normalizing constant.
Form (\ref{eq:appendixA_qa}), we can obtain (\ref{eq:mes_q_a}).
Thirdly, by substituting (\ref{eq:PMF_a}) and (\ref{eq:mes_q_a}) into (\ref{eq:message_a_q}), we have (\ref{eq:mes_a_q}). Fourthly, by substituting (\ref{eq:prior_xa}) and (\ref{eq:mes_a_q}) into (\ref{eq:message_q_x}), (\ref{eq:approximate_q_x}) is readily obtained. Fifthly, by substituting (\ref{lemma2:eq}) and (\ref{eq:approximate_q_x}) into (\ref{eq:posterior_x_except}) and by the Gaussian product theorem, we have (\ref{eq:approximate_x_f}). 
By (\ref{eq:approximate_x_f}) and the Gaussian product theorem, we have (\ref{eq:mean_x_f_x}).
At last, by (\ref{eq:approximate_x_f}), $\hat{\nu}^{(t+1)}_{ x_{n,p,m} \rightarrow f_{l,m}} = \int|x_{n,p,m}|^2 \mu^{(t+1)}_{x_{n,p,m} \rightarrow f_{l,m}}(x_{n,p,m}) d x_{n,p,m}
- |\hat{x}^{(t+1)}_{ x_{n,p,m} \rightarrow f_{l,m}}|^2$, and the Gaussian product theorem, we have (\ref{eq:mean_nu_f_x}). 
Next, we prove Lemma~\ref{lemma:consistent_message} (ii). Firstly, by substituting (\ref{eq:PMF_a}) and (\ref{eq:mes_q_a}) into (\ref{eq:detector_MP}), we have (\ref{eq:lemma_llr}). Secondly, by substituting (\ref{eq:productL}) and (\ref{eq:approximate_q_x}) into (\ref{eq:inte_posterior_x}) and the Gaussian product theorem, we have (\ref{eq:lemma_x}). At last, by substituting (\ref{eq:prior_xa}) and (\ref{eq:approximate_q_x}) into (\ref{eq:inte_posterior_h}) and the Gaussian product theorem, we have (\ref{eq:lemma_h}).

	\section{Proof of the Lemma \ref{theo:tau}} \label{appendix:lemma4}
First, by (\ref{eq:new_v_sec4}), (\ref{eq:appro_tau}), and (\ref{eq:taylor_v}), we have:
\begin{align} \nonumber
	&\tau^{(t)}_{l,m}  =  \hat{\tau}^{(t)}_{m} + \frac{1}{L} \sum\limits_{n \in \mathcal{N}} \sum\limits_{p \in \mathcal{P}}\mathcal{O}\left(\frac{1}{N}\right) + \mathcal{O}\left(\frac{1}{L}\right) \\ \label{eq:appendix_tau_l_m}
	&\overset{(a)}{=} \hat{\tau}^{(t)}_{m} + \mathcal{O}\left(\frac{1}{L}\right) = \hat{\tau}^{(t)}_{m} + o(1), \text{as $L,N \rightarrow \infty$},
\end{align}
where $(a)$ is due to that $\frac{1}{L} \sum\limits_{n \in \mathcal{N}} \sum\limits_{p \in \mathcal{P}}\mathcal{O}\left(\frac{1}{N}\right) = \frac{1}{L} \sum\limits_{p \in \mathcal{P}}\mathcal{O}\left(1\right) = \mathcal{O}\left(\frac{1}{L}\right)$. 
Then, by (\ref{eq:mean_V_f_x}) and (\ref{eq:new_v_sec4}), we have:
\begin{align} \nonumber
	&\hat{\gamma}^{(t)}_{f_{l,m} \rightarrow x_{n,p,m}} = \tau^{(t)}_{l,m} - |A_{l,n,p}|^2 \hat{\nu}^{(t)}_{ x_{n,p,m} \rightarrow f_{l,m}}\\  \label{eq:appendix_gamma}
	&\overset{(b)}{=} \tau^{(t)}_{l,m} + \mathcal{O}\left(\frac{1}{L}\right) =\tau^{(t)}_{l,m} + o(1), \text{as $L,N \rightarrow \infty$},
\end{align}
where $(b)$ is due to that $|A_{l,n,p}|^2 = \mathcal{O}\left(\frac{1}{L}\right)$, and $\hat{\nu}^{(t)}_{ x_{n,p,m} \rightarrow f_{l,m}}= \mathcal{O}\left(1\right)$ \cite{DonohoAMP}. Next, by (\ref{eq:lemma_sigma_x}) and (\ref{eq:appendix_gamma}), we have:
\begin{align} \nonumber
	&\hat{\sigma}^{(t)}_{x_{n,p,m}} = \left(\sum_{l \in \mathcal{L}}\frac{|A_{l,n,p}|^2}{\tau^{(t)}_{l,m} + \mathcal{O}\left(\frac{1}{L}\right)}\right)^{-1}\overset{(c)}{=} \left(\sum_{l \in \mathcal{L}}\frac{|A_{l,n,p}|^2}{\hat{\tau}^{(t)}_{m} + \mathcal{O}\left(\frac{1}{L}\right)}\right)^{-1}\\
	&\overset{(d)}{=} \hat{\tau}^{(t)}_{m} + \mathcal{O}\left(\frac{1}{L}\right) = \hat{\tau}^{(t)}_{m} + o(1), \text{as $L,N \rightarrow \infty$}, 
\end{align}
where $(c)$ is due to (\ref{eq:appendix_tau_l_m}), and $(d)$ is due to $\sum_{l \in \mathcal{L}}|A_{l,n,p}|^2 = 1$. Finally, by (\ref{eq:lemma_sigma_x}) and (\ref{eq:lemma_sigma_f_x}), we have:
\begin{align} \nonumber
	&(\hat{\sigma}^{(t)}_{x_{n,p,m}\rightarrow f_{l,m}})^{-1} = (\hat{\sigma}^{(t)}_{x_{n,p,m}})^{-1} - \frac{|A_{l,n,p}|^2}{\hat{\gamma}^{(t)}_{f_{l,m} \rightarrow x_{n,p,m}}}\\ \nonumber
	&\overset{(e)}{=} (\hat{\sigma}^{(t)}_{x_{n,p,m}})^{-1} + \mathcal{O}\left(\frac{1}{L}\right) = (\hat{\sigma}^{(t)}_{x_{n,p,m}})^{-1} + o(1),
\end{align}
as $L,N \rightarrow \infty$, where $(e)$ is due to that $\hat{\gamma}^{(t)}_{f_{l,m} \rightarrow x_{n,p,m}}= \mathcal{O}\left(1\right)$\cite{DonohoAMP} and $|A_{l,n,p}|^2 = \mathcal{O}\left(\frac{1}{L}\right)$. Therefore, we complete the proof.

\section{Proof of the Theorem \ref{theo:algorithm}} \label{appendix:theorem1}
First, we obtain $\hat{r}^{(t)}_{x_{n,p,m} \rightarrow f_{l,m}} - \hat{r}^{(t)}_{x_{n,p,m}}$. By (\ref{eq:lemma_r_x}) and  (\ref{eq:new_z}), we have: 
\begin{align} \label{app:new_r}
	&\hat{r}^{(t)}_{x_{n,p,m} } =  \sum_{l \in \mathcal{L}}A^{*}_{l,n,p}(z^{(t)}_{f_{l,m}}- A_{l,n,p}\hat{x}^{(t)}_{ x_{n,p,m} \rightarrow f_{l,m}}).
\end{align}
By (\ref{eq:lemma_r_f_x}) and  (\ref{eq:new_z}), we have: 
\begin{align} \label{app:new_r_f}
	&\hat{r}^{(t)}_{x_{n,p,m} \rightarrow f_{l,m}} = \sum_{\substack{b \in \mathcal{L},\\ b \neq l}}A^{*}_{b,n,p}(z^{(t)}_{f_{l,m}}- A_{l,n,p}\hat{x}^{(t)}_{ x_{n,p,m} \rightarrow f_{l,m}}).
\end{align}
By (\ref{app:new_r}), (\ref{app:new_r_f}), and (\ref{eq:new_z}), we have:
\begin{align} \nonumber
	&\hat{r}^{(t)}_{x_{n,p,m} \rightarrow f_{l,m}} - \hat{r}^{(t)}_{x_{n,p,m} } = -A^{*}_{l,n,p}(z^{(t)}_{f_{l,m}}- A_{l,n,p}\hat{x}^{(t)}_{ x_{n,p,m} \rightarrow f_{l,m}}) \\ \label{eq:app:r_diff}
	&\overset{(a)}{=}-A^{*}_{l,n,p}z^{(t)}_{f_{l,m}} + \mathcal{O}\left(\frac{1}{L}\right), \text{as $L,N \rightarrow \infty$}, 
\end{align}
where $(a)$ is due to that $|A_{l,n,p}|^2 = \mathcal{O}\left(\frac{1}{L}\right)$, and $\hat{x}^{(t)}_{ x_{n,p,m} \rightarrow f_{l,m}}$ is bounded \cite{DonohoAMP}.
Then, by substituting (\ref{eq:app:r_diff}) into (\ref{eq:theorem1_taylor}), we have:
\begin{align}\nonumber
	\hat{x}^{(t+1)}_{ x_{n,p,m} \rightarrow f_{l,m}} =& \hat{x}^{(t+1)}_{ n,p,m} - \eta^{'}_n(\hat{r}^{(t)}_{x_{n,p,m} },\hat{\tau}^{(t)}_{m},\lambda^{(t+1)}_{n,p,m})A^{*}_{l,n,p}z^{(t)}_{f_{l,m}}\\ \label{app:taylor_x}
	& + \mathcal{O}(\frac{1}{L}) + \mathcal{O}(\frac{1}{N}), \text{as $L,N \rightarrow \infty$}.
\end{align}
Next, we obtain (\ref{eq:theorem_z}). By substituting (\ref{app:taylor_x}) into (\ref{eq:new_z}), we have:
\begin{align} \nonumber 
	z^{(t+1)}_{f_{l,m}}=& y_{l,m} - \sum_{n \in \mathcal{N}}\sum_{p \in \mathcal{P}} A_{l,n,p} \hat{x}^{(t+1)}_{n,p,m} \\ \nonumber
	&+ z^{(t)}_{f_{l,m}}\sum_{n \in \mathcal{N}}\sum_{p \in \mathcal{P}} |A_{l,n,p}|^2 \eta_n^{'}\left(\hat{r}^{(t)}_{x_{n,p,m}}, \hat{\tau}^{(t)}_{m}, \lambda^{(t+1)}_{n,p,m}\right) \\ \nonumber
	&+ \mathcal{O}(\frac{1}{L})\sum_{n \in \mathcal{N}}\sum_{p \in \mathcal{P}} A_{l,n,p}+ \mathcal{O}(\frac{1}{N})\sum_{n \in \mathcal{N}}\sum_{p \in \mathcal{P}} A_{l,n,p}\\ \nonumber
	\overset{(b)}{=}& y_{l,m} - \sum_{n \in \mathcal{N}}\sum_{p \in \mathcal{P}} A_{l,n,p} \hat{x}^{(t+1)}_{n,p,m}  \\ \nonumber
	&+ \frac{1}{L} z^{(t)}_{f_{l,m}}\sum_{n \in \mathcal{N}}\sum_{p \in \mathcal{P}}    \eta_n^{'}\left(\hat{r}^{(t)}_{x_{n,p,m}}, \hat{\tau}^{(t)}_{m}, \lambda^{(t+1)}_{n,p,m}\right)  \\ \nonumber
	&+ \mathcal{O}(\frac{N}{L^{\frac{3}{2}}}) + \mathcal{O}(\frac{1}{\sqrt{L}}) + \mathcal{O}(\frac{1}{L}), \text{as $L,N \rightarrow \infty$},
\end{align}
where $(b)$ is due to (\ref{eq:appro_A_x}) and $A_{l,n,p} = \mathcal{O}\left(\frac{1}{\sqrt{L}}\right)$. Thus, we have (\ref{eq:theorem_z}). Finally, we obtain (\ref{eq:theorem_r}). By substituting (\ref{app:taylor_x}) into (\ref{app:new_r}), we have:
\begin{align} \nonumber
	\hat{r}^{(t)}_{x_{n,p,m}} =& \sum_{l \in \mathcal{L}}A^{*}_{l,n,p}z^{(t)}_{f_{l,m}}  + \hat{x}^{(t)}_{ n,p,m}\sum_{l \in \mathcal{L}}|A_{l,n,p}|^2 \\ \nonumber
	&+ \sum_{l \in \mathcal{L}}|A_{l,n,p}|^2 \eta^{'}_n(\hat{r}^{(t)}_{x_{n,p,m} },\hat{\tau}^{(t)}_{m},\lambda^{(t+1)}_{n,p,m})A^{*}_{l,n,p}z^{(t)}_{f_{l,m}} \\ \nonumber
	&+ \mathcal{O}(\frac{1}{L}) \sum_{l \in \mathcal{L}}|A_{l,n,p}|^2 + \mathcal{O}(\frac{1}{N}) \sum_{l \in \mathcal{L}}|A_{l,n,p}|^2
	\\ \nonumber
	\overset{(c)}{=}& \sum_{l \in \mathcal{L}}A^{*}_{l,n,p}z^{(t)}_{f_{l,m}}  + \hat{x}^{(t)}_{ n,p,m} + \mathcal{O}(\frac{1}{\sqrt{L}})\\ \nonumber
	&+ \mathcal{O}(\frac{1}{L})  + \mathcal{O}(\frac{1}{N}) , \text{as $L,N \rightarrow \infty$},
\end{align}
where $(c)$ is due to that $\sum_{l \in \mathcal{L}}|A_{l,n,p}|^2 = 1$, $A_{l,n,p} = \mathcal{O}\left(\frac{1}{\sqrt{L}}\right)$, $|A_{l,n,p}|^2 = \mathcal{O}\left(\frac{1}{L}\right)$,  $\eta^{'}_n(\hat{r}^{(t)}_{x_{n,p,m} },\hat{\tau}^{(t)}_{m},\lambda^{(t+1)}_{n,p,m}) = \mathcal{O}(1)$, and $z^{(t)}_{f_{l,m}}= \mathcal{O}(1)$ \cite{DonohoAMP}. Thus, we have (\ref{eq:theorem_r}). Therefore, we complete the proof.

\section{Proof of the Theorem \ref{theo:error_prob}} \label{appendix:theorem2}
First, by substituting (\ref{eq:theta}) into (\ref{eq:a_detect}) and some manipulations, we have: 
\begin{align}
	&\hat{a}^{(t)}_n=\begin{cases}
		0,\ (\tilde{\mathbf{r}}^{(t)}_n)^H \tilde{\mathbf{r}}^{(t)}_n < \frac{\log \frac{1-\rho_n}{\rho_n} + PM \log (1+\frac{\beta_n}{\tau^{(t)}})}{\frac{1}{\tau^{(t)}} - \frac{1}{\beta_n + \tau^{(t)}}},\\
		1,\ (\tilde{\mathbf{r}}^{(t)}_n)^H \tilde{\mathbf{r}}^{(t)}_n \geq \frac{\log \frac{1-\rho_n}{\rho_n} + PM \log (1+\frac{\beta_n}{\tau^{(t)}})}{\frac{1}{\tau^{(t)}} - \frac{1}{\beta_n + \tau^{(t)}}},\\
	\end{cases} 
\end{align}
where $\tilde{\mathbf{r}}^{(t)}_n \triangleq a_n \mathbf{h}_n + \tau^{(t)} \mathbf{v}$. Next, by the assumption on $\mathbf{v}$, for given $a_n$, the entries of $\tilde{\mathbf{r}}^{(t)}_n$ are i.i.d. circularly symmetric complex Gaussian with zero mean and variance $a_n\beta_n+\tau^{(t)}$, and therefore $\frac{(\tilde{\mathbf{r}}^{(t)}_n)^H \tilde{\mathbf{r}}^{(t)}_n}{(a_n \beta_n + \tau^{(t)})/2}$ follows $\mathcal{X}^2$ distribution with $2PM$ degrees of freedom (DoF). Thus, we have:
\begin{align} \nonumber
	&{\rm Pr}(\hat{a}^{(t)}_n = 0| a_n=1) \\ \nonumber
	=& {\rm Pr}\left(\frac{(\tilde{\mathbf{r}}^{(t)}_n)^H \tilde{\mathbf{r}}^{(t)}_n}{(\beta_n + \tau^{(t)})/2} \leq \frac{\log \frac{1-\rho_n}{\rho_n} + PM \log (1+\frac{\beta_n}{\tau^{(t)}})}{\left(\frac{\beta_n + \tau^{(t)}}{2}\right)\left(\frac{1}{\tau^{(t)}} - \frac{1}{\beta_n + \tau^{(t)}}\right)}\right)\\
	=& \frac{\underline{\Gamma}\left(PM, \frac{\tau^{(t)}}{\beta_n} \log \frac{1-\rho_n}{\rho_n}+ b_{t, n}(\beta_n,\tau^{(t)}) PM\right)}{\Gamma(PM)},\\ \nonumber
	&{\rm Pr}(\hat{a}^{(t)}_n = 1| a_n=0) \\ \nonumber
	=& {\rm Pr}\left(\frac{(\tilde{\mathbf{r}}^{(t)}_n)^H \tilde{\mathbf{r}}^{(t)}_n}{\tau^{(t)}/2} > \frac{\log \frac{1-\rho_n}{\rho_n} + PM \log (1+\frac{\beta_n}{\tau^{(t)}})}{\left(\frac{ \tau^{(t)}}{2}\right)\left(\frac{1}{\tau^{(t)}} - \frac{1}{\beta_n + \tau^{(t)}}\right)}\right)\\ 
	=& \frac{\bar{\Gamma}\left(PM, \frac{\tau^{(t)}+\beta_n}{\beta_n} \log \frac{1-\rho_n}{\rho_n} +c_{t, n}(\beta_n,\tau^{(t)}) PM\right)}{\Gamma(PM)}.
\end{align}
Finally, by $p_{n}^{\mathrm{e}(t)} = {\rm Pr}(a_n=1){\rm Pr}(\hat{a}^{(t)}_n = 0| a_n=1) + {\rm Pr}(a_n=0){\rm Pr}(\hat{a}^{(t)}_n = 1| a_n=0)$, we complete the proof.

 \section{Proof of the Theorem \ref{theorem:channel_error}} \label{appendix:h}
First, by (\ref{eq:version1_theorem_x}), we have:
\begin{align} \nonumber
	&\frac{1}{PM}\mathbb{E}_{\mathbf{h}_{n},\hat{\mathbf{h}}_{n}}\left[\|\mathbf{h}_{n} - \hat{\mathbf{h}}_{n}\|^2_2\right]\\ \nonumber
	=&\frac{1}{PM}\mathbb{E}_{\mathbf{h}_{n},\hat{\mathbf{r}}^{(t)}_{n}}\Big[ \left(\mathbf{h}_{n} -\frac{g(\hat{\mathbf{r}}^{(t)}_{n},\tau^{(t)}\mathbf{I}) \beta_n}{\beta_n + \tau^{(t)}} {\hat{\mathbf{r}}^{(t)}_{n}} \right)^H\\ \nonumber
	&\left(\mathbf{h}_{n}-\frac{g(\hat{\mathbf{r}}^{(t)}_{n},\tau^{(t)}\mathbf{I}) \beta_n}{\beta_n + \tau^{(t)}} {\hat{\mathbf{r}}^{(t)}_{n}}  \right)\Big]\\ \nonumber
	=& \frac{1}{PM}\mathbb{E}_{\hat{\mathbf{r}}^{(t)}_{n}}\Big[\mathbb{E}_{\mathbf{h}_{n}}\big[ \mathbf{h}^H_{n}\mathbf{h}_{n} -\frac{g(\hat{\mathbf{r}}^{(t)}_{n},\tau^{(t)}\mathbf{I}) \beta_n}{\beta_n + \tau^{(t)}} {\mathbf{h}^H_{n} \hat{\mathbf{r}}^{(t)}_{n}}\\ \nonumber
	&-\frac{g(\hat{\mathbf{r}}^{(t)}_{n},\tau^{(t)}\mathbf{I}) \beta_n}{\beta_n + \tau^{(t)}} {(\hat{\mathbf{r}}^{(t)}_{n})^H \mathbf{h}_{n}}|\hat{\mathbf{r}}^{(t)}_{n}\big]\Big] \\ \label{eq:h_est_o}
	&+ \mathbb{E}_{\hat{\mathbf{r}}^{(t)}_{n}}\Big[\frac{g^2(\hat{\mathbf{r}}^{(t)}_{n},\tau^{(t)}\mathbf{I}) \beta^2_n}{(\beta_n + \tau^{(t)})^2} {(\hat{\mathbf{r}}^{(t)}_{n})^H \hat{\mathbf{r}}^{(t)}_{n}} \Big].
\end{align}
Next, we simplify (\ref{eq:h_est_o}). Firstly, we have:
\begin{align}\nonumber
	p(\mathbf{h}_n|\hat{\mathbf{r}}^{(t)}_{n}) &\overset{(a)}{=} \frac{p(\hat{\mathbf{r}}^{(t)}_{n}|\mathbf{h}_n) p(\mathbf{h}_n)}{\int p(\hat{\mathbf{r}}^{(t)}_{n}|\mathbf{h}_n) p(\mathbf{h}_n) d \mathbf{h}_n}\\ \nonumber
	&\overset{(b)}{=} \frac{f_{\mathcal{CN}}(\hat{\mathbf{r}}^{(t)}_{n};\mathbf{h}_n,\tau^{(t)}\mathbf{I})f_{\mathcal{CN}}(\mathbf{h}_n;0,\beta_n \mathbf{I})}{f_{\mathcal{CN}}(\hat{\mathbf{r}}^{(t)}_{n};\mathbf{0},(\beta_n +\tau^{(t)})\mathbf{I})}\\ \label{eq:app_h_r}
	&\overset{(b)}{=} f_{\mathcal{CN}}(\mathbf{h}_n;\frac{\beta_n\hat{\mathbf{r}}^{(t)}_{n} }{\beta_n + \tau^{(t)}},\frac{\beta_n \tau^{(t)}}{\beta_n + \tau^{(t)}}),
\end{align}
where $(a)$ is due to the Bayes' theorem, $(b)$ is due to the Gaussian product theorem.
Secondly, we have:
\begin{align} \nonumber
	&\mathbb{E}_{\hat{\mathbf{r}}^{(t)}_{n}}\big[\mathbb{E}_{\mathbf{h}_{n}}[ \mathbf{h}^H_{n}\mathbf{h}_{n}|\hat{\mathbf{r}}^{(t)}_{n}]\big] \\ \label{eq:h_est1}
	\overset{(a)}{=}& \frac{PM\beta_n \tau^{(t)}}{\beta_n + \tau^{(t)}} + \mathbb{E}_{\hat{\mathbf{r}}^{(t)}_{n}}\Big[\frac{\beta^2_n(\hat{\mathbf{r}}^{(t)}_{n})^H \hat{\mathbf{r}}^{(t)}_{n}}{(\beta_n + \tau^{(t)})^2}\Big],\\ \nonumber
	&\mathbb{E}_{\hat{\mathbf{r}}^{(t)}_{n}}\big[\mathbb{E}_{\mathbf{h}_{n}}[ \frac{g(\hat{\mathbf{r}}^{(t)}_{n},\tau^{(t)}\mathbf{I}) \beta_n}{\beta_n + \tau^{(t)}} {\mathbf{h}^H_{n} \hat{\mathbf{r}}^{(t)}_{n}}|\hat{\mathbf{r}}^{(t)}_{n}]\Big]\\ \label{eq:h_est2}
	\overset{(a)}{=}&  \mathbb{E}_{\hat{\mathbf{r}}^{(t)}_{n}}\Big[\frac{g(\hat{\mathbf{r}}^{(t)}_{n},\tau^{(t)}\mathbf{I}) \beta^2_n}{(\beta_n + \tau^{(t)})^2} {(\hat{\mathbf{r}}^{(t)}_{n})^H\hat{\mathbf{r}}^{(t)}_{n}}\Big],\\ \nonumber
	&\mathbb{E}_{\hat{\mathbf{r}}^{(t)}_{n}}\big[\mathbb{E}_{\mathbf{h}_{n}}[ \frac{g(\hat{\mathbf{r}}^{(t)}_{n},\tau^{(t)}\mathbf{I}) \beta_n}{\beta_n + \tau^{(t)}} {(\hat{\mathbf{r}}^{(t)}_{n})^H\mathbf{h}_{n} }|\hat{\mathbf{r}}^{(t)}_{n}]\Big] \\ \label{eq:h_est3}
	\overset{(a)}{=}& \mathbb{E}_{\hat{\mathbf{r}}^{(t)}_{n}}\Big[\frac{g(\hat{\mathbf{r}}^{(t)}_{n},\tau^{(t)}\mathbf{I}) \beta^2_n}{(\beta_n + \tau^{(t)})^2} {(\hat{\mathbf{r}}^{(t)}_{n})^H\hat{\mathbf{r}}^{(t)}_{n}}\Big].
\end{align}
where $(a)$ is due to (\ref{eq:app_h_r}). Finally, by substituting (\ref{eq:h_est1}), (\ref{eq:h_est2}), and (\ref{eq:h_est3}) into (\ref{eq:h_est_o}), we complete the proof.

\section{Proof of the Theorem \ref{theo:algorithm2}} \label{appendix:second}
First, we obtain $\hat{r}^{(t)}_{x_{n,p,m} \rightarrow f_{l,m}} - \hat{r}^{(t)}_{x_{n,p,m} }$. By (\ref{eq:lemma_r_x}) and  (\ref{eq:new_z2}), we have: 
\begin{align} \nonumber
	&\hat{r}^{(t)}_{x_{n,p,m} } \\ \label{app:new_r2}
	&=  \sum_{l \in \mathcal{L}}A^{*}_{l,n,p}(\tilde{z}^{(t)}_{f_{l,m}} + A_{l,n,p}\epsilon^{(t)}_{ x_{n,p,m} \rightarrow f_{l,m}}\hat{h}^{(t)}_{ x_{n,p,m} \rightarrow f_{l,m}}).
\end{align}
By (\ref{eq:lemma_r_f_x}) and  (\ref{eq:new_z2}), we have: 
\begin{align} \nonumber
	&\hat{r}^{(t)}_{x_{n,p,m} \rightarrow f_{l,m}} \\ \label{app:new_r_f2}
	&= \sum_{\substack{b \in \mathcal{L},\\ b \neq l}}A^{*}_{b,n,p}(\tilde{z}^{(t)}_{f_{b,m}} + A_{l,n,p}\epsilon^{(t)}_{ x_{n,p,m} \rightarrow f_{b,m}}\hat{h}^{(t)}_{ x_{n,p,m} \rightarrow f_{b,m}}).
\end{align}
By (\ref{app:new_r2}), (\ref{app:new_r_f2}), and (\ref{eq:new_z2}), we have:
\begin{align} \nonumber
	&\hat{r}^{(t)}_{x_{n,p,m} \rightarrow f_{l,m}} - \hat{r}^{(t)}_{x_{n,p,m} } = -A^{*}_{l,n,p}\hat{z}^{(t)}_{f_{l,m} \rightarrow x_{n,p,m}}\\ \nonumber
	&= -A^{*}_{l,n,p}(\tilde{z}^{(t)}_{f_{l,m}}- A_{l,n,p}\epsilon^{(t)}_{ x_{n,p,m} \rightarrow f_{l,m}}\hat{h}^{(t)}_{ x_{n,p,m} \rightarrow f_{l,m}}) \\ \label{eq:app:r_diff2}
	&\overset{(a)}{=}-A^{*}_{l,n,p}\tilde{z}^{(t)}_{f_{l,m}} + \mathcal{O}\left(\frac{1}{L}\right), \text{as $L,N \rightarrow \infty$},
\end{align} 
where $(a)$ is due to $|A_{l,n,p}|^2 = \mathcal{O}\left(\frac{1}{L}\right)$,  $\epsilon^{(t)}_{ x_{n,p,m} \rightarrow f_{l,m}}\hat{h}^{(t)}_{ x_{n,p,m} \rightarrow f_{l,m}} = \hat{x}^{(t)}_{ x_{n,p,m} \rightarrow f_{l,m}}$, and the boundness of $\hat{x}^{(t)}_{ x_{n,p,m} \rightarrow f_{l,m}}$ \cite{DonohoAMP}. 
Then, we obtain (\ref{eq:epsilon}). Firstly, by (\ref{eq:app:r_diff2}), we have:
\begin{align} \nonumber
	&\hat{r}^{(t)}_{x_{n,p,m} \rightarrow f_{l,m}} = \hat{r}^{(t)}_{x_{n,p,m}} + A^{*}_{l,n,p}\tilde{z}^{(t)}_{f_{l,m}} + \mathcal{O}\left(\frac{1}{L}\right)\\ \label{app:epslion_r}
	&\overset{(b)}{=} \hat{r}^{(t)}_{x_{n,p,m}} +\mathcal{O}\left(\frac{1}{\sqrt{L}}\right) +\mathcal{O}\left(\frac{1}{L}\right) \text{as $L,N \rightarrow \infty$},
\end{align}
where $(b)$ is due to that $A_{l,n,p} = \mathcal{O}\left(\frac{1}{\sqrt{L}}\right)$, and $\tilde{z}^{(t)}_{f_{l,m}} = \mathcal{O}(1)$ \cite{SchniterAMP}.
Secondly, by substituting (\ref{app:epslion_r}) into (\ref{eq:epsilon_origin}), we have (\ref{eq:epsilon}). Next, by substituting (\ref{eq:app:r_diff2}) into (\ref{eq:taylor_h}), we have:
\begin{align} \nonumber
	\hat{h}^{(t+1)}_{ x_{n,p,m} \rightarrow f_{l,m}} &= \hat{h}^{(t+1)}_{ n,p,m} + \frac{\beta_n A^{*}_{l,n,p}\tilde{z}^{(t)}_{f_{l,m}}}{\hat{\tau}^{(t)}_{m} + \beta_n}\\ \label{app:taylor_h}
	&+\mathcal{O}\left(\frac{1}{L}\right)+\mathcal{O}\left(\frac{1}{N}\right), \text{as $L,N \rightarrow \infty$}.
\end{align}
After that, we obtain (\ref{eq:theorem_z2}) and (\ref{eq:theorem_r2}). By substituting (\ref{app:taylor_h}) into (\ref{eq:new_z2}), we have:
\begin{align} 	\nonumber 
	\tilde{z}^{(t+1)}_{f_{l,m}} =& y_{l,m} - \sum_{n \in \mathcal{N}}\sum_{p \in \mathcal{P}} A_{l,n,p}\epsilon^{(t)}_{ x_{n,p,m} \rightarrow f_{l,m}} \hat{h}^{(t+1)}_{n,p,m} \\ \nonumber
	&+ z^{(t)}_{f_{l,m}}\sum_{n \in \mathcal{N}}\sum_{p \in \mathcal{P}} \frac{|A_{l,n,p}|^2\epsilon^{(t)}_{ x_{n,p,m} \rightarrow f_{l,m}}\beta_n}{\hat{\tau}^{(t)}_{m} + \beta_n} \\ \nonumber
	&+ \mathcal{O}(\frac{1}{L})\sum_{n \in \mathcal{N}}\sum_{p \in \mathcal{P}} A_{l,n,p}+ \mathcal{O}(\frac{1}{N})\sum_{n \in \mathcal{N}}\sum_{p \in \mathcal{P}} A_{l,n,p}\\ \nonumber
	\overset{(c)}{=}& y_{l,m} - \sum_{n \in \mathcal{N}}\sum_{p \in \mathcal{P}} A_{l,n,p} \epsilon^{(t)}_{ x_{n,p,m} \rightarrow f_{l,m}} \hat{h}^{(t+1)}_{n,p,m} \\ \nonumber
	&+ \frac{1}{L}   \tilde{z}^{(t)}_{f_{l,m}}\sum_{n \in \mathcal{N}}\sum_{p \in \mathcal{P}}   \frac{\epsilon^{(t)}_{ x_{n,p,m} \rightarrow f_{l,m}}\beta_n}{\hat{\tau}^{(t)}_{m} + \beta_n} + \mathcal{O}(\frac{N}{L^{\frac{3}{2}}}) \\ \nonumber
	& + \mathcal{O}(\frac{1}{N}), \text{as $L,N \rightarrow \infty$},
\end{align}
where $(c)$ is due to (\ref{eq:appro_A_h}) and $A_{l,n,p} = \mathcal{O}\left(\frac{1}{\sqrt{L}}\right)$ \cite{SchniterAMP}. Thus, we have (\ref{eq:theorem_z2}).
By substituting (\ref{app:taylor_h}) into (\ref{app:new_r2}), we have:
\begin{align} \nonumber
	\hat{r}^{(t)}_{x_{n,p,m}} =& \sum_{l \in \mathcal{L}}A^{*}_{l,n,p}z^{(t)}_{f_{l,m}}  + \epsilon^{(t)}_{ x_{n,p,m} \rightarrow f_{l,m}}\hat{h}^{(t)}_{ n,p,m}\sum_{l \in \mathcal{L}}|A_{l,n,p}|^2 \\ \nonumber
	&+ \sum_{l \in \mathcal{L}}\frac{|A_{l,n,p}|^2 A^{*}_{l,n,p}z^{(t)}_{f_{l,m}}\beta_n\epsilon^{(t)}_{ x_{n,p,m} \rightarrow f_{l,m}}}{\hat{\tau}^{(t)}_{m} + \beta_n} \\ \nonumber
	&+ \mathcal{O}(\frac{1}{L}) \sum_{l \in \mathcal{L}}|A_{l,n,p}|^2\epsilon^{(t)}_{ x_{n,p,m} \rightarrow f_{l,m}}\\ \nonumber
	& + \mathcal{O}(\frac{1}{N}) \sum_{l \in \mathcal{L}}|A_{l,n,p}|^2\epsilon^{(t)}_{ x_{n,p,m} \rightarrow f_{l,m}}\\ \nonumber
	\overset{(d)}{=}& \sum_{l \in \mathcal{L}}A^{*}_{l,n,p}z^{(t)}_{f_{l,m}}  + \epsilon^{(t)}_{ x_{n,p,m} \rightarrow f_{l,m}}\hat{h}^{(t)}_{ n,p,m} + \mathcal{O}(\frac{1}{\sqrt{L}})\\ \nonumber
	&+ \mathcal{O}(\frac{1}{L})  + \mathcal{O}(\frac{1}{N}) , \text{as $L,N \rightarrow \infty$},
\end{align}
where $(d)$ is due to that $\sum_{l \in \mathcal{L}}|A_{l,n,p}|^2 = 1$, $\frac{\beta_n}{\hat{\tau}^{(t)}_{m} + \beta_n} = \mathcal{O}(1)$,  $z^{(t)}_{f_{l,m}}= \mathcal{O}(1)$, $\epsilon^{(t)}_{ x_{n,p,m} \rightarrow f_{l,m}} \in [0,1]$, $A_{l,n,p} = \mathcal{O}\left(\frac{1}{\sqrt{L}}\right)$, and $|A_{l,n,p}|^2 = \mathcal{O}\left(\frac{1}{L}\right)$ \cite{SchniterAMP}.
Thus, we have (\ref{eq:theorem_r2}). Therefore, we complete the proof.
}

\bibliographystyle{IEEEtran}

\newpage

\vspace{11pt}

\vfill

\end{document}